\documentclass[11pt]{article}
\usepackage{amsfonts,amsmath,amssymb,boxedminipage,color,url,fullpage}
\usepackage{enumerate,paralist}
\usepackage[numbers]{natbib} 
\usepackage[pdfstartview=FitH,colorlinks,linkcolor=blue,filecolor=blue,citecolor=blue,urlcolor=blue]{hyperref} \usepackage[labelfont=bf]{caption}
\usepackage{amsthm} 
\usepackage{aliascnt}
\usepackage{cleveref}
\usepackage{xspace}

\providecommand{\remove}[1]{}
\newcommand{\Draft}[1]{\ifdefined\IsDraft\texttt{ #1} \fi}

\ifdefined\IsDraft
    \newcommand{\authnote}[2]{{\bf [{\color{red} #1's Note:} {\color{blue} #2}]}}
\else
    \newcommand{\authnote}[2]{}
\fi

\newcommand{\mparagraph}[1]{\paragraph{#1.}}

\newcommand{\sdotfill}{\textcolor[rgb]{0.8,0.8,0.8}{\dotfill}} 

\newenvironment{protocol}{\begin{proto}}{\end{proto}}

\newenvironment{algorithm}{\begin{algo}}{\vspace{-\topsep}\end{algo}}
\newenvironment{experiment}{\begin{expr}}{\vspace{-\topsep}\sdotfill\end{expr}}

\newcommand{\aka} {also known as\ }
\newcommand{\resp}{resp.,\ }
\newcommand{\ie}  {i.e.,\ }
\newcommand{\eg}  {e.g.,\ }

\newcommand{\wrt} {with respect to\ }
\newcommand{\wlg} {without loss of generality\ }
\newcommand{\cf}{{cf.,\ }}

\newcommand{\abs}[1]{\left\lvert #1 \right\rvert}
\newcommand{\ceil}[1]{\left\lceil #1 \right\rceil}

\newcommand{\set}[1]{\ens{#1}}

\newcommand{\floor}[1]{\left \lfloor#1 \right \rfloor}

\newcommand{\eqdef}{:=}

\newcommand{\R}{{\mathbb R}}
\newcommand{\N}{{\mathbb{N}}}
\newcommand{\Z}{{\mathbb Z}}

\newcommand{\zo}{\set{0,1}}
\newcommand{\oo}{\set{-1,1}}

\newcommand{\xor}{\oplus}

\newcommand{\eps}{\varepsilon}

\newcommand{\la}{\gets}

\newcommand{\poly}{\operatorname{poly}}

\newcommand{\Exp}{\operatorname*{E}}

\newcommand{\negl}{\operatorname{neg}}

\newcommand{\Supp}{\operatorname{Supp}}

\newcommand{\MathAlg}[1]{\mathsf{#1}}

\newcommand{\sign}{\MathAlg{sign}}

\renewcommand{\cref}{\Cref}

\newtheorem{theorem}{Theorem}[section]

\newaliascnt{lemma}{theorem}
\newtheorem{lemma}[lemma]{Lemma}
\aliascntresetthe{lemma}
\crefname{lemma}{Lemma}{Lemmas}

\newaliascnt{claim}{theorem}
\newtheorem{claim}[claim]{Claim}
\aliascntresetthe{claim}
\crefname{claim}{Claim}{Claims}

\newaliascnt{corollary}{theorem}

\aliascntresetthe{corollary}
\crefname{corollary}{Corollary}{Corollaries}

\newaliascnt{construction}{theorem}

\aliascntresetthe{construction}
\crefname{construction}{Construction}{Constructions}

\newaliascnt{fact}{theorem}
\newtheorem{fact}[fact]{Fact}
\aliascntresetthe{fact}
\crefname{fact}{Fact}{Facts}

\newaliascnt{proposition}{theorem}
\newtheorem{proposition}[proposition]{Proposition}
\aliascntresetthe{proposition}
\crefname{proposition}{Proposition}{Propositions}

\newaliascnt{conjecture}{theorem}

\aliascntresetthe{conjecture}
\crefname{conjecture}{Conjecture}{Conjectures}

\newaliascnt{definition}{theorem}
\newtheorem{definition}[definition]{Definition}
\aliascntresetthe{definition}
\crefname{definition}{Definition}{Definitions}

\newaliascnt{remark}{theorem}

\aliascntresetthe{remark}
\crefname{remark}{Remark}{Remarks}

\newaliascnt{notation}{theorem}

\aliascntresetthe{notation}
\crefname{notation}{Notation}{Notation}

\newaliascnt{proto}{theorem}

\newtheorem{proto}[proto]{Protocol}

\aliascntresetthe{proto}
\crefname{proto}{protocol}{protocols}

\newaliascnt{algo}{theorem}
\newtheorem{algo}[algo]{Algorithm}
\aliascntresetthe{algo}
\crefname{algo}{algorithm}{algorithms}

\newaliascnt{expr}{theorem}
\newtheorem{expr}[expr]{Experiment}
\aliascntresetthe{expr}
\crefname{experiment}{experiment}{experiments}

\newcommand{\Stepref}[1]{Step~\ref{#1}}

\def\FullBox{$\Box$}
\def\qed{\ifmmode\qquad\FullBox\else{\unskip\nobreak\hfil
\penalty50\hskip1em\null\nobreak\hfil\FullBox
\parfillskip=0pt\finalhyphendemerits=0\endgraf}\fi}

\def\qedsketch{\ifmmode\Box\else{\unskip\nobreak\hfil
\penalty50\hskip1em\null\nobreak\hfil$\Box$
\parfillskip=0pt\finalhyphendemerits=0\endgraf}\fi}

\newcommand{\Ex}{{\mathrm E}}
\newcommand{\ex}[2]{\Exp_{#1}\left[#2\right]}
\newcommand{\eex}[1]{\ex{}{#1}}

\renewcommand{\Pr}{{\mathrm {Pr}}}

\newcommand{\ppr}[2]{\Pr_{#1}\left[#2\right]}
\newcommand{\pr}[1]{\ppr{}{#1}}

\newcommand{\Ac}{\mathsf{A}}
\newcommand{\Bc}{\mathsf{B}}
\newcommand{\Cc}{\mathsf{C}}
\newcommand{\Pc}{{\mathsf{P}}}
\newcommand{\Af}{{\mathbb{A}}}

\newcommand{\cC}{{\mathcal{C}}}

\newcommand{\ens}[1]{\{#1\}}
\newcommand{\size}[1]{\left|#1\right|}

\newcommand{\out}{{\operatorname{out}}}

\newcommand{\Uni}{{\mathord{\mathcal{U}}}}

\newcommand{\prob}[1]{\mathsf{\textsc{#1}}}

\newcommand{\SD}{\prob{SD}}

\newcommand{\cW}{{\cal{W}}}

\newcommand{\pptm}{{\sc pptm}\xspace}
\newcommand{\ppt}{{\sc ppt}\xspace}

\newcommand{\cH}{{\cal{H}}}

\newcommand{\cs}{{\cal{S}}}

\newcommand{\cY}{{\cal{Y}}}
\newcommand{\cX}{{\cal{X}}}

\newcommand{\cI}{{\cal{I}}}

\newcommand{\Tableofcontents}{
\thispagestyle{empty}
\pagenumbering{gobble}
\clearpage
\tableofcontents
\thispagestyle{empty}
\clearpage
\pagenumbering{arabic}
}

\newcommand{\vect}[1]{{ \bf #1}}
\newcommand{\Mat}[1]{{ \bf #1}}
\newcommand{\zot}{\set{0,1,2}}

\newcommand{\OT}{\ensuremath{\operatorname{OT}}\xspace}

\newcommand{\CoinToss}{\ensuremath{\MathAlg{CoinFlip}}\xspace}

\newcommand{\Aadv}{{\Ac}}
\newcommand{\iAadv}{{\Af}}

\newcommand{\Abort}{\ensuremath{\MathAlg{Abort}}\xspace}
\newcommand{\Ideal}{\operatorname{IDEAL}}
\newcommand{\Real}{\operatorname{REAL}}

\newcommand{\BerooSign}{{\mathcal{C}}}
\newcommand{\Beroo}[1]{{\BerooSign_{#1}}}
\newcommand{\vBeroo}[1]{{\widehat{\BerooSign}_{#1}}}

\newcommand{\sBias}[2]{\widehat{\BerooSign}^{-1}_{#1}(#2)}

\newcommand{\BerzoSign}{{\mathcal{B}er}}
\newcommand{\Berzo}[1]{\BerzoSign(#1)}

\newcommand{\sss}{{\mathsf{sum}}}
\newcommand{\xs}[2]{{\sss_{#1}(#2)}}
\newcommand{\xl}[2]{{\ell_{#1}(#2)}}

\newcommand{\ms}[1]{\xs{\rnd}{#1}}
\newcommand{\ml}[1]{\xl{\rnd}{#1}}
\newcommand{\NS}[1]{{\xs{n}{#1}}}
\newcommand{\NL}[1]{\xl{n}{#1}}

\newcommand{\HypSign}{\mathcal{HG}}
\newcommand{\Hyp}[1]{{\HypSign_{#1}}}
\newcommand{\vHyp}[1]{{\widehat{\HypSign}_{#1}}}

\newcommand{\w}{w}

\newcommand{\ist}{{i^\ast}}

\newcommand{\ShareGen}{\MathAlg{SharesGen}}
\newcommand{\TwoShareGen}{\MathAlg{TwoPartySharesGen}}
\newcommand{\HidTwoShareGen}{\MathAlg{HidTwoPartySharesGen}}

\newcommand{\ThreeShareGen}{\MathAlg{ThreePartySharesGen}}

\newcommand{\bias}{\MathAlg{Bias}}
\newcommand{\val}{\MathAlg{val}}
\newcommand{\eo}{\MathAlg{o}}
\newcommand{\veo}{O}

\newcommand{\ShVecT}[3]{\vect{#1^{#2,\# #3}}}
\newcommand{\ShVecO}[2]{\vect{#1^{\# #2}}}

\newcommand{\defn}[2]{{d^{#1}_{#2}}}

\newcommand{\sVS}[1]{\ShVecO{s}{#1}}
\newcommand{\SSS}[1]{\Mat{S^{#1}}}

\newcommand{\cVS}[1]{\ShVecO{c}{#1}}
\newcommand{\dVS}[2]{\ShVecT{d}{#1}{#2}}
\newcommand{\tp}[1]{\ceil{\log #1}}
\newcommand{\Sh}{S}
\newcommand{\sh}{s}

\newcommand{\advPartNum}{\ell}
\newcommand{\partNum}{t}

\newcommand{\secParam}{\kappa}

\newcommand{\rnd}{m}

\newcommand{\game}{\mathsf{G}}

\newcommand{\ratioo}{\MathAlg{ratio}}
\newcommand{\error}{\MathAlg{error}}
\newcommand{\ofs}{t}
\newcommand{\trnd}{{\widetilde{\rnd}}}

\newcommand{\z}{z}
\newcommand{\inn}{\mathsf{inner}}
\newcommand{\outt}{\mathsf{outer}}

\newcommand{\Pitwo}{\Pi^2}
\newcommand{\Pithree}{\Pi^3}
\newcommand{\Ptwo}{\Pc^2}
\newcommand{\Pthree}{\Pc^3}

\newcommand{\Ph}{{\widehat{\Pc}}}
\newcommand{\Pih}{{\widehat{\Pi}}}

\newcommand{\PifTwo}{\Pih^2}
\newcommand{\PifThree}{\Pih^3}
\newcommand{\PhTwo}{\Ph^2}
\newcommand{\PhThree}{\Ph^3}

\newcommand{\PiffTwo}{\widetilde{\Pi^2}}

\newcommand{\Pif}{{\widehat{\Pi}}}
\newcommand{\ovr}{{\textbf{r}}}

\newcommand{\aux}{H}
\newcommand{\aset}{\cH}
\newcommand{\paux}{h}
\newcommand{\const}{\lambda}
\newcommand{\xconst}{\xi}
\newcommand{\vct}{v}

\newcommand{\gameVars}{\set{X_i,Y_i,\aux_i,\veo_i,\veo_i^-}}

\newcommand{\orac}{\tilde{o}}

\newcommand{\MNS}{\mathsf{MNS}}

\title{An Almost-Optimally Fair Three-Party Coin-Flipping Protocol\thanks{The full version was published in the SIAM Journal on Computing 2017 \cite{HaitnerT17}. An extended abstract of this work appeared in the Annual Symposium on the Theory of Computing 2014 \cite{HaitnerT14}. \cite{HaitnerT14} claims the existence of an $\frac{O(\log^2 \rnd)}\rnd$-bias protocol, compared to $\frac{O(\log^3 \rnd)}\rnd$-bias stated here, but the proof for this bias was flawed.}
\remove{\Draft{\\{\small \sc Working Draft: Please Do Not Distribute}}}}
\author{Iftach Haitner\thanks{School of Computer Science, Tel Aviv University. E-mail:
 \texttt{iftachh@cs.tau.ac.il,eliadtsf@tau.ac.il}. Research supported by ERC starting grant 638121, ISF grant 1076/11, the Israeli Centers of Research Excellence (I-CORE) program (Center  No. 4/11), US-Israel BSF grant 2010196  and Check Point Institute for Information Security.} \and Eliad Tsfadia$^{\dagger}$}

\begin{document}
\sloppy
\maketitle

\begin{abstract}
In a multiparty \emph{fair} coin-flipping protocol, the parties output a common (close to) unbiased  bit, even when some corrupted parties try to bias the output. \citeauthor{Cleve86} [STOC 1986] has shown that in the case of dishonest majority (\ie  at least half
of the parties can be corrupted), in \emph{any} $\rnd$-round coin-flipping protocol the corrupted parties can bias the honest parties' common output bit by $\Omega(\frac 1{\rnd})$. For more than two decades the best known coin-flipping protocols against dishonest majority had bias $\Theta(\frac {\advPartNum}{\sqrt{\rnd}})$, where $\advPartNum$ is the number of corrupted parties. This was changed by a recent breakthrough result of \citeauthor{MoranNS09} [TCC 2009], who constructed an $\rnd$-round, \emph{two}-party coin-flipping protocol with optimal bias  $\Theta(\frac 1 \rnd)$. In a subsequent work, \citeauthor{BeimelOO10} [Crypto 2010] extended this result to the multiparty case in which \emph{less than} $\frac23$ of the parties can be corrupted. Still for the case of $\frac23$ (or more) corrupted parties, the best known protocol had bias $\Theta(\frac {\advPartNum}{\sqrt{\rnd}})$. In particular, this was the state of affairs for the natural three-party case.

We make a step towards eliminating the above gap, presenting an $\rnd$-round,  three-party coin-flipping protocol, with bias $\frac{O(\log^3 \rnd)}\rnd$. Our approach (which we also apply for the two-party case) does not follow the ``threshold round"  paradigm used in the work of  \citeauthor{MoranNS09} and \citeauthor{BeimelOO10}, but rather is a variation of the majority protocol of \citeauthor{Cleve86}, used to obtain the aforementioned  $\Theta(\frac {\advPartNum}{\sqrt{\rnd}})$-bias protocol.
\end{abstract}


\Tableofcontents

\section{Introduction}
In a multi-party \emph{fair} coin-flipping (-tossing) protocol, the parties output a common (close to) unbiased bit, even though some corrupted parties try to bias the output. More formally, such protocols should satisfy the following two properties: first, when all parties are honest (\ie follow the prescribed protocol), they all output the \emph{same} bit, and this bit is unbiased (\ie uniform over $\zo$). Second, even when some parties are corrupted (\ie collude and arbitrarily deviate from the protocol), the remaining parties should still output the \emph{same} bit, and this bit should not be too biased (\ie its distribution should be close to uniform over $\zo$). We emphasize that, unlike weaker variants of coin-flipping protocol known in the literature, the honest parties should output a common bit, regardless of what the corrupted parties do. In particular, they are not allowed to abort if a cheat was noticed.

When a majority of the parties are honest, efficient and \emph{completely} fair coin-flipping protocols are known as a special case of secure multi-party computation with an honest majority \cite{BenGolWig88}.\footnote{Throughout, we assume a broadcast channel is available to the parties. By \cite{CohenHOR2016}, broadcast channel is necessary for fair coin-flipping protocol secure against third, or more, corruptions.} When an honest majority is not guaranteed, however, the situation is more complex.

\mparagraph{Negative results}
\citet{Cleve86} showed that for \emph{any} efficient two-party $\rnd$-round  coin-flipping protocol, there exists an efficient adversary to bias the output of the honest party by $\Theta(1/\rnd)$, and that the  lower bound extends to the multi-party case via a simple reduction.

\mparagraph{Positive results}
Assuming one-way functions exist, \citet{Cleve86} showed that a simple $\rnd$-round majority protocol can be used to derive a $\partNum$-party coin-flipping protocol with bias $\Theta(\frac{\advPartNum}{\sqrt{\rnd}})$ (against dishonest majority), where $\advPartNum$ is the number of corrupted parties. For more than two decades, \citeauthor{Cleve86}'s  protocol was the best known fair coin-flipping protocol (without honest majority), under \emph{any} hardness assumption, and for \emph{any} number of parties. In a recent breakthrough result, \citet{MoranNS09} constructed an $\rnd$-round, \emph{two}-party coin-flipping protocol with optimal bias of $\Theta(\frac 1 \rnd)$. The result holds for any parameter $\rnd \in \N$, and under the assumption that oblivious transfer protocols exist. In a subsequent work, \citet{BeimelOO10} extended the result of \cite{MoranNS09} for the multi-party case in which \emph{less than} $\frac23$ of the parties can be corrupted. More specifically, for any $\advPartNum < \frac23 \cdot \partNum$, they presented an $\rnd$-round, $\partNum$-party protocol, with bias $\frac{2^{2\advPartNum -\partNum}}{\rnd}$ against (up to) $\advPartNum$ corrupted parties.

Still for the case of $\frac23$ (or more) corrupted parties, the best known protocol was the $\Theta(\frac{\advPartNum}{\sqrt{\rnd}})$-bias majority protocol of \cite{Cleve86}. In particular, this was the state of affairs for the natural three-party case (where two parties are corrupt).

\subsection{Our Result}\label{sec:into:ourResult}
We present an almost-optimally fair, three-party coin-flipping protocol.
\begin{theorem}[main theorem, informal]\label{thm:mainLInf}
Assuming the existence of oblivious transfer protocols,\footnote{It is enough to assume the existence of a constant-round secure-with-abort protocol, which is a weaker assumption. However, we chose to assume OT for simplifying the main theorem.} then for any $\rnd\in \N$ there exists an $\rnd$-round, three-party coin-flipping protocol, with bias $\frac{O(\log^3 \rnd)}\rnd$ (against one, or two, corrupted parties).
\end{theorem}
That is, no efficient algorithm can makes the (expected) outcome of the protocol to deviate from $\frac12$ by more than $\frac{O(\log^3 \rnd)}\rnd$. As a building block towards constructing our three-party protocol, we present an alternative construction for two-party, almost-optimally fair coin-flipping protocols. Our approach does not follow the ``threshold round" paradigm used in \cite{MoranNS09,BeimelOO10}, but rather is a variation of the aforementioned $\Theta(\frac{\advPartNum}{\sqrt{\rnd}})$-bias, coin-flipping protocol of \cite{Cleve86}.

\subsection{Additional Related Work}\label{sec:relatedWork}
\citet{CleveI93} showed that in the \textit{fail-stop model},\footnote{In this model, the parties are assumed to have unbounded computation power, cannot deviate from  the prescribed protocol, but are allowed to prematurely abort their execution.} any
two-party $\rnd$-round coin-flipping protocol has bias $\Omega(\frac 1{\sqrt{\rnd}})$; adversaries in this model are
computationally unbounded, but they must follow the instructions of the protocol, except for being
allowed to abort prematurely. \citet{DachmanLMM11} showed that the same holds for $o(n/\log n)$-round protocols in the random-oracle model ---  the parties have oracle access to a uniformly  chosen  function over $n$ bit strings.

There is a vast literature concerning coin-flipping protocols with weaker security guarantees. Most notable among these are protocols that are \textit{secure with abort}. According to this security definition, if a cheat is detected or if one of the parties aborts, the remaining parties are not required to output anything. This form of security is meaningful in many settings, and it is typically much easier to achieve; assuming one-way functions exist, secure-with-abort protocols of negligible bias are known to exist against any number of corrupted parties \cite{Blum83,HNORV08,Naor91}. To a large extent, one-way functions are also necessary for such coin-flipping protocols \cite{BermanHT14,HaitnerOmri11,ImpagliazzoLu89,Maji10}.

Coin-flipping protocols were also studied in a variety of other models. Among these are collective
coin-flipping in the \textit{perfect information model}: parties are computationally unbounded
and all communication is public \cite{AlonNaor93,BenLin89,Feige99a,RussellZuc98,Saks89}, and protocols based on physical assumptions, such
as quantum computation \cite{AharonovEtAl00,Ambainis04,AmbainisBDR04} and tamper-evident seals \cite{MoranNaor05}.

Perfectly fair coin-flipping protocols (\ie zero bias) are a special case of protocols for \emph{fair} secure function evaluation (SFE). Intuitively, the security of such protocols guarantees that when the protocol terminates, either everyone receives the (correct) output of the functionality, or no one does. While \citet{Cleve86}'s result yields that some functions do not have fair SFE, it was recently shown that many interesting function families do have (perfectly) fair SFE \cite{gordonHKL11,Ash14,ABMO15}.

\subsection{Our Techniques}\label{sec:Technique}
The following is a high-level description of the ideas underlying our three-party fair coin flipping protocol.\footnote{We restrict the discussion to the intuitive  game-base definition of fairness --- the goal of the adversary is to make the honest party to output some bit $b$ with probability as further away from $\frac12$ as possible. Discission of the more standard Real/Ideal definition of fairness, in which we prove our result, is given in  \cref{sec:realIdeal}.} We start by describing the two-party protocol of \citet{MoranNS09} (hereafter the $\MNS$ protocol), and explain why natural extensions of their approach (such as the one used in \cite{BeimelOO10}) fall short when it comes to constructing three-party fair protocols. We next explain our new approach for two-party protocols, and then extend this approach to three parties.

Throughout, we assume \wlg that if a corrupted party aborts in a given round, it sends an abort message to all other parties at the \emph{end} of this round (after seeing the messages sent by the non-aborting parties). To keep the discussion simple, we focus on security against polynomially bounded fail-stop adversaries --- ones that follow the prescribed protocol, but might abort prematurely. Achieving this level of security is the heart of the matter, since (assuming one-way functions exist) there exists a round-preserving reduction from protocols secure against fail-stop adversaries into protocols of full-fledged security \cite{GoldreichMiWi91}.\footnote{Note that by restricting the parties to being fail-stop, we do not reduce the setting to the fail-stop model, since the parties considered here are computationally bounded.}

\subsubsection{The Two-Party $\MNS$ Protocol}
For $\rnd\in \N$, the $(2\rnd)$-round, two-party $\MNS$ protocol $(\Pc_0,\Pc_1)$ is defined as follows.\footnote{The protocol described below is a close variant of the original $\MNS$ protocol, which serves our presentation better. The difference is the addition of phase $(b)$, in both the share generating function and the protocol, which does not exists in the original protocol.} Following a common paradigm for fair multi-party computations \cite{BeimelOO10,GordonHKL08,Katz07}, the protocol starts by the two parties using oblivious transfer (\OT) to securely compute the following ``share generating" random function.

\begin{algorithm}[share generating function $\ShareGen$]\label{func:mns}~
\item[Input:] Round parameter $1^\rnd$.
\item[Operation:]~
\begin{enumerate}
 \item Uniformly sample $c\la \zo$ and $\ist \la [\rnd] \ (= \set{1,\ldots,m})$.

 \item For $i=1 $ to $\rnd$, let

\begin{enumerate}

\item $(\defn{0}{i},\defn{1}{i})\! = \!\left\{
 \begin{array}{ll}
 \!\mbox{\emph{uniform} sample from $\zo^2$}, & i < \ist -1 \\
 \!(c,c), & \hbox{otherwise.}
 \end{array}
 \right.$\label{step:func:mns:defn}
\item $c_i = \left\{
 \begin{array}{ll}
 \perp, & i < \ist \\
 c, & \hbox{otherwise.}
 \end{array}
 \right.$
\end{enumerate}

 \item Split each of the $3\rnd$ values $\defn{0}{1},\defn{1}{1},\ldots,\defn{0}{\rnd},\defn{1}{\rnd},c_1,\dots,c_\rnd$ into two ``shares," using a 2-out-of-2 secret sharing scheme, and output the two sets of shares.

\end{enumerate}
\end{algorithm}

\begin{protocol}[$(\Pc_0,\Pc_1)$]~
\begin{enumerate}
\item[Common input:] round parameter $1^\rnd$.

\item[Initial step:] The parties securely compute $\ShareGen(1^\rnd)$, where each party gets one set of shares.

\item[Main loop:] For $i=1$ to $\rnd$, do
\begin{enumerate}
 \item $\Pc_0$ sends to $\Pc_1$ its share of $\defn{1}{i}$, and $\Pc_1$ sends to $\Pc_0$ its share of $\defn{0}{i}$.
 \item[$\bullet$] $\Pc_0$ reconstructs the value of $\defn{0}{i}$, and $\Pc_1$ reconstructs the value of $\defn{1}{i}$.
 \item Each party sends to the other party its share of $c_i$.

 \item[$\bullet$] Both parties reconstruct the value of $c_i$.
\end{enumerate}

\item[Output:] The parties output $c_i$, for the first $i$ for which $c_i \neq \perp$.

\item[Abort:] If $\Pc_0$ aborts, party $\Pc_1$ outputs the value of $\defn{1}{i}$ for the \emph{maximal} $i\in [\rnd]$ for which it has reconstructed this value. If there is no such $i$, $\Pc_1$ outputs a uniform bit. (The case that $\Pc_1$ aborts is analogously defined).
\end{enumerate}
\end{protocol}
We start with few observations regarding the secure computation of $\ShareGen(1^\rnd)$ done in the above protocol.
\begin{itemize}
\item The computation of $\ShareGen(1^\rnd)$ is \emph{not} fair: the parties get their parts of the output (\ie their shares) in an \emph{arbitrary} manner. Specifically, the corrupted party might prematurely abort after learning its part of the output, preventing the other party from getting its part.
\item Since $\ShareGen(1^\rnd)$ is efficient, assuming \OT protocols exist, an (unfair) secure computation of $\ShareGen(1^\rnd)$ exists.

 \item Ignoring negligible terms (due to the imperfection of secure computation using \OT), the output of each party (when seen on its own) is a set of uniform strings. In particular, it contains \emph{no information} about the other party's shares, or about the values of $c$ and $\ist$.

 By construction, a party outputs a uniform bit if the other party aborts before the end of the secure computation phase. Hence, it makes no sense for a party to abort during this phase.

 \item Given the above observation, it is instructive to pretend that at the first step of the protocol, the output of a random execution of $\ShareGen(1^\rnd)$ was given to the parties by an \textit{honest dealer}.
\end{itemize}

Note that in each round of the above protocol, both honest parties send their messages without waiting for the other party's message. Hence, the above protocol is symmetric \wrt the parties' role. However, since we assume no simultaneous channel (which would have trivialized the whole question), the corrupted party can postpone sending its message until it gets the message of the honest party, and then decide whether to send its message for this round or abort.

\mparagraph{Security of the protocol}At least on the intuitive level, the security proof of the above protocol is rather simple. Since the protocol is symmetric, we assume for concreteness that $\Pc_0$ is corrupted and tries to bias the expected output of $\Pc_1$ away from $\frac12$. The following random variables are defined \wrt a random execution of $(\Pc_0,\Pc_1)$: let $V$ be the view of the corrupted $\Pc_0$, right after sending the abort message, and let $V^-$ be the view $V$ without this abort message. For a view $v$, let $\val(v)$, the view \emph{value}, be defined as the expected outcome of $\Pc_1$ conditioned that $\Pc_0$'s view is $v$, assuming no further aborts (\ie if $\Pc_0$ is not aborting in $v$, then it acts honestly till the end of the protocol). It is not hard to verify that the bias obtained by $\Pc_0$ (toward 0 or 1) is \emph{exactly} the expected value of $\size{\val(V) - \val(V^-)}$.\footnote{The expected value of $\size{\val(V) - \val(V^-)}$ actually captures  the security of the protocol in a stronger sense, characterizing the so-called $\alpha$-security of the protocol according to the Real/Ideal paradigm. See proof in \cref{sec:CFprotocols}.}

It is also easy to see that by aborting in round $(i,b)$ (\ie phase $(b)$ of round $i$), for some $i\in [\rnd]$, party $\Pc_0$ gains nothing (\ie $\val(V) = \val(V^-)$), where the $(i,j)$'th round of the execution stands for the $j$'th step of the $i$'th loop in the execution.  A slightly more complicated math yields that by aborting in round $(i,a)$,  party $\Pc_0$ only gains $\Theta(\frac1{\rnd})$ bias. It follows that the maximal bias obtained by a fail-stop strategy for $\Pc_0$ is $\Theta(\frac1{\rnd})$.

\mparagraph{Fairness via defense}
Let us present a different view of the $\MNS$ protocol. Consider a variant of this protocol without the $d_i$'s. Namely, the parties reconstruct $c_1,\ldots,c_\rnd$ one at a time, until they reach $c_i\neq \perp$. When an abort occurs, the remaining party outputs an unbiased coin if it has not yet reconstructed $c$, and outputs $c$ otherwise. It is easy to see that an aborting attacker can bias the output of the other party in this degenerate variant by $\frac14$; that is, it simply waits until it reconstructs $c$ and then aborts for biasing the other party's output towards $1-c$.

The role of  ``defense" values $(\defn{0}{i},\defn{1}{i}),\ldots,(\defn{0}{\rnd},\defn{1}{\rnd})$ is to prevent such an attack; if a party aborts after reconstructing $c$, the other party is guaranteed to output $c$ as well. The problem is, however, that the defense values themselves might cause a problem: a corrupted party might abort after reconstructing its defense value for the $i$'th round (and not only after reconstructing $c_i$). Indeed, by aborting in these rounds, a corrupted party does gain a bias, but only $\Theta(\frac 1 \rnd)$.

\mparagraph{On extending the $\MNS$ protocol for the three-party case}
We next explain why the approach of $\MNS$ does not seem to be useful for constructing fair, three-party coin-flipping protocols.

In a three-party fair coin-flipping protocol, one should deal with two, possibly non-simultaneous, aborts. In particular, after $\Pc_0$ aborts, the remaining pair of parties $\set{\Pc_1,\Pc_2}$ should interact in a two-party protocol to agree on their common coin. Since one of $\set{\Pc_1,\Pc_2}$  might be also corrupted, this two-party protocol needs to be a fair coin-flipping  protocol as well. Assuming $\Pc_0$ aborts in round $i$, the expected outcome of this  two-party protocol  should be equal (up to an additive difference of $\Theta(\frac 1 \rnd)$) to  the  expected outcome of the three party protocol assuming no aborts, from $\Pc_0$'s point of view, \emph{after} getting the $i$'th round messages. Otherwise, $\Pc_0$, by aborting after seeing the messages sent by $\set{\Pc_1,\Pc_2}$ in this round, can bias the outcome of the other parties  by more than  $\Theta(\frac 1 \rnd)$. 

Consider the following natural extension of $\MNS$ protocol to a three-party protocol. The value of $c_1,\ldots, c_\rnd$ are as in the two-party protocol (now shared between the three parties). The defense values are not bits, but rather two vectors of shares for the two remaining parties (different shares for each possible pair), to enable them to interact in some fair two-party coin-flipping protocol if the third party aborts. Assume that in the $i$'th round of the ``outer'' three-party protocol, the value of $c_i$ is one (\ie $c_i = c=1$), and consider the two-party protocol executed by the remaining parties, if a party aborts in this round. The outcome of the remaining party in the case of a premature abort in this underlying two-party protocol should be also one. Otherwise, two corrupted parties can mount the following two-phase attack: first aborting in the outer three-party protocol after seeing $c_i=1$, and then prematurely aborting in the inner two-party protocol, knowing that the other party will output something that is far from one. Now, assume that in the $i$'th round of the ``outer'' three-party protocol, the value of $c_i$ is $\perp$ (\ie $i < \ist$), and consider again the two-party protocol executed by the remaining parties if party aborts in this round. It is easy to see that  expected outcome of this two-party protocol should be close to $\frac12$ (\ie unbiased). Thus, the defense values, to be constructed by each party during the execution of this two-party protocol, cannot all be of the same value.

But the above restrictions on the two-party protocol defense values, ruin the security of the outer three-party protocol; in each round $i$, two corrupted (and thus colluding) parties can reconstruct the \emph{whole} two-party execution that they should engage in if the other (in this case, the honest) party aborts in this round. By checking whether the defense values of this two-party execution are all ones (indicating that $c =1$), all zeros (indicating that $c =0 $), or mixed (indicating that $c_i =\perp$), they get enough information for biasing the output of the protocol by a constant value.

What fails the above three-party protocol,  is that during its execution the expected outcome of the protocol given a player's  view might change  by a constant (say from $\frac12$ to $1$). As we argued above, the (long) defense values reconstructed \emph{before} each round in the three-party protocol have to contain many (\ie $\rnd$) samples drawn according to the value of the protocol at the \emph{end} of the round. It follows that two corrupted parties might extrapolate,  at the \emph{beginning} of such a round, the value of protocol when this round \emph{ends},  thus rendering the protocol insecure. 

To conclude, due to the restrictions described above which comes from the large jump in the game value of $\MNS$, extending it into three-party protocol seems unlikely.  
\subsubsection{Our Two-Party Protocol}
Given the above understanding, our first step is to construct a two-party coin-flipping protocol, whose value  only changes  \emph{slightly} (\ie smoothly) between consecutive rounds. In the next section we use a \emph{hiding} variant of such a smooth coin-flipping protocol as a building block for constructing an (almost) optimally fair three-party protocol.

Consider the   $\Theta(\frac{1}{\sqrt{\rnd}})$-bias  coin-flipping protocol of \citet{Cleve86}: in each round $i\in[\rnd]$, the parties reconstruct the value of a coin $c_i \in \oo$, and the final outcome is set to $\sign(\sum_{i\in [\rnd]} c_i)$. Since the value of $\sum c_i$ is close to being uniform over $[-\sqrt{\rnd},\sqrt{\rnd}]$, the value of the first coin $c_1$ changes the protocol's value by $\Theta(\frac1{\sqrt \rnd})$. This sounds like a good start toward achieving a smooth coin-flipping protocol.

\remove{The problem is, however, that with probability $\Theta(\frac1{\sqrt \rnd})$, the sum of $c_1,\dots, c_{\rnd-1}$ is exactly zero. Hence, with this probability, the final coin changes the protocol's value by $\frac12$.

We overcome the above problem by using a \emph{weighted} majority protocol. In the first round the parties reconstruct $\rnd$-coins (in a single shot), reconstruct $(\rnd-1)$ coins in the second round, and so on, until in the very last round only a single coin is reconstructed. Now the value of $\sum c_i$ (now each $c_i$ is an integer) is close to being uniform over $[-\rnd,\rnd]$, and the last round determines the outcome only with probability $\Theta(\frac1\rnd)$ (versus $\Theta(\frac1{\sqrt \rnd})$ in the unweighted version). Other rounds also enjoy a similar smoothness property. We emphasize that the resulting protocol \emph{does not} guarantee   small bias (but  only a $\Theta(\frac{1}{\sqrt{\rnd}})$-bias). Rather, we take advantage of its smoothness and use it as a building block of a fair two-party protocol (and then of a three-party protocol).}

\mparagraph{The protocol} As in MNS protocol, the parties start by securely computing a share generating function, and then use its outputs to slowly reconstruct the output of the protocol.

Let $\Berzo{\delta}$ be the Bernoulli distribution over $\zo$, taking the value one with probability $\delta$ and zero otherwise.

We next describe (a simplified variant of) our share generating function and  use it to describe our two-party protocol.
\begin{algorithm}[share generating function $\TwoShareGen$]~
\item[Input:] Round parameter $1^\rnd$.
\item[Operation:]~
\begin{enumerate}

 \item For $\z\in \zo$,  sample $\defn{\z}{0} \la \zo$.

 \item For $i=1 $ to $\rnd$,

 \begin{enumerate}
 \item Sample $c_i \la \oo$.
 \item For $\z\in \zo$, sample $\defn{\z}{i} \la\Berzo{\delta_i}$, for $\delta_i = \pr{\sum_{j=1}^\rnd c_j \geq 0 \mid c_1, \ldots, c_i}$.\footnote{$\delta_i$ is the probability that the protocol's output is one, given the value of the ``coins" $c_1\ldots, c_i$ (and assuming no abort).}
 \end{enumerate}

\item Split each of the $3\rnd$ values $\defn{0}{1},\defn{1}{1},\ldots,\defn{0}{\rnd},\defn{1}{\rnd},c_1,\dots,c_\rnd$ into two ``shares", using a 2-out-of-2 secret sharing scheme, to create two sets of shares: $\sVS{0}$ and $\sVS{1}$.

 \item    Output $(\defn{0}{0},\sVS{0}), (\defn{1}{0},\sVS{1})$.

\end{enumerate}
\end{algorithm}

\begin{protocol}[$\pi_2 = (\Ptwo_0,\Ptwo_1)$]\label{intro:protocol:twoparty}~
\begin{enumerate}
\item[Common input:] round parameter $1^\rnd$.

\item[Initial step:] The parties securely compute $\TwoShareGen(1^\rnd)$. Let $(\defn{i}{0},\sVS{i})$ be the local output of $\Ptwo_i$.

\item[Main loop:] For $i=1$ to $\rnd$, do
\begin{enumerate}
 \item $\Ptwo_0$ sends to $\Ptwo_1$ its share of $\defn{1}{i}$, and $\Ptwo_1$ sends to $\Ptwo_0$ its share of $\defn{0}{i}$.

\item[$\bullet$] $\Ptwo_0$ reconstructs the value of $\defn{0}{i}$, and $\Ptwo_1$ reconstructs the value of $\defn{1}{i}$.

 \item Each party sends to the other party its share of $c_i$.

 \item[$\bullet$] Both parties reconstruct the value of $c_i$.

\end{enumerate}

\item[Output:] Both parties output one if $\sum_{j=1}^{\rnd} c_j \geq 0$, and zero otherwise.

\item[Abort:] If $\Ptwo_0$ aborts, party $\Ptwo_1$ outputs the value of $\defn{1}{i}$, for the \emph{maximal} $i\in [\rnd]$ for which it has reconstructed this value (note that by construction such an $i$ always exists).

 The case that $\Ptwo_1$ aborts is analogously defined.
\end{enumerate}
\end{protocol}
Namely, the parties interact in a majority protocol, where in the $i$'th round, they reconstruct, in an unfair manner, the $i$'th coin (\ie $c_i$). If a party aborts, the remaining party outputs a defense value given to it by the honest dealer (implemented via the secure computation of $\TwoShareGen$).
\remove{Namely, the parties interact in a weighted majority protocol, where in the $i$'th round, they reconstruct, in an unfair manner, $(m+1-i)$ coins (\ie $c_i$). If a party aborts, the remaining party outputs a defense value given to it by the honest dealer (implemented via the secure computation of $\TwoShareGen$).}

A few remarks are in place. First, we will only define the protocol for odd values of $\rnd$.\remove{$\rnd \equiv 1 \bmod 4$.\footnote{Defining the protocol for $\rnd \equiv 2 \bmod 4$ works as well.}} Hence, $\sum_{j=1}^{\rnd} c_j \neq 0$, and the protocol's output is a uniform bit when played by the honest parties. Second, if $\Ptwo_0$ aborts in the first round, the party $\Ptwo_1$ could simply output a uniform bit.  We make $\Ptwo_0$  output $\defn{1}{0}$, since this  be useful when the two-party protocol will be later used as part of the three-party protocol. Finally, one can define the above protocol without exposing the coins $c_i$'s to the parties (in this case, the honest parties output $(\defn{0}{\rnd},\defn{1}{\rnd})$ as the final outcome). We do expose the coins to make the analysis of the protocol easier to follow.

\mparagraph{Security of the protocol}\label{intro:protocol:security}
Note that the defense value given in round $(i,a)$ (\ie step  $a$ of the $i$'th loop) is distributed according to the expected outcome  of the protocol, conditioned on the value of the coin to \emph{be given} in round $(i,b)$. These defense values make aborting in round $(i,b)$, for any value of $i$, harmless. So it is left to argue that aborting in round $(i,a)$, for any value of $i$, is not too harmful either. Intuitively,  this holds since the defense value reconstructed in round $(i,a)$ is only a noisy signal about the value of $c_i$.

Since the protocol is symmetric, we assume for concreteness that the corrupted party is $\Ptwo_0$. Similar to the analysis of the MNS protocol sketched above, it suffices to bound the value of $\size{\val(V) - \val(V^-)}$.

Assume that $\Ptwo_0$ aborts in round $(i,b)$. By construction, $\val(V^-) =  \delta_{i}$. Since, the defense of  $\Ptwo_1$ in round $(i,b)$ is sampled according to $\Berzo{\delta_i}$, it is also the case that $\val(V) =  \delta_i$.

Assume now that $\Ptwo_0$ aborts in round $(i,a)$. By construction, $\val(V) = \delta_{i-1}$. Note that $V^-$ does contains some information about $\delta_i$, \ie a sample from $\Berzo{\delta_i}$, and thus $\val(V^-)$ is typically different from $\val(V)$. Yet, since $V^-$ contains only a sample from $\Berzo{\delta_i}$,  a noisy signal for the actual value of $\delta_i$, we manage to prove the following.
\begin{align}\label{eq:intro:gap}
\size{\val(V) - \val(V^-)} = \eex{\frac{(\delta_i - \delta_{i-1})^2}{\delta_{i-1}}\mid \delta_{i-1}}
\end{align}
If $\Ptwo_0$ aborts in the first rounds, \cref{eq:intro:gap} yields that $\size{\val(V) - \val(V^-)} = O(\frac1\rnd)$ since by the ``smoothness" of the protocol (\ie the value of the game does not change drastically between consecutive rounds) it follows that $\size{\frac{\delta_i - \delta_{i-1}}{\delta_{i-1}}} \in O(\frac1{\sqrt{\rnd}})$. 
The problem is, however, that with probability $\Theta(\frac1{\sqrt \rnd})$, the sum of $c_1,\dots, c_{\rnd-1}$ is exactly zero. Hence, with this probability, the final coin changes the protocol's value by $\frac12$. Therefore, the bias obtained by $\Ptwo_0$ that just wait to the last round to abort is $\Theta(\frac1{\sqrt \rnd})$.

We overcome the above problem by using a \emph{weighted} majority variant of the protocol. In the first round the parties reconstruct $\rnd$-coins (in a single shot), reconstruct $(\rnd-1)$ coins in the second round, and so on, until in the very last round only a single coin is reconstructed. Now the value of $\sum c_i$ (now each $c_i$ is an integer) is close to being uniform over $[-\rnd,\rnd]$, and the last round determines the outcome only with probability $\Theta(\frac1\rnd)$ (versus $\Theta(\frac1{\sqrt \rnd})$ in the unweighted version). Other rounds also enjoy a similar smoothness property. See \cref{sec:protocol} for more details.

\remove{\cref{eq:intro:gap} yields that $\size{\val(V) - \val(V^-)} = O(\frac1\rnd)$, since by the ``smoothness" of the protocol (\ie the value of the game does not change drastically between consecutive rounds) it follows that $\size{\frac{\delta_i - \delta_{i-1}}{\delta_{i-1}}} \in O(\frac1{\sqrt{\rnd}})$ with high probability.\footnote{The $O(\log^3 \rnd)$ factor in our actual result does not appear in the above informal argument.}}

\subsubsection{Our Three-Party Protocol}
We start by applying a generic approach, introduced by \citet{BeimelOO10}, to try and extend our fair two-party protocol into a three-party one. We explain why this approach falls  short, and present a variant of our two-party protocol for which the generic approach does yield the desired three-party protocol. To keep the  presentation simple, the two-party protocol we use is  the  \emph{non}-weighted variant of our two-party protocol (the actual implementation uses the aforementioned weighted protocol).

In this first attempt protocol, the three parties interact in the following variant of the two-party protocol $\pi_2$ described in \cref{intro:protocol:twoparty}. The parties start by (securely) computing $\ThreeShareGen(1^\rnd)$  defined below.

For $\delta \in [0,1]$, let $\TwoShareGen(1^\rnd, \delta)$ be the following variant of $\TwoShareGen$ defined above: (1) the ``coin" $c_i$ takes the value $1$ with probability $\frac12 + \eps$ and $-1$ otherwise (and not a uniform coin over $\oo$ as in $\TwoShareGen$), where $\eps$ is set to the number such that $\pr{\sum_{i=1}^\rnd c_i \geq 0} = \delta$; (2) the initial defense values $\defn{0}{0}$ and $\defn{1}{0}$  are sampled according to $\Berzo{\delta}$ (and not $\Berzo{\frac12}$ as in $\TwoShareGen$).

\begin{algorithm}[share generating function $\ThreeShareGen$]~
\item[Input:] Round parameter $1^\rnd$.
\item[Operation:]~
\begin{enumerate}

 \item For $i=1 $ to $\rnd$,

 \begin{enumerate}
 \item Sample $c_i \la \oo$.
 \remove{\item Let $\eps_i \in [-\frac12,\frac12]$ be the value such that $\vBin{\rnd,\eps_i}(0) = \delta_i =\vBin{m-i,0}(-\sum_{j=1}^i c_j)$.\footnote{Namely, $\eps_i$ is set to the number such that a fresh new game with $\eps_i$-biased coins, has the same value (\ie expected outcome) as that of the current game at this point (\ie conditioning on $c_1,\dots,c_i$).}}
 \item For each pair of the three parties, generate shares for an execution of $\pi_2$, by calling $\TwoShareGen(1^\rnd, \delta_i)$ for $\delta_i = \pr{\sum_{j=1}^\rnd c_j \geq 0 \mid c_1, \ldots, c_i}$.
 \end{enumerate}

 \item Split the values of $c_1,\dots,c_\rnd$ and the defense values into three set of shares using a 3-out-of-3 secret sharing scheme, and output the three sets.
 \end{enumerate}
\end{algorithm}
\begin{protocol}[$\pi_3 = (\Pthree_0,\Pthree_1,\Pthree_2)$]\label{intro:protocol:threeparty}~
\begin{enumerate}
\item[Common input:] round parameter $1^\rnd$.

\item[Initial step:] The parties securely compute $\ThreeShareGen(1^\rnd)$, where each party gets one set of shares.

\item[Main loop:] For $i=1$ to $\rnd$, do
\begin{enumerate}
 \item Each party sends to the other parties its share of their defense values.

 \item[$\bullet$] Each pair $(\Pthree_\z,\Pthree_{\z'})$ of the parties reconstructs a pair of two sets of shares $\defn{\z,\z'}{i} = ((\defn{\z,\z'}{i})_\z,(\defn{\z,\z'}{i})_{\z'})$, to serve as input for an execution of the two-party protocol if the third party aborts (\ie $\Pthree_\z$ reconstructs $(\defn{\z,\z'}{i})_\z$, and $\Pthree_{\z'}$ reconstructs $(\defn{\z,\z'}{i})_{\z'}$).

 \item Each party sends the other parties its share of $c_i$.
 \item[$\bullet$] All parties reconstruct the value of $c_i$.
\end{enumerate}

\item[Output:] The parties output one if $\sum_{j=1}^{\rnd} c_j \geq 0$, and zero otherwise.

\item[Abort:]

\begin{itemize}
 \item If $\Pthree_0$ aborts, the parties $\Pthree_1$ and $\Pthree_2$ use the shares of $\defn{1,2}{i}$, for the \emph{maximal} $i\in [\rnd]$ that has been reconstructed, to interact in $\pi^2$ (starting right after the share reconstruction phase). If no such $i$ exists, the parties interact in the (full, unbiased) two-party protocol $\pi^2$.

 The case that $\Pthree_1$ or $\Pthree_2$ aborts is analogously defined.

 \item If two parties abort in the same round, the remaining party acts as if one party has only aborted in the very beginning of the two-party protocol.
\end{itemize}

\end{enumerate}
\end{protocol}

Similar to the analysis for the two-party protocol sketched above, it suffices to show that the defense values  reconstructed by a pair of corrupted parties in round $(i,a)$ (\ie the inputs for the two-party protocols) do not give too much information about the value of $\delta_i$ --- the expected outcome of the three-party protocol conditioned on the coins reconstructed at round $(i,b)$. Note that once two corrupted parties are given these defense values, which happens in round $(i,a)$, they can \emph{immediately} reconstruct the whole two-party execution induced by them. This two-party execution effectively contains  $\Theta(\rnd)$ \emph{independent} samples from $\Berzo{\delta_i}$: one sample is given explicitly as the final output of the execution, and the value of $2\rnd$ additional samples can be extrapolated from the $2\rnd$ defense values given to the two parties. Many such independent samples can be used to reveal the value of $\delta_i$ just by looking at the sum of those $2\rnd$ samples, and guess its value according to it.\footnote{This can be done by the following process: If the sum is greater than $2\rnd \cdot \delta_{i-1}$, guess that $c_i = 1$. Otherwise, guess that $c_i = -1$. Since $\size{\delta_i - \delta_{i-1}} \in \Omega(\frac1{\sqrt{\rnd}})$, it follows by Hoeffding inequality \cite{Hoeffding1963} that the guess is good with probability $\frac12 + \Theta(1)$. Similarly, it can be shown that the same problem also holds in the weighted variant of the protocol which mentioned in the security proof of \cref{intro:protocol:twoparty}} Hence, in round $(i,a)$,  two corrupted parties can rush and use the above information to bias the outcome of the three-party protocol by $\Theta(\size{\delta_i - \delta_{i-1}}) \in \Omega(\frac1{\sqrt{\rnd}})$.

We solve this issue using a \emph{hiding} variant of the two-party shares generating function --- a function that leaks only \emph{limited} information about the value of $\delta_i$. See details in \cref{sec:protocol}.

\subsection{Open Problems}
The existence of an optimally fair three-party coin-flipping protocol (without the $\poly(\log \rnd)$ factor) is still an interesting open question. A more fundamental question is whether there exist fair coin-flipping protocols for more than three parties (against any number of corrupted parties).
A question of a larger scope, is to find $1/\rnd$-fair protocol for other many-party functionality, as done by \citet{GordonKatz10} for two-party functionalities. In particular, can one harness our three-party protocol for this aim.

Finally, some of the proofs we give for bounding the values of the Binomial games, which are in sense equivalent to bounding the bias of our coin-flipping protocols, are long and tedious. Finding simpler proofs would be a good service, and might yield tighter bounds and  increase our understanding of these protocols.

\subsection*{Paper Organization}
General notations and definitions used throughout the paper are given in \cref{section:Preliminaries}. We also state there (\cref{sec:provingFairness}) a new game-based definition of fair coin-flipping protocols, which is equivalent to the standard real/ideal definition. Our coin-flipping protocols, along with their security proofs, are given in \cref{sec:protocol}. The security proofs of \cref{sec:protocol} use bounds on  the value of several types of online-binomial games, these bounds are proven in \cref{section:BoundsForOnlineBinomialGames}. Missing proofs can be found in \cref{sec:missinProofs}.

\section{Preliminaries}\label{section:Preliminaries}

\subsection{Notation}\label{sec:prelim:notation}
We use calligraphic letters to denote sets, uppercase for random variables and functions,  lowercase for values, boldface for vectors and capital boldface for matrices. All logarithms considered here are in base two. For $a\in \R$ and $b\geq 0$,  let $a\pm b$ stand for the interval $[a-b,a+b]$. Given sets $\cs_1,\ldots,\cs_k$ and $k$-input function $f$, let  $f(\cs_1,\ldots,\cs_k) \eqdef \set{f(x_1,\ldots,x_j) \colon x_i\in \cs_i}$, \eg $f(1\pm 0.1) = \set{f(x) \colon x\in [.9,1.1]}$. For $n\in \N$, let $[n] \eqdef \set{1,\ldots,n}$ and $(n) \eqdef \set{0,\ldots,n}$. Given a vector $\vct \in \oo^\ast$, let $\w(\vct)\eqdef \sum_{i\in [\size{\vct}]} \vct_i$. Given a vector $\vct \in \oo^\ast$ and a set of indexes $\cI \subseteq [\size{\vct}]$, let $\vct_{\cI} = (\vct_{i_1},\ldots,\vct_{i_{\size{\cI}}})$ where $i_1,\ldots,i_{\size{\cI}}$ are the ordered elements of $\cI$. We let the XOR of two integers, stands for the \emph{bitwise} XOR of their bits.

Let $\poly$ denote the set all polynomials,  \ppt denote  for probabilistic  polynomial time, and  \pptm denote a \ppt algorithm (Turing machine).  A function $\nu \colon \N \to [0,1]$ is \textit{negligible}, denoted $\nu(n) = \negl(n)$, if $\nu(n)<1/p(n)$ for every $p\in\poly$ and large enough $n$.

Given a distribution $D$, we write $x\gets D$ to indicate that $x$ is selected according to $D$. Similarly, given a random variable $X$, we write $x\gets X$ to indicate that $x$ is selected according to $X$. Given a finite set $\cs$, we let $s\la \cs$ denote that $s$ is selected according to the uniform distribution on $\cs$. The support of a distribution $D$ over a finite set $\Uni$, denoted $\Supp(D)$, is defined as $\set{u\in\Uni: D(u)>0}$. The \emph{statistical distance} of two distributions $P$ and $Q$ over a finite set $\Uni$, denoted as $\SD(P,Q)$, is defined as $\max_{\cs\subseteq \Uni} \size{P(\cs)-Q(\cs)} = \frac{1}{2} \sum_{u\in \Uni}\size{P(u)-Q(u)}$.

For $\delta\in [0,1]$, let $\Berzo{\delta}$ be the Bernoulli probability distribution over $\zo$, taking the value $1$ with probability $\delta$ and $0$ otherwise. For $\eps \in [-1,1]$, let  $\Beroo{\eps}$ be the Bernoulli probability distribution over $\oo$, taking the value $1$ with probability   $\frac{1}{2}(1+\eps)$ and $-1$ otherwise.\footnote{Notice the slight change in notation comparing to the those used in the introduction.} For $n\in \N$ and $\eps \in [-1,1]$, let $\Beroo{n,\eps}$ be the binomial distribution induced by the sum of $n$ independent random variables, each distributed according to $\Beroo{\eps}$. For $n \in \N$, $\eps \in [-1,1]$ and $k \in \Z$, let $\vBeroo{n, \eps}(k) \eqdef \ppr{x\la \Beroo{n,\eps}}{x \geq  k} = \sum_{t=k}^{n}\Beroo{n,\eps}(t)$. For $n \in \N$ and $\delta \in [0,1]$, let $\sBias{n}{\delta}$ be the value $\eps \in [-1,1]$ with $\vBeroo{n, \eps}(0) = \delta$.
\remove{For $n,n'\in \N$, $k \in \Z$ and $\eps \in [-1,1]$, let $\Bias{n}{\eps}{n'}{k}$ be the value $\eps'\in [-1,1]$ with $\vBeroo{n',\eps'}(0) = \vBeroo{n,\eps}(k)$.}

For $n \in \N$, $\ell \in [n]$ and $p\in \set{-n,\dots,n}$, define the hypergeometric probability distribution $\Hyp{n,p,\ell}$ by $\Hyp{n,p,\ell}(k) \eqdef \ppr{\cI}{\w(\vct_\cI) = k}$, where $\cI$ is an $\ell$-size set uniformly chosen from $[n]$ and $\vct \in \oo^n$ with $w(\vct)= p$. Let $\vHyp{n,p,\ell}(k) \eqdef \ppr{x\la \Hyp{n,p,\ell}}{x \geq  k} = \sum_{t=k}^{\ell}\Hyp{n,p,\ell}(t)$.

Let $\Phi\colon \R \mapsto (0,1)$ be the cumulative distribution function of the standard normal distribution, defined by $\Phi(x) \eqdef \frac{1}{\sqrt{2\pi}}\int_{x}^{\infty}e^{-\frac{t^2}{2}}dt$.

Finally, for $n\in \N$ and $i \in [n]$, let $\NL{i} \eqdef n+1-i$ and $\NS{i} \eqdef \sum_{j=i}^n \NL{j}$.

We summarize the different notations used throughout the paper in the following tables.

\begin{table}[ht]
	\begin{center}
		\caption{\label{fig:summeryOfBasicFunctions} Basic Functions.}
		\begin{tabular}{|c|c|c|}
			\hline
			\textit{Definition} & \textit{Input Range} & \textit{Output value} \\
			\hline
			$[n]$ & $n \in \N$ & $\set{1,\ldots,n}$\\
			\hline
			$(n)$ & $n \in \N$ & $\set{0,\ldots,n}$\\
			\hline	
			$\NL{i}$ & $n\in \N$, $i \in [n]$ & $n+1-i$\\
			\hline
			$\NS{i}$ & $n\in \N$, $i \in [n]$ & $\sum_{j=i}^n \NL{j}$\\
			\hline
			$\Phi(x)$ & $x \in \R$ & $\frac{1}{\sqrt{2\pi}}\int_{x}^{\infty}e^{-\frac{t^2}{2}}dt$\\
			\hline
			$\w(\vct)$ & $\vct \in \oo^\ast$ & $\sum_{i\in \cI} \vct_i$\\
			\hline
			$\vct_{\cI}$ & $\vct \in \oo^\ast$, $\cI \subseteq [\size{\vct}]$ and $i_1,\ldots,i_{\size{\cI}}$ are the ordered elemets of $\cI$ & $(\vct_{i_1},\ldots,\vct_{i_{\size{\cI}}})$\\
			\hline
			$\SD(P,Q)$ & distributions $P$ and $Q$ over a finite set $\Uni$ & $\frac{1}{2} \sum_{u\in \Uni}\size{P(u)-Q(u)}$\\
			\hline
			$\Supp(D)$ & distributions $D$ over a finite set $\Uni$ & $\set{u\in\Uni: D(u)>0}$\\
			\hline
			$a \pm b$ & $a \in \R$, $b \geq 0$ & $[a-b, a+b]$\\
			\hline
		\end{tabular}
	\end{center}
\end{table}

\newpage

\begin{table}[ht]
	\begin{center}
		\caption{\label{fig:summeryOfDistributions} Distributions.}
		\begin{tabular}{|c|c|c|}
			\hline
			\textit{Distribution} & \textit{Input Range} & \textit{Description} \\
			\hline
			$\Berzo{\delta}$ & $\delta\in [0,1]$ & $1$ with probability $\delta$ and $0$ otherwise.\\
			\hline
			$\Beroo{\eps}$ & $\eps \in [-1,1]$ & $1$ with probability $\frac{1}{2}(1+\eps)$ and $-1$ otherwise\\
			\hline		
			$\Beroo{n,\eps}$ & $n\in \N$, $\eps \in [-1,1]$ & sum of $n$ independent $\Beroo{\eps}$ random variables\\
			\hline
			$\Hyp{n,p,\ell}$ & $n \in \N$, $p\in \set{-n,\dots,n}$, $\ell \in [n]$ & The value of $\w(\vct_{\cI})$ where $\cI$ is an $\ell$-size set uniformly\\
			& & chosen from $[n]$ and $\vct \in \oo^n$ with $\w(\vct) = p$\\
			\hline
			
		\end{tabular}
	\end{center}
\end{table}

\begin{table}[ht]
	\begin{center}
		\caption{\label{fig:summeryOfOtherFunctions} Other Functions.}
		\begin{tabular}{|c|c|c|}
			\hline
			\textit{Definition} & \textit{Input Range} & \textit{Output value} \\
			\hline		
			$\Beroo{n,\eps}(k)$ & $n\in \N$, $\eps \in [-1,1]$, $k \in \Z$ & $\ppr{x \la \Beroo{n,\eps}}{x = k}$\\
			\hline
			$\vBeroo{n,\eps}(k)$ & $n\in \N$, $\eps \in [-1,1]$, $k \in \Z$ & $\ppr{x \la \Beroo{n,\eps}}{x \geq k}$\\
			\hline
			$\sBias{n}{\delta}$ & $n\in \N$, $\delta\in [0,1]$ & The value $\eps \in [-1,1]$ with $\vBeroo{n, \eps}(0) = \delta$\\
			\hline
			$\Hyp{n,p,\ell}(k)$ & $n \in \N$, $p\in \set{-n,\dots,n}$, $\ell \in [n]$, $k \in \Z$ & $\ppr{x \la \Hyp{n,p,\ell}}{x = k}$\\
			\hline
			$\vHyp{n,p,\ell}(k)$ & $n \in \N$, $p\in \set{-n,\dots,n}$, $\ell \in [n]$, $k \in \Z$ & $\ppr{x \la \Hyp{n,p,\ell}}{x \geq k}$\\
			\hline
		\end{tabular}
	\end{center}
\end{table}

\subsection{Basic Inequalities}\label{sec:prelim:BasicInq}

The following proposition is proved in \cref{app:missinProofs:BasicInq}.

\def\claimlinearproblem{Let $n \in \N$, $\alpha > 0$, $k \in [n]$ and let $\set{p_j}_{j=k}^{n}$ be a set of non-negative numbers such that $\sum_{j=i}^{n}p_j \leq \alpha \cdot (n +1-i)$ for every $i\in \set{k,k+1,\ldots,n}$. Then $\sum_{j=k}^{n}\frac{p_j}{(n+1-j)} \leq \alpha \cdot \sum_{j=k}^{n} \frac1{(n+1-j)}$.}

\begin{proposition}\label{claim:linearproblem}
\claimlinearproblem
\end{proposition}

\subsection{Facts About the Binomial Distribution}\label{sec:prelim:Binomial}

\begin{fact}[Hoeffding's inequality for $\oo$]\label{claim:Hoeffding}
Let $n,t \in \N$ and $\eps \in [-1,1]$. Then
\begin{align*}
\ppr{x\la \Beroo{n,\eps}}{\abs{x-\eps n} \geq t} \leq 2e^{-\frac{t^2}{2n}}.
\end{align*}
\end{fact}
\begin{proof}
	Immediately follows by  \cite{Hoeffding1963}.
\end{proof}

\def\factEXM{Let $n \in \N$ and $\eps \in [-\frac{1}{\sqrt{n}},\frac{1}{\sqrt{n}}]$.
Then $\ex{x\la \Beroo{n,\eps}}{x^2} \leq 2n$ and $\ex{x\la \Beroo{n,\eps}}{\abs{x}} \leq \sqrt{2n}$.}

\begin{fact}\label{fact:E_x_m}
\factEXM
\end{fact}
\begin{proof}
	A simple calculation yields that $\ex{x\la \Beroo{n,\eps}}{x^2} = n(1-\eps^2) + \eps^2n^2$, which is smaller than $2n$ by the bound on $\eps$. The second bound holds since $\ex{x\la \Beroo{n,\eps}}{\abs{x}} \leq \sqrt{\ex{x\la \Beroo{n,\eps}}{x^2}}$.
\end{proof}

The following propositions are proved in \cref{app:missinProofs:Binomial}.

\def\propBinomProbEstimation{Let $n \in \N$, $t \in \Z$ and $\eps \in [-1,1]$ be such that $t \in \Supp(\Beroo{n, \eps})$, $\size{t} \leq n^{\frac{3}{5}}$ and $\size{\eps} \leq  n^{-\frac{2}{5}}$. Then
\begin{align*}
\Beroo{n, \eps}(t) \in  (1 \pm \error) \cdot \sqrt{\frac{2}{\pi}} \cdot \frac{1}{\sqrt{n}}\cdot  e^{-\frac{(t-\eps n)^2}{2n}},
\end{align*}
for $\error  = \xi \cdot (\eps^2 \abs{t} + \frac{1}{n} + \frac{\abs{t}^3}{n^2} + \eps^4 n)$ and a universal constant $\xi$.}

\begin{proposition}\label{prop:binomProbEstimation}
\propBinomProbEstimation
\end{proposition}

\def\propBinomProbRelation{Let $n \in \N$, $t,x,x' \in \Z$, $\eps\in [-1,1]$ and $\const > 0$ be such that $t-x, t-x' \in \Supp(\Beroo{n,\eps})$, $\size{x},\size{x'},\size{t} \leq \const \cdot  \sqrt{n\log n}$ and $\size{\eps} \leq \const \cdot \sqrt{\frac{\log{n}}{n}}$, then
\begin{align*}
\frac{\Beroo{n,\eps}(t-x')}{\Beroo{n,\eps}(t-x)}
&\in (1\pm \error)\cdot \exp\left(\frac{- 2\cdot (t - \eps n)\cdot x + x^2 + 2\cdot (t - \eps n)\cdot x' - x'^2}{2n}\right),
\end{align*}
for $\error = \varphi(\const) \cdot \frac{\log^{1.5} n}{\sqrt{n}}$ and a universal function $\varphi$.}

\begin{proposition}\label{prop:binomProbRelation}
\propBinomProbRelation
\end{proposition}

\def\propGameValuesDifferenceBound{Let $n \in \N$, $k,k' \in \Z$ and $\eps \in [-1,1]$, where $n$ is larger than a universal constant, $\size{k},\size{k'} \le n^{\frac35}$ and $\size{\eps} \leq  n^{-\frac25}$. Then
\begin{align*}
\abs{\vBeroo{n, \eps}(k) - \vBeroo{n, \eps}(k')} \leq
\frac{\abs{k-k'}}{\sqrt{n}}.
\end{align*}}

\begin{proposition}\label{prop:gameValuesDifferenceBound}
\propGameValuesDifferenceBound
\end{proposition}

\def\propEpsDiff{Let $n, n' \in \N$, $k\in \Z$, $\eps \in [-1,1]$ and $\const > 0$ be such that $n \leq n'$, $\abs{k} \leq \const \cdot \sqrt{n \log n}$, $\size{\eps}  \leq \const \cdot \sqrt{\frac{\log n}{n}}$, and let $\delta = \vBeroo{n,\eps}(k)$. Then
\begin{align*}
\sBias{n'}{\delta} \in \frac{\eps n - k}{\sqrt{n\cdot n'}} \pm \error,
\end{align*}
for $\error = \varphi(\const) \cdot \frac{\log^{1.5} n}{\sqrt{n\cdot n'}}$ and a universal function $\varphi$.}

\begin{proposition}\label{prop:epsDiff}
\propEpsDiff
\end{proposition}

\def\propMainBound{Let $n\in \N$, integer $i \in [n-\floor{\log^{2.5}n}]$, $x,\beta,\beta',\alpha,\alpha' \in \Z$, $\eps \in [-1,1]$, $\cs \subseteq \Z$ and $\const > 0$  such that $\size{\alpha},\size{\alpha'} \leq \sqrt{\const \cdot  \NS{i} \cdot \log n}$,  $\size{\beta},\size{\beta'} \le 1$, $\cs \subseteq [-\sqrt{\const \cdot \NL{i} \cdot \log n},\sqrt{\const \cdot \NL{i} \cdot \log n}]$,  $x  \in \cs$, $\size{\eps} \leq \sqrt{\const \cdot  \frac{\log n}{\NS{i}}}$ and $\ex{x'\la \Beroo{\NL{i},\eps} \mid x'\in \cs}{\abs{x'}} \leq \ex{x'\la \Beroo{\NL{i},\eps}}{\abs{x'}}$. Then
\begin{align*}
\ex{x'\la \Beroo{\NL{i},\eps} \mid x'\in \cs}
{\exp\left(\frac{\alpha \cdot x + \beta \cdot x^2 + \alpha' \cdot x' + \beta' \cdot x'^2}{\NS{i+1}}\right)}
 \in 1 \pm \varphi(\const) \cdot \sqrt{\frac{\log{n}}{\NL{i}}}\left(1 + \frac{\abs{x}}{\sqrt{\NL{i}}}\right).
\end{align*}
for a universal function $\varphi$.}

\begin{proposition}\label{prop:mainBound}
\propMainBound
\end{proposition}

\subsection{Facts About the Hypergeometric Distribution}\label{sec:prelim:HypGeo}
\begin{fact}[Hoeffding's inequality for hypergeometric distribution]\label{fact:hyperHoeffding}
Let $\ell \leq n \in \N$, and $p \in \Z$ with $\size{p}\ \leq n$. Then
$$\ppr{x\la \Hyp{n,p,\ell}}{{\abs{x-\mu}} \geq t} \leq e^{-\frac{t^2}{2\ell}},$$
for $\mu = \ex{x\la \Hyp{n,p,\ell}}{x} = \frac{\ell  \cdot p}{n}$.
\end{fact}
\begin{proof}
Immediately follows by  \cite[Equations (10),(14)]{scala2009hypergeometric}.
\end{proof}

The following proposition is proved in \cref{app:missinProofs:HypGeo}.

\def\propHyperProbEstimation{Let $n \in \N$, $p, t\in \Z$ be such that $\size{p}, \size{t} \leq n^{\frac35}$ and $t \in \Supp(\Hyp{2n, p, n})$. Then
\begin{align*}
\Hyp{2n, p, n}(t) &\in
(1 \pm \error) \cdot \frac{2}{\sqrt{\pi\cdot n}} \cdot e^{-\frac{(t-\frac{p}2)^2}{n}},
\end{align*}
for $\error= \xi \cdot (\frac{n + \abs{p}^3 + \abs{t}^3}{n^2})$ and a universal constant $\xi$.}

\begin{proposition}\label{prop:hyperProbEstimation}
\propHyperProbEstimation
\end{proposition}

\def\propHyperProbRelation{Let $n \in \N$, $p,t,x,x' \in \Z$ and $\const > 0$ be such that $t-x, t-x' \in \Supp(\Hyp{2n, p, n})$ and $\size{p},\size{t},\size{x},\size{x'} \leq \const \cdot \sqrt{n \log n}$. Then
\begin{align*}
\frac{\Hyp{2n, p, n}(t-x')}{\Hyp{2n, p, n}(t-x)}
&\in (1 \pm \error) \cdot \exp\left(\frac{-2(t-\frac{p}2)x + x^2 + 2(t-\frac{p}2)x' - x'^2}{n}\right),
\end{align*}
for $\error = \varphi(\const)\cdot \frac{\log^{1.5}{n}}{\sqrt{n}}$ and a universal function $\varphi$.}
\begin{proposition}\label{prop:hyperProbRelation}
\propHyperProbRelation
\end{proposition}

\subsection{Multi-Party Protocols}\label{sec:Protocols}
The following discussion is restricted to no private input protocols (such restricted protocols suffice for our needs).

A  $\partNum$-party  protocol is defined using $\partNum$ Turing Machines (TMs) $\Pc_1,\ldots,\Pc_\partNum$,
having the security parameter $1^\secParam$ as their common input. In each round, the parties broadcast and receive
messages on  a broadcast channel. At the end of protocol, each  party outputs some binary string.

The parties communicate in a synchronous network, using only a broadcast channel: when a party broadcasts a message, all other parties see \emph{the same} message. This ensures some consistency between the information the parties have. There are no private channels and all the parties see all the messages, and can identify their sender. We do not assume simultaneous broadcast. It follows that  in each round, some parties might hear the messages sent by the other parties before broadcasting their messages. We assume that if a party aborts, it first broadcasts the message \Abort to the other parties, and \wlg only does so at the end of a round in which it is supposed to send a message.
A protocol is \emph{efficient}, if its parties are \pptm, and the protocol's number of rounds is a computable function of the security parameter.

This work focuses on  efficient protocols, and on malicious, static \ppt adversaries for such protocols. An adversary is allowed to corrupt some subset of the parties; before the beginning of the protocol, the adversary corrupts a subset of the parties that from now on may arbitrarily deviate from the protocol. Thereafter, the adversary sees the messages sent to the corrupted parties and controls their messages. We also consider the so called \textit{fail-stop} adversaries. Such adversaries follow the prescribed protocol, but might abort prematurely. Finally, the honest parties follow the instructions of the protocol to its completion.

\subsection{The Real vs. Ideal Paradigm}\label{sec:realIdeal}
The security of multi-party computation protocols is defined using the \emph{real} vs. \emph{ideal} paradigm \cite{Canetti00,Goldreich04}. In this paradigm, the \emph{real-world model}, in which
protocols is  executed is compared to an \emph{ideal model} for executing the task at hand.
The latter model involves a trusted party whose functionality captures the security requirements
of the task. The security of the real-world protocol is argued by showing that it ``emulates''  the ideal-world protocol, in the following sense: for any real-life adversary $\Aadv$, there exists an ideal-model adversary (\aka simulator) $\iAadv$  such that the global output of an execution of the protocol with $\Aadv$ in the real-world model is distributed similarly to the global output of running
$\iAadv$ in the ideal model. The following discussion is restricted  to random,  no-input functionalities. In addition, to keep the presentation simple, we limit our attention to uniform adversaries.\footnote{All results stated in this paper, straightforwardly extend to the  non-uniform settings.}

\mparagraph{The Real Model}
Let $\pi$ be an $\partNum$-party protocol and let $\Aadv$ be an adversary controlling a subset $\cC \subseteq [\partNum]$ of the parties. Let $\Real_{\pi,\Aadv,\cC}(\secParam)$ denote the output of $\Aadv$  (\ie \wlg its view: its random input and the messages it received) and the outputs of the honest parties, in a random  execution of $\pi$ on common input $1^\secParam$.

Recall that an adversary is \emph{fail stop}, if until they abort, the parties in its control follow the prescribed protocol (in particular, they property toss their private random coins). We call  an execution of $\pi$ with such a fail-stop adversary, a fail-stop execution. 

\mparagraph{The Ideal Model}
 Let  $f$ be a $\partNum$-output  functionality. If $f$ gets a security parameter (given in unary), as its first input, let $f_\secParam(\cdot) = f(1^\secParam,\cdot)$.  Otherwise, let $f_\secParam = f$.

An ideal execution of $f$ \wrt an adversary $\iAadv$  controlling a subset $\cC \subseteq [\partNum]$ of the ``parties" and a security parameter  $1^\secParam$, denoted  $\Ideal_{f,\iAadv,\cC}(\secParam)$, is the output of the adversary $\iAadv$ and that of the trusted party, in the following experiment.
\begin{experiment}~
\begin{enumerate}
  \item The trusted party sets $(y_1,\dots,y_\partNum) = f_\secParam(X)$, where $X$ is a uniform element in the domain of $f_\secParam$,  and sends $\set{y_i}_{i\in \cC}$ to $\iAadv(1^\secParam)$.

  \item $\iAadv(1^\secParam)$ sends  the message $\mathsf{Continue}$/ $\mathsf{Abort}$ to the trusted party, and locally outputs some value.

  \item The trusted party outputs $\set{o_i}_{i\in [\partNum] \setminus \cC}$, for  $o_i$ being $y_i$ if $\iAadv$ instructs  $\mathsf{Continue}$, and $\perp$ otherwise.
\end{enumerate}
\end{experiment}
An adversary $\iAadv$ is non-aborting, if it never sends the $\mathsf{Abort}$ message.

\subsubsection{\texorpdfstring{$\alpha$}{Alpha}-Secure Computation}
The following definitions adopts the notion of $\alpha$-secure computation \cite{BeimelLOO11,GordonKatz10,Katz07} for our restricted settings. 
\begin{definition}[$\alpha$-secure computation]\label{def:deltaSecure}
An efficient  $\partNum$-party protocol $\pi$ computes  a $\partNum$-output functionality $f$ in a {\sf  $\alpha$-secure} manner [\resp against fail-stop adversaries], if for every $\cC \subsetneq [\partNum]$ and every [\resp fail-stop] \ppt
adversary $\Aadv$ controlling the parties indexed by $\cC$,\footnote{The requirement that $\cC$ is a \emph{strict} subset of $[\partNum]$, is merely for notational convinced.} there exists a \ppt $\iAadv$ controlling the same parties, such that
$$\SD\left(\Real_{\pi,\Aadv,\cC}(\secParam),\Ideal_{f,\iAadv,\cC}(\secParam)\right) \leq \alpha(\secParam),$$
for large enough $\secParam$.

A protocol securely compute a functionality $f$, if it computes $f$ in a $\negl(\secParam)$-secure manner.

The protocol $\pi$ computes $f$ in a {\sf simultaneous $\alpha$-secure} manner, if the above is achieved by a {\sf non-aborting} $\iAadv$.
\end{definition}
Note that being simultaneous $\alpha$-secure is a very strong requirement, as it dictates that the cheating  real adversary has no way to prevent the honest parties from getting their part of the output, and this should be achieved with no simultaneous broadcast mechanism.

\subsection{Fair Coin-Flipping Protocols}\label{sec:CFprotocols}
\begin{definition}[$\alpha$-fair coin-flipping]\label{def:fairCTSim}
For $\partNum\in \N$ let $\CoinToss_\partNum$ be the $\partNum$-output functionality from $\zo$ to $\zo^\partNum$, defined by $\CoinToss_{\partNum}(b)= b \ldots b$ ($\partNum$ times).  A $\partNum$-party protocol $\pi$ is {\sf $\alpha$-fair coin-flipping protocol}, if it computes $\CoinToss_\partNum$ in  a {\sf simultaneous $\alpha$-secure} manner.
\end{definition}

\subsubsection{Proving Fairness}\label{sec:provingFairness}
The following lemma reduces the task of proving fairness of a coin-flipping protocol,  against fail-stop adversaries, to proving the protocol is correct: the honest parties always output the same bit, and this bit is uniform in an all honest execution, and to proving the protocol is unbiased: a fail-stop adversary cannot bias the output of the honest parties by too much.

\begin{definition}[correct coin-flipping protocols]\label{def:CorrectCT}
A protocol is a {\sf correct coin flipping}, if
 \begin{itemize}
   \item When interacting with an fails-stop adversary controlling a subset of the parties, the honest parties {\sf always} output the same bit, and
   \item The common  output in  a random {\sf honest} execution of $\pi$, is uniform over $\zo$.
 \end{itemize}
\end{definition}

Given a partial view of a fail-stop adversary, we are interesting in the  expected outcome of the parties, conditioned on this and the adversary making no further aborts.
\begin{definition}[view value]\label{def:ViewVal}
Let $\pi$ be a protocol in which the honest parties always output the same bit value. For a partial  view $v$ of the parties in a fail-stop  execution of $\pi$, let $\Cc_\pi(v)$ denote the parties' full view in an {\sf honest} execution of $\pi$ conditioned on $v$ (\ie all parties that do not abort in $v$ act honestly in $\Cc_\pi(v)$). Let $\val_\pi(v) = \ex{v' \la \Cc_\pi(v)}{\out(v')}$,  where $\out(v')$ is the common output of the non-aborting  parties in $v'$.
\end{definition}

Finally, a protocol is unbiased, if no fail-stop adversary can bias the common output of the honest parties by too much.
\begin{definition}[$\alpha$-unbiased coin-flipping protocols]\label{def:FairCTGameAlt}
A $\partNum$-party, $\rnd$-round  protocol $\pi$ is {\sf $\alpha$-unbiased}, if the following holds for every fail-stop adversary $\Aadv$ controlling  the parties indexed by a subset $\cC\subset [\partNum]$. Let $V$ be the view of the corrupted parties controlled by $\Aadv$ in a random execution of $\pi$, and let $I_j$ be the index of the $j$'th round in which $\Aadv$ sent an abort message (set to $\rnd+1$, if no such round). Let $V_{i}$ be the prefix of $V$ at the end of the $i$'th round, letting $V_0$  being the empty view, and let $V_i^{-}$ be the prefix of $V_i$ with the $i$'th round abort messages (if any) removed. Then
  $$\size{\ex{V}{\sum_{j\in \size{\cC}} \val(V_{I_j}) -  \val(V_{I_j}^-)}} \leq  \alpha,$$
  where $\val = \val_\pi$ is according to \cref{def:ViewVal}.
\end{definition}

The following is an alternative characterization of fair coin-flipping protocols (against fail-stop adversaries).
\begin{lemma}\label{prop:FairCTGameAlt}
Let $\pi$ be a correct,  $\alpha$-unbiased coin-flipping protocol with $\alpha(\secParam) \leq \frac12 - \frac1{p(\secParam)}$, for some $p\in \poly$, then $\pi$ is a $(\alpha(\secParam) + \negl(\secParam))$-secure coin-flipping protocol against fail-stop
adversaries.
\end{lemma}

\begin{proof}
Let  $\Aadv$ be a \ppt fail-stop  adversary controlling  a subset $\cC\subsetneq [\partNum]$ of the parties. The ideal-world adversary $\iAadv$ is defined as follows.
\begin{algorithm}[$\iAadv$]
\item[Input:] $1^\secParam$.

\item [Operation:] Upon receiving $\set{y_i = b}_{i\in \cC}$ from the trusted party, for some $b\in \zo$, do:
\begin{enumerate}
  \item Keep sampling uniformly at random coins for the parties of $\pi$ and for $\Aadv$, on security parameter $\secParam$, until the honest parties' common output in the resulting execution is $b$. Abort after $\secParam\cdot p(\secParam)$ failed attempts.\label{step:prop:FairCTGameAlt:1}
  \item Output $\Aadv$'s output in the above sampled execution.
\end{enumerate}
\end{algorithm}
Let $D_\secParam$ be the distribution of the honest parities common output, in a random execution of $\pi(1^\secParam)$, in which $\Aadv$
controls the parties indexed by  $\cC$. Assume for a moment  that the trusted party chooses its output on security parameter $\secParam$, according to $D_\secParam$ (and not uniformly at random). Assume further that $\iAadv$ keeps sampling in Step $1$ until good coins  are found. Under these assumptions, it is immediate that $\iAadv$ is a \emph{perfect} ideal variant simulator for  $\Aadv$, \ie $\Real_{\pi,\Aadv,\cC}(\secParam) \equiv \Ideal_{f,\iAadv,\cC}(\secParam)$ for every $\secParam$. We complete the proof showing that $\SD(D_\secParam,U)\leq \alpha(\secParam)$, where $U$ is the uniform distribution over $\zo$. This yields that $\SD\left(\Real_{\pi,\Aadv,\cC}(\secParam),\Ideal_{f,\iAadv,\cC}(\secParam)\right) \leq \alpha(\secParam)$ assuming no abort in \Stepref{step:prop:FairCTGameAlt:1}. Since, by assumption, $\alpha(\secParam) \leq \frac12 - \frac1{p(\secParam)}$,  such aborts happens only  with negligible probability, it will follow that $\SD\left(\Real_{\pi,\Aadv,\cC}(\secParam),\Ideal_{f,\iAadv,\cC}(\secParam)\right) \leq \alpha(\secParam) + \negl(\kappa)$.

Let $\val$, $V$, $V_i$, $V_i^-$ and $I_j$ be as in \cref{def:FairCTGameAlt} \wrt algorithm $\Ac$, subset  $\cC$ and protocol $\pi$.  We prove by induction on $\ell\in \size{\cC}$ that $\eex{\val(V_{I_\ell})}= \frac12 + \beta_\ell$, for $\beta_\ell = \sum_{j\in [\ell]} \eex{\val(V_{I_j}) -  \val(V_{I_j}^-)}$. Since no abort occurs after the $\size{\cC}$'th aborting round,  it follows that  $\eex{\val(V)} = \frac12 + \beta _{\size{\cC}}$. Since $\pi$ is $\alpha$-unbiased, it follows that $\eex{\val(V)} \in [\frac12 \pm \alpha(\secParam)]$, and therefore $\SD(D_\secParam,U)\leq \alpha(\secParam)$.

The base case (\ie $\ell=0$) holds by the correctness of $\pi$. Assume for $0\leq \ell < \size{\cC}$. Since no additional aborts messages were sent in $V_{I_{\ell+1}}^-$ beside the ones sent $V_{I_\ell}$, it holds that
\begin{align}\label{eq:FairCTGameAlt}
\eex{\val(V_{I_{\ell+1}}^-)} = \eex{\val(V_{I_\ell})}
\end{align}
It follows that
\begin{align*}
\eex{\val(V_{I_{\ell+1}})} &= \eex{\val(V_{I_{\ell+1}}^-)} + \eex{\val(V_{I_{\ell+1}}) -\val(V_{I_{\ell+1}}^-)}\\
&= \eex{\val(V_{I_\ell})} + \eex{\val(V_{I_{\ell+1}}) -\val(V_{I_{\ell+1}}^-)}\\
&= \left(\frac12 + \sum_{j\in [\ell]} \eex{\val(V_{I_j}) -  \val(V_{I_j}^-)}\right) +  \eex{\val(V_{I_{\ell+1}}) -\val(V_{I_{\ell+1}}^-)}\\
&= \frac12 + \sum_{j\in [\ell+1]} \eex{\val(V_{I_j}) -  \val(V_{I_j}^-)}.
\end{align*}
The second equality  holds by \cref{eq:FairCTGameAlt} and the third one by the induction hypothesis.
\end{proof}

\subsection{Oblivious Transfer}\label{sec:OT}
\begin{definition}\label{def:OTfunct}
The $\binom{1}{2}$  oblivious transfer ($\OT$ for short) functionality, is the two-output functionality $f$ over $\zo^3$, defined by $f(\sigma_0,\sigma_1,i) = ((\sigma_0,\sigma_1), (\sigma_i,i))$.
\end{definition}
Protocols the securely compute $\OT$, are known under several hardness assumptions (\cf \cite{AieIshRei01,EvenGL85,GentryPV2008,Haitner04,Kalai05,NaorPinkas01}).

\subsection{\texorpdfstring{$f$}{f}-Hybrid Model}\label{sec:OTHybrid}
Let $f$ be a $\partNum$-output functionality. The {$f$-hybrid} model is identical to the real model of computation discussed above, but in addition, each $\partNum$-size subset of  the parties involved, has access to a trusted party realizing  $f$. It is important to emphasize  that the trusted party realizes $f$ in a \emph{non}-simultaneous manner: it sends a random output of $f$ to the parties in an arbitrary order. When a  party  gets its part of the output, it  instructs the trusted party to either continue sending the output to the other parties, or to send them the abort symbol (\ie the  trusted party  ``implements" $f$  in a perfect non-simultaneous manner).

All notions given in \cref{sec:realIdeal,sec:CFprotocols} naturally extend to the $f$-hybrid model, for any functionality $f$. In addition,  the proof of \cref{prop:FairCTGameAlt} straightforwardly extends to this model.

We  make use of the following known fact.
\begin{fact}\label{fact:fHybridMToReal}
Let $f$ be a polynomial-time computable functionality, and assume there exists a $k\in O(1)$-party, $\rnd$-round, $\alpha$-fair coin-flipping  protocol in the $f$-hybrid model, making at most $t$ calls to $f$. Assuming there exist protocols for  securely computing \OT, then there exists a $k$-party, $(O(t)+ \rnd)$-round, $(\alpha +\negl)$-fair coin-flipping  protocol (in the real world).
\end{fact}
\begin{proof}
Since $f$ is a polynomial-time computable and since we assume the existence of  a protocol for  securely computing \OT, there exists a constant-round protocol $\pi_f$ for  securely computing $f$:  a constant-round protocol for $f$ that is secure against semi-honest adversaries follows by \citet{BMR90} (assuming \OT), and the latter protocol can be compiled into a constant-round protocol that securely computes $f$, against arbitrary malicious adversaries,  using the techniques of \citet{GoldreichMiWi87} (assuming one-way functions, that follows by the existence of \OT). Let $\pi$ be a $k$-party, $\rnd$-round, $\alpha$-fair coin-flipping  protocol  in the $f$-hybrid model. \citet{Canetti00} yields that by replacing the trusted party for computing $f$ used in $\pi$ with the protocol $\pi_f$, we get an $(O(t)+ \rnd)$-round, $(\alpha +\negl)$-fair coin-flipping  protocol.
\end{proof}

\section{The Protocols}\label{sec:protocol}

The following protocols follows the high-level description given in \cref{sec:Technique}.

\subsection{Two-Party Protocol}\label{sec:TwoPartyprotocol}
We start with defining a coin-flipping protocol whose parties get (correlated) shares as input,  then describe the functionality for generating these shares, and finally explain how to combine the two into a (no input) coin-flipping protocol.

\subsubsection{The Basic Two-Party Protocol}\label{sec:TwoPartyProtocolBasic}
\begin{protocol}[$\Pitwo_\rnd = (\Ptwo_0,\Ptwo_1)$]\label{protocol:TwoBasic}
\item[Common input:] round parameter $1^\rnd$.

\item[$\Ptwo_\z$'s input:] $\cVS{\z} \in \zo^{\rnd\times \tp{\rnd}}$ and $\dVS{0}{\z},\dVS{1}{\z} \in \zo^{\rnd+1}$.

\item [Protocol's description:]~
\begin{enumerate}

\item For $i=1$ to $\rnd$:
\begin{enumerate}
\item $\Ptwo_0$ sends $\dVS{1}{0}[i]$ to $\Ptwo_1$, and $\Ptwo_1$ sends $\dVS{0}{1}[i]$ to $\Ptwo_0$.\label{step:TwoParty:round:1}

\item[$\bullet$] For $\z\in \zo$, party $\Ptwo_\z$ set $d^{\z}_i = \dVS{\z}{0}[i] \xor \dVS{\z}{1}[i]$.

\item $\Ptwo_0$ sends $\cVS{0}[i]$ to $\Ptwo_1$, and $\Ptwo_1$ sends $\cVS{1}[i]$ to $\Ptwo_0$.\label{step:TwoParty:round:2}

\item[$\bullet$] Both parties set $c_i = \cVS{0}[i] \xor \cVS{1}[i]$.

\end{enumerate}

\item Both parties output 1 if $\sum_{i=1}^\rnd c_i \geq 0$, and 0 otherwise.
\end{enumerate}

\item [Abort:] If the other party aborts, the remaining party $\Ptwo_\z$ outputs $d^{\z}_i$, for the {\sf maximal} $i\in [\rnd]$ for which it has reconstructed this value. In case no such $i$ exists, $\Ptwo_\z$ outputs $\dVS{\z}{\z}[\rnd+1]$.\\
\end{protocol}

To keep the above description symmetric, in \Stepref{step:TwoParty:round:1} and in \Stepref{step:TwoParty:round:2}, both parties are supposed to send messages. This is merely for notational convince, and one might assume that the parties send their messages in an arbitrary order.

\subsubsection{Two-Party Shares Generator}
We now define the share-generating function of our two-party coin-flipping protocol. For future use, we describe a parameterized variant of this function that gets, in addition to the round parameter, also the desired expected outcome of the protocol. Our two-party protocol will call this function with expected outcome $\frac12$.

Recall that $\Berzo{\delta}$ is the Bernoulli probability distribution over $\zo$, taking the value $1$ with probability $\delta$ and $0$ otherwise, that $\Beroo{\eps}$ is the Bernoulli probability distribution over $\oo$, taking the value $1$ with probability $\frac{1}{2}(1+\eps)$ and $-1$ otherwise,\footnote{Notice the slight change in notation compared to those used in the introduction.} that $\Beroo{n, \eps}(k) = \pr{\sum_{i=1}^{n} x_i = k}$, for $x_i$'s that are i.i.d according to $\Beroo{\eps}$, and that $\vBeroo{n, \eps}(k) = \ppr{x\la \Beroo{n,\eps}}{x \geq  k}$. Also recall that $\sBias{n}{\delta}$ is the value $\eps \in [-1,1]$ with $\vBeroo{n,\eps}(0) = \delta$, and that $\NL{i} = n+1-i$ and $\NS{i} = \sum_{j=i}^n \NL{j}$. Finally, for $\z\in \zo$ let $\overline{\z} = \z+1 \bmod 2$.

\begin{algorithm}[$\TwoShareGen$]\label{alg:TwoShareGenerator}
\item[Input:] round parameter $1^\rnd$ and $\delta \in [0,1]$.
\item[Operation:]~
\begin{enumerate}
\item For $\z\in \zo$: sample $d^{\z,\#\z}_{\rnd+1} \la \Berzo{\delta}$. Set $d^{\z,\#\overline{\z}}_{\rnd+1}$ arbitrarily.
\item Let $\eps = \sBias{\ms{1}}{\delta}$.\footnote{Note that $\sBias{\ms{1}}{\frac12} = 0$ if $\ms{1}$ is odd.}

\item For $i=1$ to $\rnd$:
\begin{enumerate}
\item Sample $c_i \la \Beroo{\ml{i} ,\eps}$.
\item Sample $c_i^{\#0} \la \zo^{\tp{\rnd}}$ and set $c_i^{\#1} = c_i \xor c_i^{\#0}$.

\item For $\z\in \zo$:
\begin{enumerate}
 \item Sample $d^{\z}_i \la \Berzo{\vBeroo{\ms{i+1},\eps}(-\sum_{j=1}^{i} c_j)}$.\label{step:independently}

\item Sample $d^{\z,\#0}_i \la \zo$, and set $d^{\z,\#1}_i = d^{\z}_i \xor d^{\z,\#0}_i$.
\end{enumerate}
\end{enumerate}

 \item Output $(\sVS{0},\sVS{1})$, where $\sVS{\z} = (\cVS{\z},\dVS{0}{\z}, \dVS{1}{\z})$, for $\cVS{\z} = (c^{\#\z}_1,\dots,c^{\#\z}_\rnd)$ and $\dVS{\z}{\z'} = (d^{\z,\#\z'}_1,\dots,d^{\z,\#\z'}_{\rnd+1})$.
\end{enumerate}
\end{algorithm}

\subsubsection{The Final Two-Party Protocol}\label{sec:TwoPartyProtocolFinal}
For $\rnd\in \N$, our two-party, $(2\rnd)$-round, $\frac{O(\log^3 \rnd)}{\rnd}$-fair coin-flipping protocol $\PifTwo_\rnd$, is defined as follows.

\begin{protocol}[$\PifTwo_\rnd = (\PhTwo_0,\PhTwo_1)$]\label{prot:TwoFinal}
\item[Oracle:] an oracle $\orac$ computing $\TwoShareGen_{\frac12} = \TwoShareGen(\cdot,\frac12)$.
\item[Common input:] round parameter $1^\rnd$.
\item [Protocol's description:]~
\begin{enumerate}
\item The two parties use the oracle $\orac$ to compute $\TwoShareGen_{\frac12}(1^\rnd)$. Let $\vect{\sh_0}$ and $\vect{\sh_1}$ be the outputs of $\PhTwo_0$, and $\PhTwo_1$ respectively.\label{prot:TwoFinal:oracleStep}

\item In case the other party aborts, the remaining party outputs a uniform coin.

\item Otherwise, the two parties interact in an execution of $\Pi_\rnd^2 = (\Ptwo_0,\Ptwo_1)$, where $\PhTwo_\z$ plays the role of $\Ptwo_\z$ with private input $\vect{\sh_\z}$.
\end{enumerate}
\end{protocol}

\subsubsection{Main Theorems for Two-Party Protocols}\label{sub:MainThm2Party}
The following theorem states that \cref{prot:TwoFinal} is an almost-optimally fair, two-party coin-flipping protocol, in the $\TwoShareGen_{\frac12}$-hybrid model.

\begin{theorem}\label{thm:2partyIdeal}
For $\rnd \equiv 1 \bmod 4$, the protocol $\PifTwo_\rnd$ is a  $(2 \rnd)$-round, two-party, $O(\frac{\log^3 \rnd}{\rnd})$-fair  coin-flipping protocol against unbounded fail-stop adversaries, in the $\TwoShareGen_{\frac12}$-hybrid model.
\end{theorem}
\cref{thm:2partyIdeal} is proven in \cref{sec:ProofMainThm2} using the the bound on online-binomial games described in \cref{subsection:BG}, but we first use it to deduce an almost-optimal two-party fair coin-flipping protocol, in the real (non-hybrid) model.

\begin{theorem}[Main theorem --- two-party, fair coin flipping]\label{thm:2partyReal}
Assuming protocols for securely computing \OT exist, then for any polynomially bounded, polynomial-time computable, integer function $\rnd$, there exists an $\rnd$-round, $\frac{O(\log^3 \rnd)}{\rnd}$-fair, two-party coin-flipping protocol.
\end{theorem}
\begin{proof}
Define the integer function $\trnd$ by $\trnd(\secParam) = \floor{\rnd(\secParam)/3} -a$, where $a\in \set{0,1,2,3}$ is the value such that $ \floor{\rnd(\secParam)/3} -a \equiv 1 \bmod 4$. Note that both the functionality $\TwoShareGen_{\frac12}(1^{\trnd(\secParam)})$  and the protocol $\PifTwo_{\trnd(\secParam)}$ are polynomial-time computable in $\secParam$, and that $\PifTwo_{\trnd(\secParam)}$ has $2\cdot\trnd(\secParam)$ rounds. Using information-theoretic one-time message authentication codes (\cf \cite{MoranNS09}), the functionality $\TwoShareGen_{\frac12}(1^{\trnd(\secParam)})$ and protocol $\PifTwo_{\trnd(\secParam)}$ can be compiled into functionality $\widetilde{\TwoShareGen_\frac12}(1^{\trnd(\secParam)})$ and protocol $\PiffTwo_{\trnd(\secParam)}$ that maintains  essentially the same efficiency as the original pair, protocol $\PiffTwo_{\trnd(\secParam)}$ maintain the same round complexity, and $\PiffTwo_{\trnd(\secParam)}$ is $\left(\frac{O(\log^3 \trnd(\secParam))}{\trnd(\secParam)}+ \negl(\secParam)\right)$-fair against \emph{arbitrary} unbounded adversaries, in the $\widetilde{\TwoShareGen_{\frac12}}$-hybrid model.

Assuming protocols for securely computing \OT exist, \cref{fact:fHybridMToReal} yields that there exists an $(2\trnd(\secParam) + O(1))$-round, two-party, polynomial-time protocol that is $\left(\frac{O(\log^3 \trnd(\secParam))}{\trnd(\secParam)} + \negl(\secParam)\right)$-fair, in the \emph{standard model}. For large enough $\secParam$, the latter protocol obtains the parameters stated in the theorem (the theorem trivially holds for small values of $\secParam$, \ie smaller than some universal constant)
\end{proof}

\subsubsection{Online Binomial Games}\label{subsection:BG}
Our main tool for proving \cref{thm:2partyIdeal} are bounds on the bias of online-binomial games, defined below, that we prove in \cref{section:BoundsForOnlineBinomialGames}.

In a \textit{online binomial game},  independent random variables $X_1\dots,X_{\rnd}$   are independently sampled, and the game outcome  (or value)   is set to one if $\sum_{i=1}^{\rnd} X_i \geq 0$, and to zero otherwise.  At round $i$, the value of $X_{i-1}$ is given to the  \textit{attacker}, with some auxiliary information  (\ie hint) $\aux_i$ about the value of $X_i$. The attacker can abort, and in this case it gains the expected outcome of the game, conditioned on the values of $X_1,\dots,X_{i-1}$ (but not on the additional information). If it never aborts, it gains the game outcome. The goal of the attacker is to use the hint value to \emph{bias} its expected gain away from the game expected outcome. For instance, in the simple form of the game where  the $X_i$'s are unbiased $\oo$ bits and  $\aux_i = X_i$, the game expected outcome is $1/2$, and it is not hard to see that the expected gain of the attacker who aborts on the first round in which $H_i =X_i =1$, is $1/2 + \Theta(1/\sqrt{m})$.  Namely, such an attacker bias the game by $\Theta(1/\sqrt{m})$.

We are concerned with the weighted version of the above  games in which the  $X_i$'s are \emph{sums} of biased $\oo$ random variables.  Specially, $X_i$ is the sum of $\ml{i} = \rnd -i + 1$ coins. We will also allow the games to have an initial offset: a fixed value (\ie \textit{offset}) $X_0$ is added to the coins sum. Such online binomial games are  useful abstractions to understand the power of fail-stop  adversaries trying to violate the fairness of the coin-flipping protocols considered in this section, and the results presented below play a central role in their security proofs.

We start with formally defining online-binomial games and the bias of such games.

\def\onlineBinomialGame{
	For $\rnd\in \N$, $\ofs\in \Z$, $\eps \in [-1,1]$ and a randomized (hint) function $f$, the {\sf online game $\game_{f,\rnd,\eps,\ofs}$} is the set of the following random variables. Let $Y_0 = X_0 = \ofs$, and for $i\in [\rnd]$,
	\begin{itemize}
		
		\item $X_i$ is sampled according to  $\Beroo{\ml{i},\eps}$.

		\item $Y_{i} = \sum_{j=0}^{i} X_{j}$ and $\aux_i = f(i,Y_i)$. 
		
		\item 	$\veo_i = \eo_i(Y_{i-1},\aux_{i})$ and $\veo_i^- = \eo_i(Y_{i-1})$, for $\eo_i(y) = \pr{Y_\rnd \geq0 \mid Y_{i-1} = y}$ and $\eo_i(y,\paux) = \pr{Y_\rnd \geq 0 \mid Y_{i-1} = y, \aux_i = \paux}$, respectively. 
	\end{itemize}
	Let  $\veo_{\rnd+1} =\veo_{\rnd+1}^-= 1$  if $Y_{\rnd} \geq 0$,  and $\veo_{\rnd+1} =\veo_{\rnd+1}^-= 0$ otherwise.
	
	We let $\game_{f,\rnd,\eps} = \game_{f,\rnd,\eps,0}$.	
}

\begin{definition}[online  binomial games]\label{def:game}
	\onlineBinomialGame
\end{definition}
Namely, $\veo_i$ is the expected output of the game given $Y_{i-1}$, the coins flipped in the first $i-1$ rounds, and the hint $H_i$ about $Y_i$, and $\veo_i^-$  is this value given only $Y_{i-1}$. Consider an attacker who is getting at round $i$ the values of $Y_{i-1}$ and $H_i$, and decides whether to abort and gain $\veo_i^-$, or to continue to next round. If it never aborts, it gains $\veo_{\rnd+1}$. The goal of an attacker is to abort in the round that \emph{maximize} the gap between  $\veo_i$ and $\veo_i^-$. This is  an equivalent task to maximizing the difference between the adversary's  expected gain and the game's expected outcome (which is the expected gain of a never-aborting attacker).

\def\GameBias<#1>{
	Let $\rnd,\ofs,\eps,f$ and $\game = \game_{f,\rnd,\eps,\ofs}= \gameVars$ be as in #1. For an algorithm  $\Bc$, let $I$ be the first round in which $\Bc$ outputs $1$ in the following $m$-round process: in round $i$, algorithm $\Bc$ is getting  input $(i,Y_{i-1},\aux_i)$ and outputs a value. Let $I=m+1$ if $\Bc$ never outputs a one. The {\sf bias $\Bc$ gains in $\game$} is defined by
	$$\bias_\Bc(\game) = \size{\eex{\veo_{I} - \veo_{I}^-}}$$
	The {\sf bias of $\game$} is defined by $\bias(\game) = \max_\Bc \set{ \bias_\Bc(\game)}$, where the maximum is over {\sf all} possible algorithms $\Bc$.	
}

\begin{definition}[game bias]\label{def:GameBias}
	\GameBias<\cref{def:game}>
\end{definition}

We give upper bounds for the security of three different types of online-binomial games that we call \textit{simple}, \textit{hypergeometric} and \textit{vector} games.
The first type of online-binomial game (\ie simple game) is used for proving \cref{thm:2partyIdeal} whereas the other types (\textit{hypergeometric} and \textit{vector} games) are used later in \cref{sec:ThreePartyprotocol} for proving the fairness of our three-party protocol. 

\def\SimpleGameDef{
	For $\rnd\in \N$, $\eps \in [-1,1]$ and a randomized function $f$, the game $\game_{f,\rnd,\eps}$ is called "simple game" if $f$ on input $(i,y)$ outputs $1$ with probability $\eo_{i+1}(y)$ (= $\vBeroo{\ms{i+1},\eps}(-y)$), and zero otherwise.	
}

\def\SimpleGameLemma<#1>{
	Let $\rnd\in \N$, let $\eps \in [-1,1]$ and let $f$ be the randomized function such that $\game_{f,\rnd,\eps}$ is a simple game according to #1. Then $\bias(\game_{f,\rnd,\eps})  \leq \frac{\xi \cdot\log^3 \rnd}{\rnd}$, for some universal constant $\xi$.
}

\begin{definition}[simple game]\label{def:ValueBiomialGame}
	\SimpleGameDef
\end{definition}

Namely, in the simple game, the value of $f$ is sampled according to the expected value of the game.

\begin{lemma}\label{lemma:ValueBiomialGame}
	\SimpleGameLemma<\cref{def:ValueBiomialGame}>
\end{lemma}

\def\HyperRecalls{
Recall that for $n \in \N$, $\ell \in [n]$ and an integer $p\in [-n,n]$, we define the hypergeometric probability distribution $\Hyp{n,p,\ell}$ by $\Hyp{n,p,\ell}(k) \eqdef \ppr{\cI}{\w(\vct_{\cI}) = k}$, where $\cI$ is an $\ell$-size set uniformly chosen from $[n]$ and $\vct \in \oo^n$ with $w(\vct)= p$ (recall that $\w(\vct) = \sum_{j\in [\size{\vct}]} \vct[j]$ and that $\vct_{\cI} = (\vct_{j_1},\ldots,\vct_{j_{\size{\cI}}})$ where $j_1,\ldots,j_{\size{\cI}}$ are the ordered elements of $\cI$) and recall that $\vHyp{n,p,\ell}(k) \eqdef \ppr{x\la \Hyp{n,p,\ell}}{x \geq  k} = \sum_{t=k}^{\ell} \Hyp{n,p,\ell}(t)$.
}

\def\HyperGameDef{
	For $\rnd\in \N$, $\eps \in [-1,1]$, $\const > 0$ and a randomized function $f$, the game $\game_{f,\rnd,\eps}$ is called "$\const$-hypergeometric game" if there exists $p\in [-\rnd,\rnd]$, with $\size{p} \leq \const \cdot \sqrt{\log \rnd \cdot \ms{1}}$, such that $f$ on input $(i,y)$ outputs $1$ with probability $\vHyp{2\cdot \ms{1},p,\ms{i+1}}(-y)$ and zero otherwise.	
}

\def\HyperGameLemma<#1>{
	Let $\rnd\in \N$, let $\eps \in [-1,1]$, let $\const > 0$ and let $f$ be a randomized function such that $\game_{f,\rnd,\eps}$ is an $\const$-hypergeometric game according to #1. Then $\bias(\game_{f,\rnd,\eps}) \leq \frac{\varphi(\const) \cdot \log^3 \rnd}{\rnd}$ for some universal function $\varphi$.
}

\HyperRecalls

\begin{definition}[hypergeometric game]\label{def:ValueHgGame}
	\HyperGameDef
\end{definition}

Namely, in the above game, the value of $f$ is not sampled according to the expected value of the game, as done in the simple game above, but rather from a skewed version of it, obtained by replacing the Binomial distribution used by the game, with an Hypergeometric distribution.

\begin{lemma}\label{lemma:ValueHgGame}
	\HyperGameLemma<\cref{def:ValueHgGame}>
\end{lemma}

\def\VectorRecalls{
Recall that for $n \in \N$ and $\delta \in [0,1]$ we let $\sBias{n}{\delta}$ be the value $\eps \in [-1,1]$ with $\vBeroo{n, \eps}(0) = \delta$.	
}

\def\VectorGameDef{
	For $\rnd\in \N$, $\eps \in [-1,1]$, $\const \in \N$ and a randomized function $f$, the game $\game_{f,\rnd,\eps}$ is called "$\const$-vector game" if $f$ on input $(i,y)$ outputs a string in $\oo^{\const \cdot \ms{1}}$, where each of entries takes the value $1$ with probability $\sBias{\ms{1}}{\delta}$ for $\delta = \eo_{i+1}(y) (= \vBeroo{\ms{i+1},\eps}(-y))$.	
}

\def\VectorGameLemma<#1>{
	Let $\rnd\in \N$, let $\const > 0$ and let $f$ be a randomized function such that $\game_{f,\rnd,\eps=0}$ is an $\const$-vector game according to #1. Then $\bias(\game_{f,\rnd,\eps=0}) \leq \frac{\varphi(\const) \cdot \log^3 \rnd}{\rnd}$ for some universal function $\varphi$.	
}

\VectorRecalls

\begin{definition}[vector game]\label{def:ValueVectorGame}
	\VectorGameDef
\end{definition}

In the last game, the function $f$ outputs a vector (\ie a string), and not a bit as in the previous games. The distribution from which the vector is drawn, however, is very related to the expected value of the game.

\begin{lemma}\label{lemma:ValueVectorGame}
	\VectorGameLemma<\cref{def:ValueVectorGame}>
\end{lemma}

The proof of the above lemmas are given in \cref{sec:Bernoulli,sec:HG,sec:Vector}.
In addition, we make use of the following lemma which asserts that if the expected value of an online-binomial game is almost determined, then there is no much room for an attacker to gain much bias.

\def\DeterminedGameLemma<#1>{
	Let $\game_{f,\rnd,\eps,\ofs} = \gameVars$ be according to #1.
	Assume that $\veo_1^- \notin [\frac1{\rnd^2}, 1- \frac1{\rnd^2}]$, then $\bias(\game_{f,\rnd,\eps,\ofs}) \leq \frac2{\rnd}$. 	
}

\begin{lemma}\label{lemma:DeterminedGames}
	\DeterminedGameLemma<\cref{def:game}>
\end{lemma}

The proof of \cref{lemma:DeterminedGames} is given in \cref{section:GenericApproach}.

\subsubsection{Proving \texorpdfstring{\cref{thm:2partyIdeal}}{Main Theorem in Hybrid Model}}\label{sec:ProofMainThm2}

\begin{proof}[Proof of \cref{thm:2partyIdeal}]
Fix $\rnd \equiv 1 \bmod 4$. By construction, the honest parties in $\PifTwo_\rnd$ always output the same bit, where under the  assumption about $\rnd$, it holds that $\ms{1}$, the total number of coins flipped, is odd. It follows that the common output of a random honest execution of $\PifTwo_\rnd$, is a uniform bit. Namely, protocol $\PifTwo_\rnd$ is correct according to \cref{def:CorrectCT}.

We assume \wlg that if a party aborts in the $i$'th round, it does so by sending the message \Abort, after seeing the other party message of that round.

Let the $(i,j)$'th round in a random execution of $\PifTwo_\rnd$, for $(i,j)\in (\rnd) \times \set{a,b}$, stands for the $j$'th step of the $i$'th loop in the execution. Letting $(0,a)$ being the zero round (\ie before the call to $\TwoShareGen_\frac12$ is made) and $(0,b)$ denote the round where the call to $\TwoShareGen_\frac12$ is made.

Let $\z\in \zo$ and let $\Aadv$ be a fail-stop adversary controlling $\PhTwo_\z$. Let $V$ be $\PhTwo_\z$'s view in a random execution of $\PifTwo_\rnd$. For $\ovr = (i,j)\in (\rnd) \times \set{a,b}$, let $V_\ovr$ be $\ovr$'th round prefix of $V$, and let $V_\ovr^-$ be the value of $V_\ovr$ with the abort message sent in the $\ovr$'th round (if any) removed. Finally, let $I$ be the round in which $\Aadv$ sent the abort message, letting $I=(\rnd,b)$, in case no abort occurred.

In the following we show that
\begin{align}\label{eq:2partyIdeal}
\size{\eex{\val(V_I) - \val(V_{I}^-)}} \le \frac{\xi \cdot \log^3 \rnd}{\rnd},
\end{align}
for some universal (independent of $\rnd$) constant $\xi \ge 0$, where $\val(v)$ is the expected outcome of an honest (non aborting) execution of the parties that do no abort in $v$, conditioned on $v$ (see \cref{def:ViewVal}).

Since \cref{eq:2partyIdeal} holds for any $\rnd \equiv 1 \bmod 4$ and any fail-stop adversary $\Aadv$, protocol $\PifTwo_\rnd$ is $\frac{ \xi \log^3 \rnd}{\rnd}$-biased according to \cref{def:FairCTGameAlt}.
Since, see above, $\PifTwo_\rnd$ is correct according to \cref{def:CorrectCT}, the proof of the theorem follows by \cref{prop:FairCTGameAlt}.

So it is left to prove \cref{eq:2partyIdeal}. Notice that the next rounds shares held by $\PhTwo_\z$ (when playing the role of $\Ptwo_\z$) at the end of round $(i,b)$ (\ie $\cVS{\z}_{i+1,\dots,\rnd}$, $\dVS{0}{\z}_{i+1,\dots,\rnd+1}$ and $\dVS{1}{\z}_{i+1,\dots,\rnd+1}$), are uniformly chosen strings from $\PhTwo_\z$'s point of view. In particular, these shares contains no information about the expected output of the protocol, or the other party's action in case of future aborts. It follows that $\val(V_{0,b}) = \frac12$ (recall that $V_{0,b}$ is $\PhTwo_\z$'s view after getting its part of $\TwoShareGen_\frac12$'s output). We also note that by construction, in case $\PhTwo_\z$ aborts during the call to $\TwoShareGen_\frac12$ (and in this case the honest party  gets no value from the functionality), then the honest party outputs a uniform bit. Namely, $\val(V_{(0,b)}^-)= \frac12$. Hence, the adversary $\Aadv$ gains \emph{nothing} by aborting during the call to $\TwoShareGen_\frac12$, and in the following we assume \wlg that $\Aadv$ only aborts (if any) during the execution of the embedded execution of $\Pi_\rnd^2 = (\Ptwo_0,\Ptwo_1)$.

In the rest of the proof we separately consider the case $I = (\cdot,a)$ and the case $I = (\cdot,b)$. We conclude the proof showing that the first type of aborts might help $\Aadv$ to gain $ \frac{O(\log^3 \rnd)}{\rnd}$ advantage, where the second type give him \emph{nothing}.

Since both steps are symmetric, we assume for concreteness that $\Aadv$ controls $\Ptwo_0$.

\begin{itemize}

 \item[$I = (\cdot,b)$.] In case $I = (i,b)$, the adversary's view $V_I$ contains the value of $(c_1,\dots, c_i)$ sampled by $\TwoShareGen_\frac12$, and some random function of these values, \ie the shares of the next rounds it got from $\TwoShareGen_\frac12$, which are uniform strings from his point of view, and the shares used till this round, which are random function of $(c_1,\dots, c_i)$. Hence, the expected outcome of the protocol given $\Aadv$'s view is $\delta_i \eqdef \vBeroo{\ms{i},0}\left(-\sum_{j=1}^{i} c_j \right)$. By construction, however, the expected outcome of $\Ptwo_1$ in case $\Ptwo_0$ aborts in round $(i,b)$, is also $\delta_i$. Hence, the adversary gains nothing (\ie $\val(V_i) = \val(V_{i}^-)$), by aborting in these steps.

 \item[$I = (\cdot,a)$.] Since $\Aadv$ gains nothing by aborting at \Stepref{step:TwoParty:round:2} of the loop, we assume \wlg that $\Aadv$ only aborts at \Stepref{step:TwoParty:round:1} of the loop, and the proof  by the next claim (proven below).
 \begin{claim}\label{claim:2partyConnectionToSimpleGame}
 	Assuming $\Aadv$ only aborts at \Stepref{step:TwoParty:round:1} of the loop in $\Pitwo_\rnd$, then
 	\begin{align*}
 	\size{\eex{\val(V_I) - \val(V_I)^-}} \leq \frac{\xi \cdot \log^3 \rnd}{\rnd},
 	\end{align*}
 	for some universal (independent of $\rnd$) constant $\xi \ge 0$.
 \end{claim}
 
  \cref{claim:2partyConnectionToSimpleGame} yields that  the overall bias $\Ac$ gains, which equals to $\size{\eex{\val(V_I) - \val(V_I)^-}}$, is bounded by $ \frac{\xi \cdot \log^3 \rnd}{\rnd}$, for some universal constant $\xi\geq0$.

\end{itemize}
\end{proof}

\paragraph{Proving \cref{claim:2partyConnectionToSimpleGame}}
We prove the claim via reduction to online-binomial game, described in \cref{subsection:BG}.
The proof immediately follows from the following claim and \cref{lemma:ValueBiomialGame}.

\begin{claim}\label{claim:2partyConnectionToSimpleGameGen}
	Assuming $\Aadv$ only aborts at \Stepref{step:TwoParty:round:1} of the loop in $\Pitwo_\rnd$,  then
	\begin{align*}
	\size{\eex{\val(V_I) - \val(V_I)^-}} \leq \bias(\game_{f,\rnd,0}),
	\end{align*}
	where $\game_{f,\rnd,0}$ is a simple online-binomial game according to \cref{def:ValueBiomialGame} and $\bias(\cdot)$ is according to \cref{def:GameBias}.
\end{claim}

\begin{proof}[Proof of \cref{claim:2partyConnectionToSimpleGameGen}]
	\remove{
	Informally (see \cref{section:BoundsForOnlineGames} for formal definition), 
	in an $\rnd$-round online binomial game $\game_{f,\rnd,\eps}$, binomial random variables $X_1\dots,X_{\rnd}$ over $\oo$ are independently sampled ($X_i$ is sampled according $\Beroo{\ml{i},\eps}$), and the value of the game is set to one if $\sum_{i=0}^{\rnd} X_i \geq 0$, and to zero otherwise (where $X_0 = 0$). At the $i$'th round of the game, the value of $X_{i-1}$ is exposed to an (unbounded) attacker, who is also getting some auxiliary information $H_i$ about the value of $X_i$. The attacker can abort, and in this case it gets the expected value of the game, conditioned on the values of $X_0,X_1,\dots,X_{i-1}$ (but not on the additional information). If no abort occurred, the attacker is getting the final value of the game. The goal of the attacker is to bias the value it gets \emph{away} from the expected value of the game.
	}
	
	\remove{the $\rnd$-round simple binomial game $\game_{f,\rnd,0}$ (as defined in \cref{def:ValueBiomialGame}) is an online game where in each round $i \in [\rnd]$, $\ml{i}$ coins are sampled (denote their sum by $X_i$), $X_{i-1}$ is exposed to an unbounded attacker who is also getting some hint $H_i$ about the value of $X_i$. In the simple game, the hint $H_i$ is an $\zo$ bit which is a sampled according to the expected sign of $\sum_{i=1}^{\rnd}X_i$. The attacker can abort, and in this case it gets the expected value of the game, conditioned on the values of $X_0,X_1,\dots,X_{i-1}$ (but not on the additional information). If no abort occurred, the attacker is getting the final value of the game. The goal of the attacker is to bias the value it gets \emph{away} from the expected value of the game.}
	
	We prove that an attacker for the coin-flipping protocol of the type considered in this claim, \ie one that only aborts in \Stepref{step:TwoParty:round:1} of the loop in $\Pitwo_\rnd$ and achieves bias $\alpha$, yields a player for the simple binomial game (described in \cref{subsection:BG}) that achieves the same bias. Thus, the bound on the former attacker follows from the bound on the latter one.
	
	Let $\game = \game_{f,\rnd,0} = \gameVars$  be the simple binomial game  as defined in \cref{def:ValueBiomialGame} and recall that in this game, $f$ is the randomized function that on input $(i,y)$ outputs $1$ with probability $\vBeroo{\ms{i+1},\eps}(-y)$, and zero otherwise. Now, consider the player $\Bc$ for $\game$ that emulates interaction with $\Aadv$ in $\Pitwo_\rnd = (\Ptwo_0,\Ptwo_1)$, where $\Aadv$ controls $\Ptwo_0$. The emulation goes as follows:
	
	\begin{algorithm}[Player $\Bc$ for the simple game]\label{claim:2partyConnectionToSimpleGame:attacker}
		\item[Operation:]~
		\begin{enumerate}
			
			\item Choose uniformly $\cVS{0} \la \zo^{\rnd\times \tp{\rnd}}$ and $\dVS{0}{0},\dVS{1}{0} \la \zo^{\rnd+1}$ and set $\vect{\sh_0} = \{\cVS{0}, \dVS{0}{0} ,\dVS{1}{0}\}$ as input for $\Ptwo_0$ in $\Pitwo_\rnd$.\label{claim:2partyConnectionToSimpleGame:attacker:firstStep}
			
			
			\item For $i=1$ to $\rnd$:
			\begin{enumerate}
				
				\item Receive input $(i,Y_{i-1},\aux_i)$ from the game $\game$.
				
				\item If $i>1$, emulate a sending of $\cVS{1}[i-1]$ from $\Ptwo_1$ to $\Ptwo_0$ at step $(i-1,b)$ of $\Pitwo_\rnd$,  where $\cVS{1}[i-1] \eqdef X_{i-1} \xor \cVS{0}[i-1]$ (recall that $X_{i-1} \eqdef Y_{i-1} - Y_{i-2}$ according to \cref{def:game}).
				
				\item  Emulate a sending of $\dVS{0}{1}[i]$ from $\Ptwo_1$ to $\Ptwo_0$ at step $(i,a)$ of $\Pitwo_\rnd$, where $\dVS{0}{1}[i] \eqdef \aux_i \xor \dVS{0}{0}[i]$.  If $\Aadv$ aborts at this step, output $1$ (abort at round $i$). Otherwise, output $0$ (continue to next round).
				
			\end{enumerate}
			
		\end{enumerate}
	\end{algorithm}
	
	Note that by the definition of $f$, the set of $\game$'s variables $(X_1, \ldots, X_\rnd, \aux_1, \ldots, \aux_\rnd)$ has the same distribution as the set of $\Pitwo_\rnd$'s variables $(c_1, \ldots,c_\rnd, d^0_1, \ldots, d^0_\rnd)$ where the parties' inputs are sampled according to $\TwoShareGen_{\frac12}$. Therefore, $\Aadv$'s view in the  emulation done by $\Bc$, is distributed exactly the same as its view when interacting with honest $\Ptwo_1$ in $\Pitwo_\rnd$;  in both cases, the only meaningful information it gets is the revelation of these values. Since $\Bc$ aborts at round $i$ of $\game$ (i.e outputs $1$) iff $\Aadv$ aborts at round $(i,a)$ of $\Pitwo_\rnd$, it follows that
	
	\begin{align*}
	\size{\eex{\val(V_I) - \val(V_I)^-}} = \bias_\Bc(\game) \leq \bias(\game),
	\end{align*}
	as required.
\end{proof}

\subsection{Three-Party Protocol}\label{sec:ThreePartyprotocol}
As done in \cref{sec:TwoPartyprotocol}, we start with defining a three-party coin-flipping protocol whose parties get (correlated) shares as input,  then describe the functionality for generating these shares, and finally explain how to combine the two into a (no input) coin-flipping protocol.

\subsubsection{The Basic Three-Party Protocol}\label{sec:ThreePartyProtocolBasic}

\begin{protocol}[$\Pithree_\rnd = (\Pthree_0,\Pthree_1,\Pthree_2)$]\label{Protocol:Three}
\item[Common input:] round parameter $1^\rnd$.

\item[$\Pthree_\z$'s input:] $\cVS{\z} \in \zo^{\rnd\times \tp{\rnd}}$ and $\Mat{D^{(\z',\z''),\#\z}}\in \zo^{\rnd \times (\rnd\cdot \tp{\rnd} + 2(\rnd+1))}$, for all $\z' \neq \z''\in \zot$.

\item [Protocol's description:]~
\begin{enumerate}
\item For $i=1$ to $\rnd$:
\begin{enumerate}
\item For all $\z_s, \z_r,\z_o\in \zot$ with $\z_r \notin \set{\z_s,\z_o}$, party $\Pthree_{\z_s}$ sends $\Mat{D^{(\z_r,\z_o),\#\z_s}}[i]$ to $\Pthree_{\z_r}$.\label{step:ThreeParty:round:1}

\item[$\bullet$] For all $\z\neq \z' \in \zot$, party $\Pthree_\z$ sets $\vect{d^{(\z,\z')}_i} = \bigoplus_{\z'' \in \zot} \Mat{D^{(\z,\z'),\#\z''}}[i]$.

\item For all $\z \in \zot$, party $\Pthree_\z$ sends $\cVS{\z}[i]$ to the other parties.\label{step:ThreeParty:round:2}

\item[$\bullet$] All parties set $c_i = \cVS{0}[i] \xor \cVS{1}[i] \xor \cVS{2}[i]$.
\end{enumerate}
\item[Output:] All parties output 1 if $\sum_{i=1}^\rnd c_i \geq 0$, and 0 otherwise.

\item [Abort:] ~
\begin{itemize}
 \item[One party aborts:] Let $\z < \z' \in \zot$ be the indices of the remaining parties, and let $i\in [\rnd]$ be the {\sf maximal} $i\in [\rnd]$ for which both $\Pthree_\z$ and $\Pthree_{\z'}$ have reconstructed $\vect{d^{(\z,\z')}_i}$ and $\vect{d^{(\z',\z)}_i}$, respectively. Set $i$ to $\perp$ in case no such index exists. To decide on a common output, $\Pthree_\z$ and $\Pthree_{\z'}$ interact in the following two-party protocol.

 \begin{itemize}
 \item [$i= \perp$:] $\Pthree_\z$ and $\Pthree_{\z'}$ interact in $\PifTwo_\rnd$.
 \item [$i\neq \perp$:] $\Pthree_\z$ and $\Pthree_{\z'}$ interact in $\Pitwo_{\rnd} = (\Ptwo_0,\Ptwo_1)$, where $\Pthree_\z$ with input $\vect{d^{(\z,\z')}_i}$ plays the role of $\Ptwo_0$, and $\Pthree_{\z'}$ with input $\vect{d^{(\z',\z)}_i}$ plays the role of $\Ptwo_1$.
 \end{itemize}

 \item[Two parties abort (in the same round):] Let $\Pthree_\z$ be the remaining party and for an arbitrary $\z'\neq \z \in \zot$, let $i\in [\rnd]$ be the {\sf maximal} index for which $\Pthree_\z$ have reconstructed $\vect{d^{(\z',\z)}_i}$, set to $\perp$ in case no such index exists.
 \begin{itemize}
 \item [$i= \perp$:] $\Pthree_\z$ outputs a uniform bit.

 \item [$i\neq \perp$:] The remaining party $\Pthree_\z$ acts as if $\Pthree_{\z'}$ has only aborted at the very beginning of the following two-party protocol: $\Pthree_\z$ ``interact" with $\Pthree_{\z'}$ in $(\Ptwo_0,\Ptwo_1)$, where $\Pthree_\z$ with input $\vect{d^{(\z',\z)}_i}$ plays the role of $\Ptwo_0$ in case $\z< \z'$ and as $\Ptwo_1$ otherwise.\footnote{The latter protocol is well defined, since when aborting right at the beginning, $\Pthree_{\z'}$ does not send any message.}
 \end{itemize}
\end{itemize}
\end{enumerate}
\end{protocol}
Namely, at Step $(a)$ the parties help each other to reconstruct inputs for the two-party protocol $\Pitwo_\rnd$. More specifically, each pair of parties reconstructs two inputs (shares)  for  an execution of $\Pitwo_\rnd$, one input for each party in the pair. In case a party aborts, the remaining parties use the above inputs for interacting in $\Pitwo_\rnd$. In Step $(b)$ the parties help each other to reconstruct the round coins (\ie $c_i$).

Note that the above protocol has $4\rnd$ rounds (in case one party abort at the end of the outer three-party protocols). While it is possible to reduce this number to $2\rnd$ (to match the two-party case), we chose to present the somewhat simpler protocol given above.

\subsubsection{Hiding Two-Party Shares Generator}\label{sec:HidingTwoShareGenerator}
As mentioned in \cref{sec:Technique}, we construct a \emph{hiding} variant $\HidTwoShareGen$ of the two-party share-generating function $\TwoShareGen$. The construction is done by modifying the way the defense values (given to the parties in the three-party protocol) are sampled. On input $\delta\in[0,1]$, $\HidTwoShareGen$ first draws $\Theta(\rnd^2)$ independent samples from $\Beroo{\eps}$ for $\eps = \sBias{\ms{1}}{\delta}$, and then uses these samples via a simple derandomization technique for drawing the $\Theta(\rnd)$ defense values given in the three-party protocol.

Roughly, these $\Theta(\rnd^2)$ values sampled by $\HidTwoShareGen$ give about the same information as a \emph{constant} number of independent samples from $\Berzo{\delta}$ would, unlike the non-hiding $\TwoShareGen$ which gives about the same information as $\Theta(\rnd)$ independent samples from $\Berzo{\delta}$. As mentioned in \cref{sec:Technique}, these $\Theta(\rnd)$ samples can be used to bias the outcome by $\Omega(\frac1{\sqrt{\rnd}})$, which makes $\HidTwoShareGen$ crucial for the fairness of our protocol.

Recall that for a vector $\vct \in \oo^\ast$ we let $\w(\vct) \eqdef \sum_{i \in [\size{\cI}]}\vct_i$, and given a set of indexes $\cI \subseteq [\size{\vct}]$, we let $\vct_{\cI} = (\vct_{i_1},\ldots,\vct_{i_{\size{\cI}}})$ where $i_1,\ldots,i_{\size{\cI}}$ are the ordered elements of $\cI$.

\begin{algorithm}[$\HidTwoShareGen$]\label{alg:HidingTwoShareGenerator}
\item[Input:] Round parameter $1^\rnd$ and $\delta \in [0,1]$.

\item[Operation:]~
\begin{enumerate}

\item Let $\eps = \sBias{\ms{1}}{\delta}$.\label{step:bias}

\item  For $\z\in \zo$: sample a random vector $\vct^\z \in \oo^{2\cdot \ms{1}}$, where each coordinate is independently drawn from $\Beroo{\eps}$.\label{alg:HidingTwoShareGenerator:sampleSets}

\item For $\z\in \zo$: sample a random $(\ms{1})$-size subset $\cI^\z \subset [2\cdot \ms{1}]$, and set $d^{\z,\#\z}_{\rnd+1}$ to one if $\w(\vct^\z_{\cI^\z}) \geq 0$, and to zero otherwise. Set $d^{\z,\#\overline{\z}}_{\rnd+1}$ arbitrarily.

\item For $i=1$ to $\rnd$:
\begin{enumerate}

\item Sample $c_i \la \Beroo{\ml{i} ,\eps}$.\label{alg:HidingTwoShareGenerator:sampleCi}

\item Sample $c_i^{\#0} \la \zo^{\tp{\rnd}}$, and set $c_i^{\#1} = c_i \xor c_i^{\#0}$.

\item For  $\z\in \zo$:
\begin{enumerate}
 \item Sample a random $(\ms{i+1})$-size subset $\cI^\z \subset [2\cdot \ms{1}]$, and set $d^{\z}_i$ to one if\\ $\sum_{j=1}^i c_j + \w(\vct^\z_{\cI^\z}) \geq 0$, and to zero otherwise.\label{step:derandomization}

\item Sample $d^{\z,\#0}_i \la \zo$, and set $d^{\z,\#1}_i = d^{\z}_i \xor d^{\z,\#0}_i$.
\end{enumerate}
\end{enumerate}

\item Output $(\sVS{0},\sVS{1})$, where $\sVS{\z} = (\cVS{\z},\dVS{0}{\z}, \dVS{1}{\z}$), for $\cVS{\z} = (c^{\#\z}_1,\dots,c^{\#\z}_\rnd)$ and $\dVS{\z}{\z'} = (d^{\z,\#\z'}_1,\dots,d^{\z,\#\z'}_{\rnd+1})$.
 \end{enumerate}
\end{algorithm}
Namely, rather then sampling the defense values in \Stepref{step:derandomization} \emph{independently} (as done in its non-hiding variant  $\TwoShareGen$), the defense values used by $\HidTwoShareGen$ in the different rounds, are correlated via the vectors $\vct^0$ and $\vct^1$ ($\vct^\z$ is used for the defense values of the party $\Ptwo_\z$). Note, however, that each round defense value on \emph{its own}, has exactly the same distribution as in $\TwoShareGen$. 

\subsubsection{Three-Party Shares Generator}
Using the above two-party shares generator, our three-party shares generator is defined as follows.

\begin{algorithm}[$\ThreeShareGen$]\label{alg:ThreeShareGenerator}
\item[Input:] round parameter $1^\rnd$.
\item[Operation:]~
\begin{enumerate}

\item For $i=1$ to $\rnd$:
\begin{enumerate}
\item Sample $c_i \la \Beroo{\ml{i} ,0}$.
\item Sample $(c_i^{\#0}, c_i^{\#1}) \la (\zo^{\tp{\rnd}})^2$, and set $c_i^{\#2} = c_i \xor c_i^{\#0} \xor c_i^{\#1}$.

\item Let $\delta_i = \vBeroo{\ms{i+1},0}(-\sum_{j=1}^i c_j)$.

\item For $\z< \z'\in \zot$:
\begin{enumerate}

\item Sample $(\vect{s^{(\z,\z')}_i},\vect{s^{(\z',\z)}_i}) \la \HidTwoShareGen(1^\rnd,\delta_i)$.\label{step:hidTwoShareGen}

\item Sample $(\vect{\sh^{(\z,\z'),\#0}_i},\vect{\sh^{(\z,\z'),\#1}_i}, \vect{\sh^{(\z',\z),\#0}_i},\vect{\sh^{(\z',\z),\#1}_i}) \la (\zo^{\rnd \cdot \tp{\rnd} + 2(\rnd + 1)})^4$.

 Set $\vect{\sh^{(\z,\z'),\#2}_i} = \vect{\sh^{(\z,\z')}_i} \xor \vect{\sh^{(\z,\z'),\#0}_i} \xor \vect{\sh^{(\z,\z'),\#1}_i}$ and $\vect{\sh^{(\z',\z),\#2}_i} = \vect{\sh^{(\z',\z)}_i} \xor \vect{\sh^{(\z',\z),\#0}_i} \xor \vect{\sh^{(\z',\z),\#1}_i}$.
\end{enumerate}
\end{enumerate}

 \item Output $(\SSS{0},\SSS{1},\SSS{2}$), where $\SSS{\z}= (\cVS{\z},\Mat{D^{(0,1),\#\z}},\Mat{D^{(0,2),\#\z}},\Mat{D^{(1,0),\#\z}},\Mat{D^{(1,2),\#\z}}, \Mat{D^{(2,0),\#\z}},\Mat{D^{(2,1),\#\z}})$, for $\cVS{\z} = (c^{\#\z}_1,\dots,c^{\#\z}_\rnd)$ and $\Mat{D^{(\z',\z''),\#\z}} = (\vect{\sh^{(\z',\z''),\#\z}_1},\dots,\vect{\sh^{(\z',\z''),\#\z}_\rnd})$.
\end{enumerate}
\end{algorithm}

\subsubsection{The Final Three-Party Protocol}\label{sec:ThreePartyProtocolFinal}
For $\rnd\in \N$, our three-party, $3\rnd$-round, $\frac{O(\log^3 \rnd)}{\rnd}$-fair coin-flipping protocol $\Pi^3_\rnd$ is defined as follows.
\begin{protocol}[$\PifThree_\rnd = (\PhThree_0,\PhThree_1,\PhThree_2)$]\label{prot:ThreeFinal}

\item[Input:] round parameter $1^\rnd$.

\item[Oracle:] Oracle $\orac_2$ and $\orac_3$ for computing $\TwoShareGen_\frac12$ and $\ThreeShareGen$ respectively.
\item [Protocol's description:]~
\begin{enumerate}
\item The three parties using the oracle $\orac_3$ to securely compute $\ThreeShareGen(1^\rnd)$. Let $\Sh_0$, $\Sh_1$, and $\Sh_2$ be the outputs obtained by $\PhThree_0$, $\PhThree_1$ and $\PhThree_2$ respectively.\label{prot:ThreeFinal:oracleStep}

\item In case one party aborts, the remaining parties use oracle $\orac_2$ to interact in $\Pif_\rnd^2$ (\cref{prot:TwoFinal}).

\item In case two parties aborts, the remaining party outputs a uniform bit.

\item Otherwise, the three parties interact in $\Pithree_\rnd = (\Pthree_0,\Pthree_1,\Pthree_2)$, where $\PhThree_\z$ plays the role of $\Pthree_\z$ with private input $\Sh_\z$.
\end{enumerate}
\end{protocol}

\subsubsection{Main Theorems for Three-Party Protocols}\label{sub:MainThm3Party}

\begin{theorem}\label{thm:3partyIdeal}
For $\rnd \equiv 1 \bmod 4$,  protocol $\PifThree_\rnd$  is a  $(4 \rnd)$-round, three-party, $O(\frac{\log^3 \rnd}{\rnd})$-fair,  coin-flipping protocol, against unbounded fail-stop adversaries, in the $(\TwoShareGen_\frac12,\ThreeShareGen)$-hybrid model.
\end{theorem}

As in the two-party case, we deduce the following result.
\begin{theorem}[Main theorem --- three-party, fair coin flipping]\label{thm:3partyReal}
Assuming protocols for securely computing \OT exist, then for any polynomially bounded, polynomial-time computable, integer function $\rnd$, there exists an $\rnd$-round, $\frac{O(\log^3 \rnd)}{\rnd}$-fair, three-party coin-flipping protocol.
\end{theorem}
\begin{proof}
The only issue one should take care of in the current proof, which does not occur in the proof of \cref{thm:2partyReal}, is that the function $\HidTwoShareGen$, called by $\ThreeShareGen$, and in particular  calculating the value of $\sBias{\ms{1}}{\delta}$, is not necessarily polynomial-time computable. (This issue was not a problem in the proof of \cref{thm:2partyReal}, since $\sBias{\ms{1}}{\cdot}$ is only called there with $\delta= \frac12$, and in this case its output is simply $0$).
Note, however, that after $\HidTwoShareGen$ calculates $\eps = \sBias{\ms{1}}{\delta}$, it merely uses $\eps$ for sampling $5\cdot\ms{1}$ independent samples from $\Beroo{\eps}$. Hence, one can efficiently estimate $\eps$ by a value $\widetilde{\eps}$ (via binary search),\footnote{The binary search on a value $\widetilde{\eps} \in [-1,1]$ is done by sampling $x \la \Beroo{\rnd^{11}, \widetilde{\delta}-\frac12}$ (for $\widetilde{\delta} = \vBeroo{\ms{1}, \widetilde{\eps}}(0)$) and guessing if $\widetilde{\eps}$ is bigger, smaller or close enough to $\eps$ according to whether $x$ is bigger, smaller or inside the range $(\delta-\frac12)\cdot \rnd^{11} \pm \rnd^6$, respectively. By Hoeffding Inequality (\cref{claim:Hoeffding}), all the guesses are good with probability $1-\negl(\rnd)$ and the binary search outputs $\widetilde{\eps}$ such that $\size{\eps - \widetilde{\eps}} \leq \size{\delta - \widetilde{\delta}} \leq \frac{1}{\rnd^5}$.} such that $\abs{\eps-\widetilde{\eps}} < \frac1{\rnd^5}$, which yields that the statistical distance of $5\cdot \ms{1}$ independent samples from $\Beroo{\eps}$, from $5\cdot \ms{1}$ independent samples from $\Beroo{\widetilde{\eps}}$, is bounded by $\frac1{\rnd^2}$. It follows that there exists a polynomial-time computable function $\widetilde{\ThreeShareGen}$, such that protocol $\Pif^3_\rnd$ given in \cref{prot:ThreeFinal}, is a $(4\rnd)$-round, $\left(\frac{O(\log^3 \rnd)}{\rnd} + \frac{\rnd}{\rnd^2}\right)$-fair, three-party coin-flipping protocol, against unbounded fail-stop adversaries, in the $(\TwoShareGen_\frac12,\widetilde{\ThreeShareGen})$-hybrid model. The proof continues like the proof of \cref{thm:2partyReal}.

\end{proof}

\mparagraph{Proving \cref{thm:3partyIdeal}}
We advise to reader to read first the proof of \cref{thm:2partyIdeal}.

\begin{proof}[Proof of \cref{thm:3partyIdeal}]
Fix $\rnd \equiv 1 \bmod 4$. As in the proof of \cref{thm:2partyIdeal}, it holds that protocol $\PifThree_\rnd$ is correct according to \cref{def:CorrectCT}. Also as in the proof of \cref{thm:2partyIdeal}, we assume \wlg that if a party aborts in the $i$'th round, it does so by sending the message \Abort, and  after seeing the other parties' message of that round.

Let the $(p,i,j)$'th round in a random execution of $\PifThree_\rnd$, for $(p,i,j)\in \set{\outt,\inn}\times (\rnd) \times \set{a,b}$, stands for the $j$'th step of the $i$'th loop in the execution of the $\PifThree_\rnd$, where $p = \outt$  means that this is a step of the outer execution of $\PifThree_\rnd$, and $p = \inn$  means that this is a step of the inner execution of $\Pitwo_\rnd$ (whose execution starts in case a party aborts). We let $(\outt,0,a)$ be the zero round, let $(\outt,0,b)$ denote the round where the  call to $\ThreeShareGen$ is made, and let  $(\inn,0,b)$ be the zero round in the  inner execution of $\Pitwo_\rnd$.

Let $\Aadv$ be a fail-stop adversary controlling the parties $\set{\PhThree_\z}_{\z\in \cC}$, for some $\cC \subsetneq \zot$. Let $V$ be the view of $\Aadv$ in a random execution of $\Pif^3_\rnd$, in which $\Ac$ controls the parties indexed by $\cC$. For $\ovr \in\set{\outt,\inn}\times (\rnd) \times \set{a,b}$, let $V_\ovr$ be the $\ovr$'th round prefix of $V$, and let $V_\ovr^-$ be the value of $V_\ovr$ with the $\ovr$'th round abort messages (if any) removed. Finally, let $I_1$ and $I_2$ be the rounds in which $\Aadv$ sent an abort message, letting $I_k=(\outt,\rnd,b)$ in case less than $k$ aborts happen. In the following we show that for both $k\in \set{1,2}$, it holds that
\begin{align}\label{eq:3partyIdeal}
\size{\eex{\val(V_{I_k}) - \val(V_{I_k}^-)}} \leq \frac{ \xi \cdot \log^3 \rnd}{\rnd}
\end{align}
for some universal (independent of $\rnd$) constant $\xi \ge 0$, where $\val(v)$ is the expected outcome of an honest (non aborting) execution of the parties that do not abort in $v$, conditioned on $v$ (see \cref{def:ViewVal}). Since \cref{eq:3partyIdeal}  holds for any $\rnd \equiv 1 \bmod 4$ and any fail-stop adversary $\Aadv$, the proof of the theorem follows by \cref{prop:FairCTGameAlt}.

So it is left to prove \cref{eq:3partyIdeal}. By construction, the only non-redundant information in $\Aadv$'s view  at the end of round $(i,b)$ is  the coins constructed by the parties at the end of this round. In particular, it holds that $\val(V_{\outt,0,b}) = \frac12$. By construction, in case two parties abort during the call to $\ThreeShareGen$, the remaining party outputs one with probability $\frac12$. In case one party aborts, the remaining parties interact in the unbiased protocol $\PifTwo_\rnd$. In both cases, it holds that $\val(V_{\outt,0,b}^-) = \frac12$. Taken the  security of protocol $\PifTwo_\rnd$ (proven in \cref{thm:2partyIdeal}) into account, we can assume \wlg that $\Aadv$ only aborts (if any) during the embedded execution of $\Pi_\rnd^3 = (\Pthree_0,\Pthree_1,\Pthree_2)$.


In the rest of the proof we separately bound the case $k=1$ and $k=2$. Note that $I_1$ is of the form $(\outt,\cdot,\cdot)$, where $I_2$, unless equals $(\outt,\rnd,b)$, is of the form $(\inn,\cdot,\cdot)$ (\ie the first abort is in the outer three-party protocol, and the second, if any, is in the inner  two-party protocol).

\mparagraph{First abort}
We separately consider the case $I_1 = (\outt,\cdot,a)$ and the case $I_1 =  (\outt,\cdot,b)$. We conclude the proof, of this part, showing that the first type of aborts might help $\Aadv$ to gain $ \frac{O(\log^3 \rnd)}{\rnd}$ advantage, where the second type give him \emph{nothing}.

\mparagraph{Case $I_1 = (\outt,\cdot,b)$} Assume that $I_1 = (\outt,i,b)$ for some $i\in [\rnd]$. The view of $\Aadv$ at this point (\ie $V_{I_1}$) contains the value of $(c_1,\dots, c_i)$, and some random function of these values, \ie the shares of the next rounds it got from $\ThreeShareGen$, which are uniform strings from his point of view, and the shares used till this round, which are random function of $(c_1,\dots, c_i)$. Hence, $\val(V_{I_1}^-) = \delta_i \eqdef \vBeroo{\ms{i+1},0}\left(-\sum_{j=1}^i c_j \right)$. By construction, $\delta_i$ is also the expected outcome of the remaining parties, in case an abort message was sent in this round. Namely, $\val(V_{I_1}) =\delta_i$. Hence, the adversary gains nothing (\ie $\val(V_{I_1}) = \val(V_{I_1}^-)$), by aborting in this round.

\mparagraph{Case $I_1 = (\outt,\cdot,a)$}

Since $\Aadv$'s first abort at \Stepref{step:ThreeParty:round:2} of the loop gains nothing, we assume \wlg that $\Aadv$'s first abort is at \Stepref{step:ThreeParty:round:1} of the loop. We use the following claim, proof given below,  for bounding the effect of such abort.

\begin{claim}\label{claim:3partyConnectionToVectorGame}
	Assuming $\Aadv$'s first abort is only at \Stepref{step:ThreeParty:round:1} of the loop in $\Pithree_\rnd$, it follows that
	\begin{align*}
	\size{\eex{\val(V_{I_1}) - \val(V_{I_1})^-}} \leq \frac{\xi \cdot \log^3 \rnd}{\rnd},
	\end{align*}
	for some universal (independent of $\rnd$) constant $\xi \ge 0$.
\end{claim}

\cref{claim:3partyConnectionToVectorGame} yields that  the overall bias $\Ac$'s first abort gains, which equals to $\size{\eex{\val(V_{I_1}) - \val(V_{I_1})^-}}$, is bounded by $ \frac{\xi \cdot \log^3 \rnd}{\rnd}$, for some universal constant $\xi\geq0$.

\mparagraph{Second abort}
We assume \wlg that $I_2 = (\inn,\cdot,\cdot)$ (\ie a second abort occurred). Assume $I_1 = (\out,j,\cdot)$ and let $\eps$ be the value of $\sBias{\rnd}{\delta_j}$ computed by $\HidTwoShareGen$ on input $\delta_j$, for generating the shares of the two-party  protocol. 

The following proof is similar to analysis of the two-party protocol $\Pitwo_\rnd$, done in the proof of \cref{thm:2partyIdeal}, but with few differences. We separately consider the case $I_2 = (\inn,\cdot,a)$ and the case $I_2 = (\inn,\cdot,b)$. We conclude the proof showing that the first type of aborts might help $\Aadv$ to gain $ \frac{O(\log^3 \rnd)}{\rnd}$ advantage, where the second type give him \emph{nothing}.

 \mparagraph{Case $I_2 = (\inn,\cdot,b)$} Assume  $I_2 = (\inn,i,b)$. The view of $\Aadv$ at this point (\ie $V_{I_2}$) contains the value of $(c_1,\dots, c_i)$ sampled by $\HidTwoShareGen$, and some random function of these values, \ie the shares of the next rounds it got from $\HidTwoShareGen$, which are uniform strings from his point of view, and the shares used till this round, which are random function of $(c_1,\dots, c_i)$. Hence,
 $$\val(V_{I_2}^-) = \delta_i \eqdef \vBeroo{\ms{i+1},\eps}\left(-\sum_{j=1}^i c_j \right)$$

 By construction, $\delta_i$ is also the expected outcome of the remaining party, in case an abort message was sent in this round. It follows that $\val(V_{I_2}) =\delta_i$, and the adversary gains nothing (\ie $\val(V_{I_2}) = \val(V_{I_2}^-)$), by aborting in this round.

 The above point needs is somewhat subtle and deserves some justification. Note that the output of the remaining party $\Ptwo_\z$ is not directly sampled from $\Berzo{\delta_i}$, as in the case of protocol $\PifTwo_\rnd$ considered in the proof of \cref{thm:2partyIdeal}. Rather, a $(2 \cdot\ms{1})$-size vector $\vct^\z$ is sampled according to $\Beroo{\eps}$ (see \cref{alg:HidingTwoShareGenerator}). Then, the output of the remaining party is set to one if $\sum_{j=1}^i c_j + \w(\vct^\z_{\cI^\z}) \geq 0$, and to zero otherwise, where $\cI^\z$ is random $(\ms{i+1})$-size subset of $[2 \cdot\ms{1}]$. Yet, by construction, since $V_{I_2}$ does not contain any extra information about the remaining party's vector $\vct^\z$, it follows that 
 \begin{align}
 \val(V_{I_2}) &= \ppr{\vct^\z \la (\Beroo{\eps})^{2 \cdot\ms{1}}, \cI^\z \subset [2 \cdot\ms{1}]}{\sum_{j=1}^i c_j + \w(\vct^\z_{\cI^\z}) \geq 0}\\
 &= \ppr{u \la (\Beroo{\eps})^{\ms{i+1}}}{\sum_{j=1}^i c_j + \w(u) \geq 0}\nonumber\\
 &= \delta_i.\nonumber
 \end{align}

\mparagraph{Case $I_2 = (\inn,\cdot,a)$} 
Since $\Aadv$'s second abort at \Stepref{step:TwoParty:round:2} of the loop gains nothing, we assume \wlg that $\Aadv$'s second abort is at \Stepref{step:TwoParty:round:1} of the loop. We use the following claim, proof given below,  for bounding the effect of such abort.
\begin{claim}\label{claim:3partyConnectionToHyperGame}
	Assuming $\Aadv$'s second abort is only at \Stepref{step:TwoParty:round:1} of the loop in $\Pitwo_\rnd$, it follows that
	\begin{align*}
	\size{\eex{\val(V_{I_2}) - \val(V_{I_2})^-}} \leq \frac{\xi \cdot \log^3 \rnd}{\rnd},
	\end{align*}
	for some universal (independent of $\rnd$) constant $\xi \ge 0$.
\end{claim}

\cref{claim:3partyConnectionToHyperGame} yields that  the overall bias $\Ac$'s second abort gains, which equals to $\size{\eex{\val(V_{I_2}) - \val(V_{I_2})^-}}$, is bounded by $ \frac{\xi \cdot \log^3 \rnd}{\rnd}$, for some universal constant $\xi\geq0$.
\end{proof}

It is left to prove \cref{claim:3partyConnectionToHyperGame,claim:3partyConnectionToVectorGame}. Similarly to the proof of \cref{claim:2partyConnectionToSimpleGame}, we prove these claims via reductions to online-binomial games, described in \cref{subsection:BG}.


\paragraph{Proving \cref{claim:3partyConnectionToVectorGame}}
The proof immediately follows from the following claim and \cref{lemma:ValueVectorGame}.

\begin{claim}\label{claim:3partyConnectionToVectorGameGen}
	Assuming $\Aadv$'s first abort is only at \Stepref{step:ThreeParty:round:1} of the loop in $\Pithree_\rnd$,  then
	\begin{align*}
	\size{\eex{\val(V_{I_1}) - \val(V_{I_1})^-}} \leq \bias(\game_{f,\rnd,0}),
	\end{align*}
	 where $\game_{f,\rnd,0}$ is a $9$-vector game according to  \cref{def:ValueVectorGame} and $\bias(\cdot)$ is according to \cref{def:GameBias}.
\end{claim}

\begin{proof}[Proof of \cref{claim:3partyConnectionToVectorGameGen}]
We prove that an attacker for the coin-flipping protocol of the type considered in this claim, \ie one that achieves bias $\alpha$ from his first abort at \Stepref{step:ThreeParty:round:1} of the loop in $\Pithree_\rnd$, yields a player for the vector binomial game (described in \cref{subsection:BG}) that achieves the same bias. Thus, the bound on the former attacker follows from the bound on the latter one.

Recall that $\Aadv$ controls the parties $\set{\Pthree_\z}_{\z\in \cC}$ in $\Pithree_\rnd$ and assume \wlg that $\cC = \zo$. 
Consider the $9$-vector binomial game $\game = \game_{f,\rnd,0} = \gameVars$ as defined in \cref{def:ValueVectorGame} and recall that in this game, $f$ is the randomized function that on input $(i,y)$ outputs a string in $\oo^{9\cdot \ms{1}}$, where each of entries takes the value $1$ with probability $\sBias{\ms{1}}{\delta}$ for $\delta = \vBeroo{\ms{i+1},0}(-y)$. Now, consider the player $\Bc$ for $\game$ that emulates interaction with $\Aadv$ in $\Pithree_\rnd = (\Pthree_0,\Pthree_1,\Pthree_2)$, where $\Aadv$ controls $\Pthree_0$ and $\Pthree_1$. The emulation goes as follows:

\begin{algorithm}[Player $\Bc$ for the vector game]
	\item[Operation:]~
	\begin{enumerate}
		
		\item For $\z \in \zo$, choose $\SSS{\z} = (\cVS{\z},\Mat{D^{(0,1),\#\z}},\Mat{D^{(0,2),\#\z}},\Mat{D^{(1,0),\#\z}},\Mat{D^{(1,2),\#\z}}, \Mat{D^{(2,0),\#\z}},\Mat{D^{(2,1),\#\z}})$ $ \la \zo^{\rnd \times \tp{\rnd} + 6 \times \rnd \times (\rnd \times \tp{\rnd} + 2(\rnd+1))}$ as inputs for $\Pthree_0$ and $\Pthree_1$ in $\Pithree_\rnd = (\Pthree_0,\Pthree_1,\Pthree_2)$.
		
		
		\item For $i=1$ to $\rnd$:
		\begin{enumerate}
			
			\item Receive input $(i,Y_{i-1},\aux_i)$ from the game $\game$ (recall that  $\aux_i \in \oo^{9\cdot \ms{1}}$).
			
			\item If $i>1$, emulate a sending of $\cVS{2}[i-1]$ from $\Ptwo_2$ to the other parties at step $(i-1,b)$ of $\Pithree_\rnd$,  where $\cVS{2}[i-1] \eqdef X_{i-1} \xor \cVS{0}[i-1] \xor \cVS{1}[i-1]$ (recall that $X_{i-1} \eqdef Y_{i-1} - Y_{i-2}$ according to \cref{def:game}).
			
			\item  Emulate a sending of $\Mat{D^{(0,1),\#2}}[i]$ and $\Mat{D^{(0,2),\#2}}[i]$ from $\Pthree_2$ to $\Pthree_0$, and emulate a sending of $\Mat{D^{(1,0),\#2}}[i]$ and $\Mat{D^{(1,2),\#2}}[i]$ from $\Pthree_2$ to $\Pthree_1$ at step $(i,a)$ of $\Pithree_\rnd$, where 
			
			\begin{enumerate}
				\item $(\Mat{D^{(0,1),\#2}}[i], \Mat{D^{(1,0),\#2}}[i])$ is set to the output of $\HidTwoShareGen$ where $\vct^\z$ at \Stepref{alg:HidingTwoShareGenerator:sampleSets} of $\HidTwoShareGen$ is set to $\{\aux_i[2\z\cdot \ms{1}],\ldots,\aux_i[4\z\cdot \ms{1}-1]\}$ for $\z \in \zo$, and $c_i$ at \Stepref{alg:HidingTwoShareGenerator:sampleCi} is set to $\aux_i[4\cdot \ms{1} - 1 + i]$ for $i \in [\rnd]$.
				
				\item For $\z \in \zo$, choose a random $(\ms{1})$-size subset $\cW^\z \subset \{\aux_i[(5+2\z)\cdot \ms{1}], \ldots, \aux_i[(7+2\z)\cdot \ms{1}-1]\}$, choose uniformly $\Mat{D^{(\z,2),\#2}}[i] \la \zo^{\rnd\cdot\tp{\rnd} + (\rnd+1)}$, and change the last bit of $\Mat{D^{(\z,2),\#2}}[i]$ to $1$ if $\sum_{\w \in \cW^\z}w \geq 0$ and to $0$ otherwise.
			\end{enumerate}\label{lemma:3partyConnectionToVectorGame:defenseStep}
			
		If $\Pthree_0$ or $\Pthree_1$ aborts at this step, output $1$ (abort at round $i$). Otherwise, output $0$ (continue to next round).
			
		\end{enumerate}
		
	\end{enumerate}
\end{algorithm}

Note that by the definition of $f$, the set of $\Bc$'s emulation variables $(X_1, \ldots, X_\rnd, \{\Mat{D^{(\z,\z')}}\}_{\z \in \zo, \z \ne \z'})$ has the same joint distribution as the set of $\Pithree_\rnd$'s variables $(c_1, \ldots,c_\rnd, \{\Mat{D^{(\z,\z')}}\}_{\z \in \zo, \z \ne \z'})$, where the parties' inputs are sampled according to $\ThreeShareGen$. Therefore, $\Aadv$'s view in the  emulation done by $\Bc$, is distributed exactly the same as its view when interacting with honest $\Pthree_2$ in $\Pithree_\rnd$; in both cases, the only meaningful information it gets is the revelation of these values. In addition, note that $\Bc$ aborts at round $i$ of $\game$ (i.e outputs $1$) iff $\Aadv$'s first aborts is at round $(i,a)$ of $\Pithree_\rnd$.

For concluding the analysis, let $g$ be the function that on input $\aux_i \in \oo^{9\cdot \ms{1}}$, outputs $\{\Mat{D^{(\z,\z')}}\}_{\z \in \zo, \z \ne \z'}$ as described in \Stepref{lemma:3partyConnectionToVectorGame:defenseStep}, and let $\Bc'$ be an attacker for $\game_{g \circ f,\rnd,0}$ that operates just like $\Bc$, only that it gets the output of $g \circ f$ directly instead of constructing the output as $\Bc$ does at \Stepref{lemma:3partyConnectionToVectorGame:defenseStep}. It follows that

\begin{align}
\bias_{\Bc'}(\game_{g \circ f,\rnd,0}) = \size{\eex{\val(V_{I_1}) - \val(V_{I_1})^-}},
\end{align}

We conclude that
\begin{align}
\size{\eex{\val(V_{I_1}) - \val(V_{I_1})^-}}
&= \bias_{\Bc'}(\game_{g \circ f,\rnd,0})\\
&\leq \bias(\game_{g \circ f,\rnd,0})\nonumber\\
&\leq \bias(\game_{f,\rnd,0}).\nonumber
\end{align}
\end{proof}


\paragraph{Proving \cref{claim:3partyConnectionToHyperGame}}
The proof immediately follows from the following claim and \cref{lemma:ValueHgGame}.
\begin{claim}\label{claim:3partyConnectionToHyperGameGen}
	Assuming $\Aadv$'s second abort is only at \Stepref{step:TwoParty:round:1} of the loop in $\Pitwo_\rnd$ and that $\bias(\game_{f,\rnd,\eps}) \leq \alpha(\rnd)$ for any $\eps \in [-1,1]$ and for any randomized function $f$ of a $12$-hypergeometric game, as defined in \cref{def:ValueHgGame}. Then 
	\begin{align*}
	\size{\eex{\val(V_{I_2}) - \val(V_{I_2})^-}} \leq \alpha(\rnd) + \frac2{\rnd}.
	\end{align*}
\end{claim}

\begin{proof}[Proof of \cref{claim:3partyConnectionToHyperGame}]
	We prove that an attacker for the coin-flipping protocol of the type considered in this claim, \ie one that achieves bias $\alpha$ from his second abort at \Stepref{step:TwoParty:round:1} of the loop in $\Pitwo_\rnd$, yields a player for the hypergeometric binomial game (described in \cref{def:ValueHgGame}) that achieves the same bias. Thus, the bound on the former attacker follows from the bound on the latter one.
	
	Assume \wlg that $\Aadv$ controls $\Ptwo_0$ in $\Pitwo_\rnd$ after the first abort, and recall that the parties' inputs are sampled according to $\HidTwoShareGen(\delta)$ for some $\delta \in [0,1]$. Let $\eps = \sBias{\ms{1}}{\delta}$ and let $p = \w(\vct^0)$ where $\vct^0$ is the vector that has been sampled in \Stepref{alg:HidingTwoShareGenerator:sampleSets} of $\HidTwoShareGen$.
	Consider the binomial game $\game = \game_{f_p,\rnd,\eps} = \gameVars$ as defined in \cref{def:game} where $f_p$ is the randomized function that on input $(i,y)$ outputs $1$ with probability $\delta_p = \vHyp{2\cdot \ms{1},p,\ms{i+1}}(-y)$ and zero otherwise (i.e hypergeometric game). Let $\Bc$ be the player described in the proof of \cref{claim:2partyConnectionToSimpleGame} at \cref{claim:2partyConnectionToSimpleGame:attacker}, and let $\Bc'$ be the player for $\game$ which operates as $\Bc$, where the only difference is that at step \ref{claim:2partyConnectionToSimpleGame:attacker:firstStep} it chooses $\dVS{0}{0}[\rnd+1]$ according to $\Berzo{\delta_p}$ instead of uniformly over $\zo$ as $\Bc$ does.
	By the definition of $f_p$, the set of $\game$'s variables $(X_1, \ldots, X_\rnd, \aux_1, \ldots, \aux_\rnd)$ has the same distribution as the set of $\Pitwo_\rnd$'s variables $(c_1, \ldots,c_\rnd, d^0_1, \ldots, d^0_\rnd)$ after $\Aadv$'s first abort. 
	
	Therefore,  after the first abort, $\Aadv$'s view in the  emulation done by $\Bc$, is distributed exactly the same as its view when interacting with honest with honest $\Ptwo_1$ in $\Pitwo_\rnd$; in both cases, the only meaningful information $\Aadv$' gets is the revelation of these values together with the value of $\dVS{0}{0}[\rnd+1]$, which is distributed the same in both cases. Since $\Bc'$ aborts at round $i$ of $\game$ (i.e outputs $1$) iff $\Aadv$ aborts at round $(i,a)$ of $\Pitwo_\rnd$, it follows that

	\begin{align}
	\ex{\delta, p}{\bias_{\Bc'}(\game)} = \size{\eex{\val(V_{I_2}) - \val(V_{I_2})^-}},
	\end{align}
	
	where $\bias_{\Bc'}(\game)$ is according to \cref{def:GameBias} and the expectation is on the values of $\delta$ and $p$ which are being set after $\Aadv$'s first abort.
    Note that by the definition of $\eps$ it holds that $\veo_1^- = \delta$, where $\veo_1^-$ is according to \cref{def:game}. Therefore, in case $\delta \notin [\frac1{\rnd^2}, 1-\frac1{\rnd^2}]$, \cref{lemma:DeterminedGames} tells us that $\bias_{\Bc'}(\game) \leq \frac1{\rnd}$. Assuming that $\delta \in [\frac1{\rnd^2}, 1-\frac1{\rnd^2}]$, Hoeffding inequality (\cref{claim:Hoeffding}) yields that $\size{\eps} < 4\cdot \sqrt{\frac{\log \rnd}{\ms{1}}}$. Therefore, since $p$ is distributed according to $\Beroo{2\ms{1},\eps}$, it follows that 
    \begin{align}
    \lefteqn{\pr{\size{p} > 12 \sqrt{ \log \rnd \cdot \ms{1}} \mid \delta \in [\frac1{\rnd^2}, 1-\frac1{\rnd^2}]}}\\
    &\le \pr{\size{p - 2\eps\cdot \ms{1}} > 4 \sqrt{ \log \rnd \cdot \ms{1}} \mid \delta \in [\frac1{\rnd^2}, 1-\frac1{\rnd^2}]}\nonumber\\
    &\le \frac{1}{\rnd},\nonumber
    \end{align}
    where the second inequality holds again by Hoeffding inequality. We conclude that
    \begin{align}
    \size{\eex{\val(V_{I_2}) - \val(V_{I_2})^-}}
    &= \ex{\delta, p}{\bias_{\Bc'}(\game)}\\
    &\leq \ex{\delta, p}{\bias_{\Bc'}(\game) \mid \delta \in [\frac1{\rnd^2}, 1-\frac1{\rnd^2}]} + \frac1{\rnd}\nonumber\\
    &\leq \ex{\delta, p}{\bias_{\Bc'}(\game) \mid \delta \in [\frac1{\rnd^2}, 1-\frac1{\rnd^2}] \bigwedge \size{p} \leq 12 \sqrt{ \log \rnd \cdot \ms{1}}} + \frac2{\rnd}\nonumber\\
    &\leq \alpha(\rnd)+\frac2{\rnd},\nonumber
    \end{align}
	where the last inequality holds by the assumption.
\end{proof}

\section{Bounds for Online-Binomial Games}\label{section:BoundsForOnlineBinomialGames}
In the following section we focus on online-binomial games, as defined in \cref{subsection:BG}.
In \cref{section:GenericApproach} we state two basic tools for bounding the bias of such games, in \cref{subsection:Final} we develop our main tool for bounding the bias of such games and in Sections \ref{sec:Bernoulli}-\ref{sec:Vector} we prove \cref{lemma:ValueBiomialGame,lemma:ValueHgGame,lemma:ValueVectorGame}.

Recall that $\Beroo{\eps}$ is the Bernoulli probability distribution over $\oo$, taking the value $1$ with probability $\frac{1}{2}(1+\eps)$ and $-1$ otherwise, that $\Beroo{n,\eps}$ is the binomial distribution induced by the sum of $n$ independent random variables, each distributed according to  $\Beroo{\eps}$, and that $\ml{i} = \rnd+1-i$ and $\ms{i} = \sum_{j=i}^\rnd \ml{j}$.

We recall the following definitions from \cref{subsection:BG}.

\begin{definition}[online  binomial games -- Restatement of \cref{def:game}]\label{bounds:def:game}
	\onlineBinomialGame
\end{definition}

\begin{definition}[game bias -- Restatement of \cref{def:GameBias}]\label{bounds:def:GameBias}
	\GameBias<\cref{bounds:def:game}>
\end{definition}

\subsection{Basic Tools}\label{section:GenericApproach}
We present two basic tools for bounding a game bias. The first tool asserts that the game bias can only decrease when applying a random function to the hint.

\begin{lemma}\label{lemma:InformationIncreaseValue}
For randomized functions $f$ and $g$, $\rnd\in \N$, $\eps\in [-1,1]$ and $\ofs \in \Z$, let $\game_{f,\rnd,\eps,\ofs}$ and $\game_{g \circ f,\rnd,\eps,\ofs}$  be according to \cref{bounds:def:game}. It holds that  $\bias(\game_{g \circ f,\rnd,\eps,\ofs}) \leq \bias(\game_{f,\rnd,\eps,\ofs})$.
\end{lemma}
\begin{proof}
	Let $\set{Y_i}_{i = 0}^{\rnd}$ be distributed as in  $\game_{f,\rnd,\eps,\ofs}$ (according to \cref{bounds:def:game}),  let $\tau = g \circ f$
	and let $\Bc^\tau$ be the algorithm with $\bias_{\Bc^\tau}(\game_{\tau,\rnd,\eps, \ofs}) = \bias(\game_{\tau,\rnd,\eps, \ofs})$. Let $\Bc^{f}$ be the algorithm for $\game_{f,\rnd,\eps,\ofs}$ that emulates an execution of $\Bc^{\tau}$ in $\game_{\tau,\rnd,\eps, \ofs}$. Namely, $\Bc^{f}$ on input $(i,y,\paux)$ emulates a sending of $(i,y,g(\paux))$ from the game $\game_{\tau,\rnd,\eps, \ofs}$ to $\Bc^{\tau}$ and outputs $\Bc^{\tau}$'s output. Let $I$ be the first round on which $\Bc^\tau$ outputs one in $\game_{\tau,\rnd,\eps,\ofs}$, let $\aux_i^f = f(i,Y_i)$ and let $\aux_i^\tau = \tau(i,Y_i)$. It follows that
	\begin{align*}
	\lefteqn{\bias(\game_{\tau,\rnd,\eps,\ofs})}\\
	& = \bias_{\Bc^\tau}(\game_{\tau,\rnd,\eps,\ofs})\\
	&= \size{\ex{i \la I}{\ex{y\la Y_{i-1},\paux \la \aux^\tau_i \mid I=i}{\pr{Y_{\rnd} \geq0 \mid Y_{i-1} = y,\aux^\tau_i = \paux} -\pr{Y_{\rnd} \geq0 \mid Y_{i-1}=y}}}}\\
	&= \size{\ex{i \la I}{\ex{y\la Y_{i-1},\paux \la \aux^\tau_i \mid I=i}{ \ex{\paux' \la  \aux^f_i\mid Y_{i-1} = y,\aux^\tau_i=\paux }{\pr{Y_{\rnd} \geq0 \mid Y_{i-1} = y,\aux^f_i = \paux'} -\pr{Y_{\rnd} \geq0 \mid Y_{i-1}=y}}}}}\\
	&= \size{\ex{i \la I}{\ex{y\la Y_{i-1},\paux' \la \aux^f_i \mid I=i}{\pr{Y_{\rnd} \geq0 \mid Y_{i-1} = y,\aux^f_i = \paux'} - \pr{Y_{\rnd} \geq0 \mid Y_{i-1}=y}}}}\\
	&= \bias_{\Bc^f}(\game_{f,\rnd,\eps,\ofs})\\
	&\leq \bias(\game_{f,\rnd,\eps,\ofs}),
	\end{align*}
	where the last equality holds since $I$ also describes the first output on which $\Bc^f$ outputs one in $\game_{f,\rnd,\eps,\ofs}$.
\end{proof}

The second tool is a restatement of \cref{lemma:DeterminedGames}. It asserts that if the expected value of a game is almost determined, then there is mo much room for an attacker to gain much bias.

\begin{lemma}\label{bounds:lemma:DeterminedGames}[Restatement of \cref{lemma:DeterminedGames}]
	\DeterminedGameLemma<\cref{bounds:def:game}>
\end{lemma}
\begin{proof}
We prove the case $\veo_1^- \leq \frac1{\rnd^2}$, where the other case is analogues.
By a simple averaging argument, it holds that
\begin{align}
 \pr{\exists i\in [\rnd] \colon \veo_i^- > \frac1{\rnd}} \leq \frac1{\rnd}
\end{align}
Consider the game $\game_{g,\rnd,\eps,\ofs}$ for $g(i,y) = y$. By the above, $\bias(\game_{g,\rnd,\eps,\ofs}) \leq \frac2{\rnd}$. Hence, \cref{lemma:InformationIncreaseValue} yields that the same also holds for $\game_{f,\rnd,\eps,\ofs}$.
\end{proof}

\subsection{Main Tool --- Expressing Game Bias using Ratio}\label{subsection:Final}
In this section, we develop our main tool for bounding the bias of online-binomial games.
Informally, we reduce the task of bounding the bias of a game into evaluating the ``ratio" of the game, where $\ratioo$ (defined below) is a useful game-depend measurement on how much the distribution of $X_i$ is far from the distribution of $X_i \mid \aux_i$.

\remove{
Throughout, we will restrict our attention to setting in which the ``coins'' flipped during the game execution take typical values.
\begin{definition}[typical sets]\label{def:SetsXY}
For $\rnd\in\N$,   $\eps\in [-1,1]$ and 	$i \in [\rnd]$, let $\xset^\rnd_i = \set{x \in \Z \colon \abs{x} \leq 4 \cdot \sqrt{\log{\rnd}\cdot \ml{i}}}$ and $\yset^{\rnd,\eps}_i = \set{y \in \Z\colon \abs{y + \eps \cdot \ms{i}} \leq 4\sqrt{\log{\rnd} \cdot \ms{i}}}$.
\end{definition}
We will  use the above sets in context of a binomial game $\game_{f,\rnd,\eps}$, and will sometimes omit $\rnd$ from the notation of this sets.}

\begin{definition}\label{def:ratio}
	Let $\game_{f,\rnd,\eps} = \gameVars$ be according to  \cref{bounds:def:game} and let $\cX_i  \eqdef \set{x \in \Supp(X_i) \colon \abs{x} \leq 4 \cdot \sqrt{\log{\rnd}\cdot \ml{i}}}$. For $i \in [\rnd]$, 
	$y \in \Supp(Y_{i-1})$, $x \in \Supp(X_i)$ and  $\paux \in \Supp(\aux_i)$, define
	\begin{align*}
	\ratioo_{i,y,\paux}(x)= \frac{\pr{X_i=x \mid Y_{i-1} = y, \aux_i = \paux, X_i \in \cX_i}}{\pr{X_i= x\mid Y_{i-1}=y, X_i \in \cX_i}}.
	\end{align*}
\end{definition}

Namely, $\ratioo_{i,y,\paux}(x)$ measures the change (in multiplicative term) of the probability   $X_i =x$, due to the knowledge of  $\paux$, assuming that  $X_i$ is typical (\ie  $X_i\in \cX_i$).

The following lemma states that an appropriate upper-bound on $\size{1-\ratioo_{i,Y_{i-1},\aux_i}}$ for any ``interesting'' round $i$, yields an upper-bound on the game bias.

\begin{lemma}\label{lemma:Final}[main tool: expressing game bias using ratio]
	Let $\game= \game_{f,\rnd,\eps} = \set{X_i,Y_i,H_i,O_i,O_i^-}$ be according to  \cref{bounds:def:game}, and let $\ratioo$ be according to  \cref{def:ratio}. 
	Assume that for every $i\in [\rnd-\floor{\log^{2.5}\rnd}]$ and $y\in  \cY_{i-1} \eqdef \set{y \in \Supp(Y_{i-1}) \colon \abs{y + \eps \cdot \ms{i}} \leq 4\sqrt{\log{\rnd} \cdot \ms{i}}}$, there exist $\const >0$ and  a set $\aset_{i,y}$  such that:
	\begin{enumerate}
		
		\item $\pr{\aux_i \notin \aset_{i,y} \mid Y_{i-1} = y} \leq \frac1{\rnd^2}$, and\label{lemma:Final:1}
		
		\item $\abs{1-\ratioo_{i,y,\paux}(x)} \leq \const \cdot \sqrt{\frac{ \log{\rnd}}{\ml{i+1}}}\cdot (\frac {\abs{x}}{\sqrt{\ml{i}}} + 1)$ for every $(x,\paux) \in \cX_i\times \aset'_{i,y}$, \label{lemma:Final:2}
	\end{enumerate}
	for $\aset'_{i,y} = \aset_{i,y} \bigcap \Supp(\aux_i \mid Y_{i-1}=y, X_i \in \cX_i)$. Then 
	$$\bias(\game) \leq \varphi(\const) \cdot \frac{\log^{3}\rnd}{\rnd}$$
	 for a universal function $\varphi$.
\end{lemma}

In the following we fix $f,\rnd,\eps$, we let $\game=\game_{f,\rnd,\eps} = \set{X_i,Y_i,H_i,O_i,O_i^-}$ be according to  \cref{bounds:def:game}, and for $i\in [\rnd]$, we let $\cX_i$ be according to \cref{def:ratio} and $\cY_{i-1}$ be according to \cref{lemma:Final}. We assume \wlg that $\rnd$ is larger than some universal constant and that $\size{\eps} \leq  4\sqrt{\frac{\log{\rnd}}{\ms{1}}}$ (Otherwise, Hoeffding's inequality (\cref{claim:Hoeffding}) yields that $\veo_1^- \notin  [\frac1{\rnd^2} ,1 - \frac1{\rnd^2}]$ and the proof follows by \cref{bounds:lemma:DeterminedGames}).

The following sub-lemmas are the main building blocks for proving \cref{lemma:Final}. The first one (proved in \cref{subsubsection:OnlyWithOi}) states that an appropriate bound of each round bias, yields a bound on the game bias.
\begin{lemma}\label{proposition:OnlyWithOi}
	Assume that for every $i\in [\rnd-\floor{\log^{2.5}\rnd}]$ and $y\in\cY_{i-1}$, there exists $\const > 0$ and a set $\widehat{\aset}_{i,y} \subseteq \Supp(\aux_i \mid Y_{i-1} =y)$ such that
	\begin{enumerate}
		
		\item $\pr{\aux_i \notin \widehat{\aset}_{i,y} \mid Y_{i-1} = y} \leq \frac3{\rnd^2}$, and\label{proposition:OnlyWithOi:1}
		
		\item $\size{\eo_i(y)-\eo_i(y,\paux)} \leq \const \cdot \frac{\sqrt{\log{\rnd}}}{\ml{i+1}}$ for every $\paux \in \widehat{\aset}_{i,y}$. \label{proposition:OnlyWithOi:2}
		
	\end{enumerate}
	
	Then $$\bias(\game) \leq \varphi(\const) \cdot \frac{\log^{3}\rnd}{\rnd}$$ for a universal function $\varphi$.
\end{lemma}

The following lemma (proved  in  \cref{subsubsection:biasBoundOneRound}) relates the bias that can be obtained in a given round, to the $\ratioo$ function.
\begin{lemma}\label{proposition:biasBoundOneRound}
	For  $i \in [\rnd]$, $y\in \Supp(Y_{i-1})$ and $\paux \in \Supp(\aux_i \mid Y_{i-1}=y, X_i \in \cX_i)$, it holds that
	$$\size{\eo_{i}(y) - \eo_{i}(y,\paux)} \leq \ex{x\la X_i\mid x\in \cX_i}{\size{\eo_{i+1}(y+x) - \eo_{i+1}(y)}\cdot \size{1-\ratioo_{i,y,\paux}(x)}} + 2\cdot(q + q_\paux),$$
	for $q = \Pr[X_i \notin \cX_i]$ and $q_\paux = \Pr[X_i \notin \cX_i \mid Y_{i-1} = y, \aux_i= \paux]$.\footnote{It can be easily shown that $\eo_{i}(y) - \eo_{i}(y,\paux) = \Ex_{x\la X_i}[(\eo_{i+1}(y+x) - \eo_{i+1}(y))\cdot (1-\frac{\pr{X_i=x \mid Y_{i-1} = y, \aux_i = \paux}}{\pr{X_i= x\mid Y_{i-1}=y}})]$. However, \cref{proposition:biasBoundOneRound} allows us to ignore ``non-typical'' $x$'s.}
\end{lemma}
Intuitively, the above tells that if $\aux_i$ is unlikely to tell  much information about  $X_i$, reflected by $\ratioo_{i,y,\aux_i}(X_i)$ being close to one,  then the bias of round $i$ is small.

We will also use  the following two simple facts. The first one states some useful properties of the sets $\cX_{i}$'s.
\begin{claim}\label{claim:XiiPurpose}
	The following holds for every $i \in [\rnd]$.
	\begin{enumerate}
		
		\item $\pr{X_i \notin \cX_i} < \frac1{\rnd^3}$,\label{claim:XiiPurpose:1}
		
		\item $\ex{x\la X_i\mid x\in \cX_i}{\size{x}} < \ex{x\la X_i}{\size{x}}$,\label{claim:XiiPurpose:2}\nonumber
		
		\item $\ex{x\la X_i\mid x\in \cX_i}{x^2} < \ex{x\la X_i}{x^2}$.\label{claim:XiiPurpose:3}
	\end{enumerate}
\end{claim}
\begin{proof}
	For \cref{claim:XiiPurpose:1}, compute
	\begin{align*}
	\pr{X_i \notin \cX_i}\\
	&= \pr{\abs{X_i} > 4\sqrt{\log{\rnd}\cdot \ml{i}}}\nonumber\\
	&< \pr{\abs{X_i - \eps\cdot \ml{i}} > 3\cdot \sqrt{\log{\rnd}\cdot \ml{i}}}\nonumber\\
	&\leq 2\cdot \exp\left(-\frac{3^2 \cdot \log{\rnd} \cdot \ml{i}}{2\cdot \ml{i}}\right)\nonumber\\
	&< \frac1{\rnd^3},\nonumber
	\end{align*}
	where the first inequality holds since $\abs{\eps} \leq 4\sqrt{\frac{\log{\rnd}}{\ms{1}}}$ yields that $\abs{\eps}\cdot \ml{i} < \sqrt{\log{\rnd}\cdot \ml{i}}$ and the second inequality holds by Hoeffding's inequality (\cref{claim:Hoeffding}).
	
	For \cref{claim:XiiPurpose:2}, compute
	\begin{align*}
	\ex{x\la X_i}{\abs{x}}& =  \ppr{x\la X_i}{X_i \in \cX_i}  \cdot \ex{x\la X_i\mid x\in \cX_i}{\abs{x}} +  \ppr{x\la X_i}{X_i \notin \cX_i}  \cdot \ex{x\la X_i\mid x\notin \cX_i}{\abs{x}}\\
	&> \ppr{x\la X_i}{X_i \in \cX_i}  \cdot \ex{x\la X_i\mid x\in \cX_i}{\abs{x}} +  \ppr{x\la X_i}{X_i \notin \cX_i}  \cdot \ex{x\la X_i\mid x\in \cX_i}{\abs{x}}\nonumber\\
	&= \ex{x\la X_i\mid x\in \cX_i}{\abs{x}},\nonumber
	\end{align*}
	where the inequality holds since $\ex{x\la X_i\mid x\notin \cX_i}{\abs{x}} > 4\sqrt{\log{\rnd} \cdot \ml{i}} \geq \ex{x\la X_i\mid x\in \cX_i}{\abs{x}}$. The proof of \cref{claim:XiiPurpose:3} is analogous to the above.
\end{proof}

The second claim bounds the change of the expected game value in a single round.
\begin{claim}\label{claim:GameValueDiffOneRound}
	For $i \in [\rnd-\floor{\log^{2.5}\rnd}]$, $x \in \cX_i$ and $y \in \cY_{i-1}$, it holds that $$\abs{\eo_{i+1}(y+x) - \eo_{i+1}(y)}\leq \frac{\abs{x}}{\sqrt{\ms{i+1}}}.$$
\end{claim}
\begin{proof}
	Note that $\size{x}, \size{y+x} \leq 5\cdot \sqrt{\log \rnd \cdot \ms{i}} < \ms{i}^{\frac35}$, that $\size{\eps} \leq 4\sqrt{\frac{\log \rnd}{\ms{1}}} < \ms{i}^{-\frac25}$ and that $\abs{\eo_{i+1}(y+x) - \eo_{i+1}(y)} = \size{\vBeroo{\ms{i+1},\eps}(-y-x) - \vBeroo{\ms{i+1},\eps}(-y)}$. Therefore, the proof immediately follows by \cref{prop:gameValuesDifferenceBound}.
\end{proof}

\mparagraph{Putting it together}
\begin{proof}[Proof of \cref{lemma:Final}]
	Let $i \in [\rnd-\floor{\log^{2.5}\rnd}]$, $y\in\cY_{i-1}$ and $\aset_{i,y}$ be the set that satisfies assumptions \ref{lemma:Final:1} and \ref{lemma:Final:2} of \cref{lemma:Final}, let $\aset'_{i,y} = \aset_{i,y} \bigcap \Supp(\aux_i \mid Y_{i-1}=y, X_i \in \cX_i)$ and let $\widehat{\aset}_{i,y} = \set{\paux \in \aset'_{i,y} \mid \pr{X_i\notin \cX_i\mid Y_{i-1}=y,\aux_i=\paux} \leq \frac1{\rnd}}$. We first show that $\pr{\aux_i \notin \widehat{\aset}_{i,y} \mid Y_{i-1} = y} \leq \frac3{\rnd^2}$. Next, we use \cref{proposition:biasBoundOneRound} for bounding $\size{\eo_{i}(y,\paux) - \eo_{i}(y)}$ for every $\paux \in \widehat{\aset}_{i,y}$ and the proof will follow by \cref{proposition:OnlyWithOi}.
	
	For the first part, let $\cs_{i,y} = \set{\paux \in \Supp(\aux_i) \mid \pr{X_i\notin \cX_i\mid Y_{i-1}=y,\aux_i=\paux} \leq \frac1{\rnd}}$ and
	assume by contradiction that $\pr{\aux_i \notin \cs_{i,y} \mid Y_{i-1}=y} > \frac1{\rnd^2}$. It follows that
	\begin{align*}
	\pr{X_i\notin \cX_i}
	&= 	\pr{X_i\notin \cX_i \mid Y_{i-1} = y}\\
	&\geq \pr{X_i\notin \cX_i \mid Y_{i-1}=y,\aux_i \notin \cs_{i,y}}\cdot \pr{\aux_i \notin \cs_{i,y} \mid Y_{i-1}=y}\\
	&> \frac1{\rnd}\cdot \frac1{\rnd^2}\\
	&= \frac1{\rnd^3},
	\end{align*}
	In contradiction to \cref{claim:XiiPurpose:1} of \cref{claim:XiiPurpose}.
	Therefore,
	\begin{align}\label{claim:GoodCommonValuesOfA:11}
	\pr{\aux_i \notin \cs_{i,y} \mid Y_{i-1}=y} \leq \frac1{\rnd^2}.
	\end{align}
	In addition, note that
	\begin{align}\label{claim:GoodCommonValuesOfA:12}
	\pr{\aux_i \notin \Supp(\aux_i \mid Y_{i-1}=y, X_i \in \cX_i) \mid Y_{i-1}=y}
	\leq \pr{X_i \notin \cX_i}
	\leq \frac1{\rnd^2}.
	\end{align}
	Using simple union bound, we conclude from \cref{claim:GoodCommonValuesOfA:11,claim:GoodCommonValuesOfA:12} that
	\begin{align}
	\pr{\aux_i \notin \widehat{\aset}_{i,y} \mid Y_{i-1} = y} 
	\leq \pr{\aux_i \notin \aset_{i,y} \mid Y_{i-1} = y} + \frac2{\rnd^2}
	\leq \frac3{\rnd^2},
	\end{align}
	where the second inequality holds by assumption \ref{lemma:Final:1} of \cref{lemma:Final}.

	For the second part, note that for every $\paux \in \widehat{\aset}_{i,y}$ it holds that
	\begin{align}
	\size{\eo_{i}(y,\paux) - \eo_{i}(y)}
	&\leq \ex{x\la X_i\mid x\in \cX_i}{\abs{\eo_{i+1}(y+x) - \eo_{i+1}(y)}\cdot \abs{ 1-\ratioo_{i,y,\paux}(x)}} + \frac{4}{\rnd}\\
	&\leq \ex{x\la X_i\mid x\in \cX_i}{\frac{\abs{x}}{\sqrt{\ms{i+1}}} \cdot \left(\const\cdot\sqrt{\frac{ \log{\rnd}}{\ml{i+1}}}\bigl(\frac {\abs{x}}{\sqrt{\ml{i}}} + 1\bigr)\right)} + \frac{4}{\rnd}\nonumber\\
	&= \ex{x\la X_i\mid x\in \cX_i}{\frac{\abs{x}}{\sqrt{\frac12\ml{i}\ml{i+1}}} \cdot \left(\const\cdot\sqrt{\frac{ \log{\rnd}}{\ml{i+1}}}\bigl(\frac {\abs{x}}{\sqrt{\ml{i}}} + 1\bigr)\right)} + \frac{4}{\rnd}\nonumber\\\nonumber\\
	&= \frac{\sqrt{2}\const\cdot \sqrt{\log{\rnd}}}{\ml{i+1}}
	\cdot \ex{x\la X_i\mid x\in \cX_i}{\frac{x^2}{\ml{i}} + \frac{\abs{x}}{\sqrt{\ml{i}}}} + \frac{4}{\rnd}\nonumber\\
	&\leq \frac{\sqrt{2}\const\cdot \sqrt{\log{\rnd}}}{\ml{i+1}}
	\cdot \left(\frac{2\cdot \ml{i}}{\ml{i}} + \frac{\sqrt{2\cdot\ml{i}}}{\sqrt{\ml{i}}}\right) + \frac{4}{\rnd}\nonumber\\
	&\leq \frac{(5\const + 4)\cdot \sqrt{\log{\rnd}}}{\ml{i+1}}.\nonumber
	\end{align}
	The first inequality holds by \cref{proposition:biasBoundOneRound} (recalling \cref{claim:XiiPurpose:1} of \cref{claim:XiiPurpose} and that $\pr{X_i\notin \cX_i\mid Y_{i-1}=y,\aux_i=\paux} \leq \frac1{\rnd}$ by the definition of $\widehat{\aset}_{i,y}$). The second inequality holds by \cref{claim:GameValueDiffOneRound} and by assumption \ref{lemma:Final:2} of \cref{lemma:Final}. The third inequality holds by \cref{fact:E_x_m} (recalling \cref{claim:XiiPurpose:2,claim:XiiPurpose:3} of \cref{claim:XiiPurpose}).
	
	In conclusion, we proved that for every $i \in [\rnd-\floor{\log^{2.5}\rnd}]$ and $y\in\cY_{i-1}$, the set $\widehat{\aset}_{i,y}$ satisfies the constrains  of \cref{proposition:OnlyWithOi}. Hence, \cref{proposition:OnlyWithOi} yields that $\bias(\game) \leq \varphi(5\const + 4) \cdot \frac{\log^{3}\rnd}{\rnd}$, for some universal function $\varphi$, as required.
\end{proof}

\subsubsection{Proving \texorpdfstring{\cref{proposition:OnlyWithOi}}{First Lemma}}\label{subsubsection:OnlyWithOi}
We will use the following facts.  The first claim yields that if $Y_{i-1} \notin \cY_{i-1}$ (\ie $\abs{Y_{i-1}}$ is untypically large), then the expected value of the game at round $i$ is almost determined.

\begin{claim}\label{claim:YiPurpose}
For every $i \in [\rnd]$ and $y \in \Supp(Y_{i-1}) \setminus \cY_{i-1}$, it holds that $$\pr{\sum_{j=i}^{\rnd}X_{j} \geq -y} \notin \left[\frac1{\rnd^2},1-\frac1{\rnd^2}\right].$$
\end{claim}
\begin{proof}
Let $Z_i \eqdef \sum_{j=i}^{\rnd}X_{j}$. We assume that $y + \eps \cdot \ms{i} \leq 0$, where the proof of the case $y + \eps \cdot \ms{i} > 0$ is analogous. Note that since $y \notin \cY_{i-1}$, then $-(y + \eps \cdot \ms{i}) > 4\sqrt{\log{\rnd} \cdot \ms{i}}$, and since $Z_i$ is distributed according to $\Beroo{\ms{i},\eps}$, then $\eex{Z_i} = \eps \cdot \ms{i}$. Therefore, Hoeffding's inequality (\cref{claim:Hoeffding}) yields that
\begin{align*}
\pr{Z_i\geq -y}&= \pr{Z_i-\eps\cdot \ms{i} \geq -(y + \eps\cdot \ms{i})}\\
&\leq \pr{Z_i-\eps\cdot \ms{i} \geq 4\sqrt{\log{\rnd} \cdot \ms{i}}}\nonumber\\
&\leq 2\cdot \exp\left(-\frac{16 \cdot \ms{i}\log{\rnd}}{2\cdot \ms{i}}\right)\nonumber\\
&< \frac1{\rnd^2}.\nonumber
\end{align*}
\end{proof}

We associate the following events with $\game$. For $i\in [\rnd]$, let $E_i$ be the event that $Y_{i-1} \in \cY_{i-1}$ and for $i \in (\rnd)$ let $L_i = E_1 \bigcap E_2 \bigcap \ldots \bigcap E_i \bigcap \neg E_{i+1}$, letting $E_{\rnd+1} =\emptyset$. In words, $E_i$ is the event that $\abs{Y_{i-1}}$ is not large, and $L_i$ is the event that $i$ is the minimal index such that $\abs{Y_i}$ is large (where $L_\rnd$ is the event that all the $Y_i$'s are not large). Note that $\set{L_j}_{j\in (\rnd)}$ are disjoint events and that $\pr{\bigcup_{j\in (\rnd)} L_j} = 1$. We use the following fact.

\begin{claim}\label{claim:PrEiNew}
For integer $i \in [\frac{\rnd}{2},\rnd]$, it holds that $\Pr[E_i] \leq \frac{12\cdot\ml{i}\sqrt{\log{\rnd}}}{\rnd}$.
\end{claim}
\begin{proof}
Note that $Y_{i-1}$ is the outcome of $\ms{1} - \ms{i}$ coins. Compute
\begin{align}
\ms{1} - \ms{i} &= \frac12 \left(\ml{1}(\ml{1} + 1) - \ml{i}(\ml{i} + 1)\right)\\
&= \frac12 \left(\rnd(\rnd + 1) - (\rnd-i+1)(\rnd-i+2)\right)\nonumber\\
&\geq  \frac12 \left(\rnd(\rnd + 1) - (\frac{\rnd}2+1)(\frac{\rnd}2 + 2)\right)\nonumber\\
&\geq \frac{\rnd^2}{4}.\nonumber
\end{align}
\cref{prop:binomProbEstimation} yields that the probability $Y_{i-1}$ equals a given value in $\cY_{i-1}$ is at most $\frac{1}{\sqrt{(\ms{1} - \ms{i})}} \leq \frac{2}{\rnd}$ (recall that we only care about large enough $m$). Since $\size{\cY_{i-1}} < 8\sqrt{\ms{i} \log{\rnd}}$, it follows that
\begin{align*}
\Pr[E_i] &\leq 8\sqrt{\ms{i} \log{\rnd}} \cdot \frac{2}{\rnd}\\
&= \frac{16\sqrt{\ms{i} \log{\rnd}}}{\rnd}\\
&= \frac{16\sqrt{\frac12\cdot \ml{i}\bigl(\ml{i} + 1\bigr) \log{\rnd}}}{\rnd}\\
&\leq \frac{12\cdot \ml{i}\sqrt{\log{\rnd}}}{\rnd}.
\end{align*}
\end{proof}

The following claim bounds the sum $\sum_{j=i}^{\rnd}\pr{L_j}$ for every integer $i \in [\frac{\rnd}{2},\rnd]$.

\begin{claim}\label{claim:constraintsNew}
For integer $i \in [\frac{\rnd}{2},\rnd]$, it holds that $\sum_{j=i}^{\rnd}\pr{L_j} \leq \frac{12\cdot \ml{i}\sqrt{\log{\rnd}}}{\rnd}$.
\end{claim}
\begin{proof}
Since $\set{L_j}_{j=0}^{\rnd}$ are disjoint events and since $\bigcup_{j=i}^{\rnd} L_j \subseteq E_i$, it follows that
$$\sum_{j=i}^{\rnd}\pr{L_j} = \Pr[\bigcup_{j=i}^{\rnd} L_i] \leq \Pr[E_i] \leq \frac{12\cdot \ml{i}\sqrt{\log{\rnd}}}{\rnd},$$
where the last inequality holds by \cref{claim:PrEiNew}.
\end{proof}

\mparagraph{Putting it together}
\begin{proof}[Proof of \cref{proposition:OnlyWithOi}]
Let $\Bc$ be an algorithm and let $\Bc'$ be the algorithm that operates like $\Bc$ with the following difference: if $\Bc$ aborts (\ie output $1$) at round $i$, and $i > \rnd-\log^{2.5}\rnd$ or $i\geq i'$, for $i'$ being the minimal index with $\overline{E_{i'}}$, then $\Bc'$ does not abort, and outputs $0$'s till the end of the game. Combining \cref{claim:PrEiNew,claim:YiPurpose,bounds:lemma:DeterminedGames} yields that
\begin{align}
\size{\bias_{\Bc}(\game) - \bias_{\Bc'}(\game)} \leq \frac{1}{\rnd} + \frac{12 \cdot \log^{3}\rnd}{\rnd}
\end{align}
Let $\Bc''$ be the strategy that acts like $\Bc'$, but does not abort (even if $\Bc'$ does) in rounds $\set{i,\ldots,\rnd}$, for $i$ being the minimal index with $\aux_i \notin \widehat{\aset}_{i,Y_{i-1}}$ and let $I'' = I(\game,\Bc'')$ be according to \cref{bounds:def:GameBias}. Since we assume that $\pr{\aux_i \notin \widehat{\aset}_{i,y} \mid Y_{i-1} = y} \leq \frac3{\rnd^2}$ for every $i\in [\rnd-\floor{\log^{2.5}\rnd}]$ and $y\in\cY_{i-1}$,  a simple averaging argument yields that
\begin{align}
\size{\bias_{\Bc'}(\game) - \bias_{\Bc''}(\game)}
\leq \pr{\exists i\in [\rnd-\floor{\log^{2.5}\rnd}]\colon Y_{i-1}\in \cY_{i-1} \land  \aux_i \notin \widehat{\aset}_{i,Y_{i-1}}}
\leq \frac{3}{\rnd}.
\end{align}

Let $J \in (\rnd)$ be the index for which $L_J$ happens (\ie $J$ is the minimal index such that $Y_J \notin \cY_{J}$). The definition of $\Bc''$ yields that $I'' \leq J$, $I'' \leq \rnd-\log^{2.5}\rnd$, $Y_{I''-1} \in \cY_{I''-1}$ and $\aux_{I''} \in \widehat{\aset}_{I'',Y_{I''-1}}$. Since, by assumption, $\size{\eo_i(y)-\eo_i(y,\paux)} \leq \const \cdot \frac{\sqrt{\log{\rnd}}}{\ml{i+1}}$ for every $i\in [\rnd-\floor{\log^{2.5}\rnd}]$, $y\in\cY_{i-1}$ and $\paux \in \widehat{\aset}_{i,y}$, it follows that
\begin{align}
\abs{\veo_{I''}^- - \veo_{I''}} = \abs{\eo_{I''}(Y_{I''-1}) - \eo_{I''}(Y_{I''-1}, \aux_{I''})} \leq
\const \cdot \frac{\sqrt{\log{\rnd}}}{\ml{I''+1}}
\leq \const \cdot \frac{\sqrt{\log{\rnd}}}{\ml{J+1}}.
\end{align}
We conclude that
\begin{align}
\bias_{\Bc''}(\game) 
&= \size{\ex{i \la I''}{\veo_{i}^- - \veo_{i}}}\\
&\leq \ex{i \la I''}{\size{\veo_{i}^- - \veo_{i}}}\nonumber\\
&\leq \sum_{i=0}^{\rnd-1}{\pr{L_i}\cdot \frac{\const \cdot \sqrt{\log{\rnd}}}{\ml{i+1}}}\nonumber\\
&\leq \const \cdot \sqrt{\log{\rnd}}\cdot \bigl(\sum_{i=0}^{\ceil{\frac{\rnd}{2}}-1}{\frac{\pr{L_i}}{\ml{i+1}}} + \sum_{i=\ceil{\frac{\rnd}{2}}}^{\rnd-1}{\frac{\pr{L_i}}{\ml{i+1}}}\bigr)\nonumber\\
&\leq \const \cdot \sqrt{\log{\rnd}}\cdot \bigl(
\frac{1}{\ml{\ceil{\frac{\rnd}{2}}}} + \frac{12\cdot \sqrt{\log{\rnd}}}{\rnd} \cdot \sum_{i=\ceil{\frac{\rnd}{2}}}^{\rnd-1}{\frac{1}{\ml{i+1}}}\bigr)\nonumber\\
&\leq \const \cdot \sqrt{\log{\rnd}}\cdot \bigl(\frac{2}{\rnd} + \frac{12\cdot \sqrt{\log{\rnd}}}{\rnd} \cdot \sum_{i=\ceil{\frac{\rnd}{2}}}^{\rnd-1}{\frac{1}{\rnd-i}}\bigr)\nonumber\\
&\leq \const \cdot \sqrt{\log{\rnd}}\cdot \bigl(\frac{2}{\rnd} + \frac{12\cdot \log^{1.5}\rnd}{\rnd}\bigr)\nonumber\\
&\leq 13 \const \cdot \frac{\log^{2}\rnd}{\rnd}.\nonumber
\end{align}
The fourth inequality holds by \cref{claim:constraintsNew,claim:linearproblem}, and the one before last
inequality holds since $\sum_{i=1}^{\floor{\frac{\rnd}{2}}}\frac{1}{i} \leq \log{\rnd}$. Hence, $\bias(\game) \leq  13 \const \cdot \frac{\log^{2}\rnd}{\rnd} + \frac4{\rnd} + \frac{12\cdot \log^{3}\rnd}{\rnd} \leq (13\const + 13)  \frac{\log^{3}\rnd}{\rnd}$.
\end{proof}

\subsubsection{Proving \texorpdfstring{\cref{proposition:biasBoundOneRound}}{Second Lemma}}\label{subsubsection:biasBoundOneRound}
The following claim states a more convenient, yet equivalent, expression for the ratio function.
\begin{claim}\label{claim:ratioDifferentVal}
For $x \in \cX_i$, $y \in \Supp(Y_{i-1})$ and $\paux \in \Supp(\aux_i \mid Y_{i-1}=y, X_i \in \cX_i)$, it holds that $$\ratioo_{i,y,\paux}(x) = \frac{\Pr[\aux_i=\paux \mid Y_{i-1} = y, X_i = x]}{\Pr[\aux_i=\paux \mid Y_{i-1} = y, X_i \in \cX_i]}.$$
\end{claim}
\begin{proof}
	A simple calculation yields that
	\begin{align}\label{claim:ratioDifferentVal:1}
	\frac{\pr{\aux_i=\paux \mid Y_{i-1} = y, X_i = x}}{\pr{\aux_i=\paux \mid Y_{i-1} = y, X_i \in \cX_i}}
	= \frac{\pr{X_i = x \mid Y_{i-1} = y, \aux_i=\paux}}{\pr{X_i = x \mid Y_{i-1} = y}} \cdot \frac{\pr{X_i \in \cX_i \mid Y_{i-1} = y}}{\pr{X_i \in \cX_i \mid Y_{i-1} = y, \aux_i = \paux}}
	\end{align}
	Since $x \in \cX_i$, it follows that
	\begin{align}\label{claim:ratioDifferentVal:2}
	\pr{X_i = x \mid Y_{i-1} = y, X_i \in \cX_i} = \frac{\pr{X_i = x \mid Y_{i-1} = y}}{\pr{X_i \in \cX_i \mid Y_{i-1} = y}}
	\end{align}
	and
	\begin{align}\label{claim:ratioDifferentVal:3}
	\pr{X_i = x \mid Y_{i-1} = y, X_i \in \cX_i, \aux_i = \paux} = \frac{\pr{X_i = x \mid Y_{i-1} = y, \aux_i = \paux}}{\pr{X_i \in \cX_i \mid Y_{i-1} = y, \aux_i = \paux}}
	\end{align}
	We conclude that
	\begin{align*}
	\frac{\pr{\aux_i=\paux \mid Y_{i-1} = y, X_i = x}}{\pr{\aux_i=\paux \mid Y_{i-1} = y, X_i \in \cX_i}}
	= \frac{\pr{X_i = x \mid Y_{i-1} = y, X_i \in \cX_i}}{\pr{X_i = x \mid Y_{i-1} = y, X_i \in \cX_i, \aux_i = \paux}}
	= \ratioo_{i,y,\paux}(x).
	\end{align*}
\end{proof}

We now ready to prove \cref{proposition:biasBoundOneRound}.
\begin{proof}[Proof of \cref{proposition:biasBoundOneRound}]
Let $p = \Pr[X_i \in \cX_i] = 1-q$ and $p_\paux = \Pr[X_i \in \cX_i \mid Y_{i-1} = y, \aux_i= \paux] = 1-q_\paux$.
Then,
\begin{align}\label{lemma:DifferenceOneRound:1}
\eo_{i}(y) &=\Pr[Y_{\rnd}\geq 0 \mid Y_{i-1} = y]\\
&= p \cdot \Pr[Y_{\rnd}\geq 0 \mid Y_{i-1} = y,X_i \in \cX_i] +
q \cdot \Pr[Y_{\rnd}\geq 0 \mid Y_{i-1} = y,X_i \notin \cX_i]\nonumber\\
&= p \cdot \ex{x\la X_i\mid x\in \cX_i}{\Pr[Y_{\rnd}\geq 0 \mid Y_{i-1} = y, X_i = x]}
+ q \cdot p',\nonumber\\
&= p \cdot \ex{x\la X_i\mid x\in \cX_i}{\eo_{i+1}(y+x)}
+ q \cdot p',\nonumber\\
&= p_\paux \cdot \ex{x\la X_i\mid x\in \cX_i}{\eo_{i+1}(y+x)} + (p - p_\paux)\cdot \ex{x\la X_i\mid x\in \cX_i}{\eo_{i+1}(y+x)}
+ q \cdot p',\nonumber
\end{align}
for $p'=\Pr[Y_{\rnd}\geq 0 \mid Y_{i-1} = y,X_i \notin \cX_i]$. In addition,
\begin{align}\label{lemma:DifferenceOneRound:2}
\eo_{i}(y,\paux) &= \Pr[Y_{\rnd} \geq 0 \mid Y_{i-1} = y, \aux_i= \paux]\\
&= p_\paux \cdot \Pr[Y_{\rnd} \geq 0 \mid Y_{i-1} = y, \aux_i= \paux,X_i \in \cX_i]
+ q_\paux \cdot \Pr[Y_{\rnd} \geq 0 \mid Y_{i-1} = y, \aux_i= \paux,X_i \notin \cX_i]\nonumber\\
&= p_\paux \cdot \frac{\Pr[Y_{\rnd} \geq 0 \land \aux_i= \paux \mid Y_{i-1} = y,X_i \in \cX_i]}{\Pr[\aux_i= \paux \mid Y_{i-1} = y,X_i \in \cX_i]}
+ q_\paux \cdot p''\nonumber\\
&= p_\paux \cdot \frac{\ex{x\la X_i\mid x\in \cX_i}{\Pr[Y_{\rnd} \geq 0 \land \aux_i= \paux \mid Y_{i-1} = y, X_i = x]}}{\Pr[\aux_i= \paux \mid Y_{i-1} = y,X_i \in \cX_i]}
+ q_\paux \cdot p''\nonumber\\
&= p_\paux \cdot \frac{\ex{x\la X_i\mid x\in \cX_i}{\Pr[Y_{\rnd} \geq 0 \mid Y_{i-1} = y, X_i = x] \cdot \Pr[\aux_i= \paux \mid Y_{i-1} = y, X_i = x]}}{\Pr[\aux_i= \paux \mid Y_{i-1} = y,X_i \in \cX_i]} + q_\paux \cdot p''\nonumber\\
&= p_\paux \cdot \ex{x\la X_i\mid x\in \cX_i}{\eo_{i+1}(y+x) \cdot \frac{\Pr[\aux_i= \paux \mid Y_{i-1} = y, X_i = x]}{\Pr[\aux_i= \paux \mid Y_{i-1} = y,X_i \in \cX_i]}} + q_\paux \cdot p''\nonumber\\
&= p_\paux \cdot \ex{x\la X_i\mid x\in \cX_i}{\eo_{i+1}(y+x) \cdot \ratioo_{i,y,\paux}(x)} + q_\paux \cdot p'',\nonumber
\end{align}
for $p''=\Pr[Y_{\rnd} \geq 0 \mid Y_{i-1} = y, \aux_i= \paux,X_i \notin \cX_i]$, where the last equality holds by \cref{claim:ratioDifferentVal}.
Combing \cref{lemma:DifferenceOneRound:1,lemma:DifferenceOneRound:2} yields that
\begin{align*}
\lefteqn{\size{\eo_{i}(y) - \eo_{i}(y,\paux)}}\\
&\leq p_\paux \cdot \size{\ex{x\la X_i\mid x\in \cX_i}{\eo_{i+1}(y+x) \cdot (1 - \ratioo_{i,y,\paux}(x))}} + \abs{p - p_\paux} + q + q_\paux\\
&\leq \size{\ex{x\la X_i\mid x\in \cX_i}{(\eo_{i+1}(y+x) - \eo_{i+1}(y)) \cdot (1 - \ratioo_{i,y,\paux}(x))}} + \abs{q - q_\paux} + q + q_\paux\\
&\leq \ex{x\la X_i\mid x\in \cX_i}{\size{\eo_{i+1}(y+x) - \eo_{i+1}(y)} \cdot \size{1 - \ratioo_{i,y,\paux}(x)}}
+ 2\cdot(q + q_\paux),
\end{align*}
where the second inequality holds since
\begin{align*}
\lefteqn{\size{\ex{x\la X_i\mid x\in \cX_i}{\eo_{i+1}(y+x) \cdot (1 -\ratioo_{i,y,\paux}(x))}}}\\
&\leq \size{\ex{x\la X_i\mid x\in \cX_i}{(\eo_{i+1}(y) + \eo_{i+1}(y+x) - \eo_{i+1}(y)) \cdot (1 - \ratioo_{i,y,\paux}(x))}}\\
&\leq \size{\ex{x\la X_i\mid x\in \cX_i}{\eo_{i+1}(y) \cdot (1 - \ratioo_{i,y,\paux}(x))}} +
\size{\ex{x\la X_i\mid x\in \cX_i}{(\eo_{i+1}(y+x) - \eo_{i+1}(y)) \cdot (1 - \ratioo_{i,y,\paux}(x))}}\\
&= \size{\eo_{i+1}(y) \cdot (1 - \ex{x\la X_i\mid x\in \cX_i}{\ratioo_{i,y,\paux}(x)})} +
\size{\ex{x\la X_i\mid x\in \cX_i}{(\eo_{i+1}(y+x) - \eo_{i+1}(y)) \cdot (1 - \ratioo_{i,y,\paux}(x))}}\\
&= \size{\ex{x\la X_i\mid x\in \cX_i}{\eo_{i+1}(y) \cdot (1 - 1)}} +
\size{\ex{x\la X_i\mid x\in \cX_i}{(\eo_{i+1}(y+x) - \eo_{i+1}(y)) \cdot (1 - \ratioo_{i,y,\paux}(x))}}\\
&=  \size{\ex{x\la X_i\mid x\in \cX_i}{(\eo_{i+1}(y+x) - \eo_{i+1}(y)) \cdot (1 - \ratioo_{i,y,\paux}(x))}}.
\end{align*}
\end{proof}

\subsection{Bounding the Simple Binomial Game}\label{sec:Bernoulli}
\remove{
Recall that $a\pm b$ stands for the interval $[a-b,a+b]$ and that $f(\cs_1,\ldots,\cs_k) \eqdef \set{f(x_1,\ldots,x_j) \colon x_i\in \cs_i}$, \eg $f(1\pm 0.1) = \set{f(x) \colon x\in [.9,1.1]}$.}
In this section we prove \cref{lemma:ValueBiomialGame} restated below.
\begin{definition}[simple game -- Restatement of \cref{def:ValueBiomialGame}]\label{bounds:def:ValueBiomialGame}
	\SimpleGameDef
\end{definition}

\begin{lemma}[Restatement of \cref{lemma:ValueBiomialGame}]\label{bounds:lemma:ValueBiomialGame}
	\SimpleGameLemma<\cref{bounds:def:ValueBiomialGame}>
\end{lemma}
\begin{proof}
We  view the function $f$ as the composition $g \circ \tau$, where $\tau(i,y)$ outputs $y + t$, for  $t \la\Beroo{\ms{i+1},\eps}$, and $g(y + t)$ outputs $1$ if $y+t \geq 0$, and zero otherwise.
Using \cref{lemma:InformationIncreaseValue}, for bounding the value of  $\game_{f,\rnd,\eps}$ it
suffices to bound that of $\game_{\tau,\rnd,\eps}$. We would also like to assume that $\size{\eps} \leq  4\sqrt{\frac{\log{\rnd}}{\ms{1}}}$. Indeed, if this is not the case, then $\veo_1^- \notin  [\frac1{\rnd^2} ,1 - \frac1{\rnd^2}]$, and the proof follows  by \cref{bounds:lemma:DeterminedGames}. Therefore, in the following we assume that $\size{\eps}  \leq 4 \sqrt{\frac{\log{\rnd}}{\ms{1}}}$.

In the following, we fix $i\in [\rnd-\floor{\log^{2.5}\rnd}]$ and $y \in \cY_{i-1}$, where $\cY_{i-1}$ is according to \cref{lemma:Final}. Let
\begin{align}
\aset_{i,y} = \set{\paux \in \Z \colon \abs{\paux-y} \leq 8\cdot \sqrt{\log{\rnd}\cdot \ms{i}}},
\end{align}
Since $(\aux_i - y)$ is distributed according to $\Beroo{\ms{i},\eps}$ (given that $Y_{i-1} = y$) and since $\size{\eps \cdot \ms{i}} \leq 4\cdot \sqrt{\log{\rnd} \cdot\ms{i}}$, Hoeffding's inequality yields that
\begin{align}\label{eq:Simple:ProbA}
\pr{\aux_i \notin \aset_{i,y} \mid Y_{i-1} = y}
&= \pr{\abs{\aux_i - y} > 8\cdot \sqrt{\log{\rnd}\cdot \ms{i}} \mid Y_{i-1} = y}\\
&\leq \pr{\abs{(\aux_i - y) - \eps\cdot \ms{i}} > 4\cdot \sqrt{\log{\rnd}\cdot \ms{i}} \mid Y_{i-1} = y}\nonumber\\
&\leq 2\cdot \exp\left(-\frac{16 \cdot \ms{i}\log{\rnd}}{2\cdot \ms{i}}\right)\nonumber\\
&\leq \frac1{\rnd^2}.\nonumber
\end{align}
Fix $\paux = y + t \in \aset'_{i,y} \eqdef \aset_{i,y} \bigcap \Supp(\aux_i \mid Y_{i-1}=y, X_i \in \cX_i)$, and let $t_0 = t - \eps \cdot \ms{i}$. Note that $\size{t_0} = \size{t - \eps \cdot \ms{i}} = \size{\paux - y - \eps \cdot \ms{i}} \leq \size{\paux-y} + \size{\eps \cdot \ms{i}} \leq 12\sqrt{\ms{i}\log{\rnd}}$. In addition, note that since $y + t \in \aset'_{i,y}$, there exists $x_0 \in \cX_i$ such that $t-x_0 \in \Supp(\Beroo{\ms{i+1},\eps})$, where $\cX_i$ is according to \cref{def:ratio}. Therefore, we can deduce that $t-x \in \Supp(\Beroo{\ms{i+1},\eps})$ for every $x\in \cX_i$. The latter holds since $\size{t-x} \leq \size{t} + \size{x} < (8+4)\cdot \sqrt{\log{\rnd}\cdot \ms{i}} < \ms{i+1}$ (recalling that $i\in [\rnd-\floor{\log^{2.5}\rnd}]$ for large $\rnd$) and since $x$ has the same parity as $x_0$ (all the elements of $\cX_i$ has the same parity since $\cX_i\subseteq \Supp(X_i)$).

Fix $x \in \cX_i$ and compute
\begin{align}\label{eq:Simple:oneOverRatio}
\frac{1}{\ratioo_{i,y,\paux}(x)}&= \frac{\Pr[\aux_i=y+t \mid Y_{i-1} = y, X_i \in \cX_i]}{\Pr[\aux_i=y+t \mid Y_{i-1} = y, X_i = x]}\\
&= \ex{x' \la X_i\mid x'\in \cX_i}{\frac{\Beroo{\ms{i+1},\eps}(t-x')}{\Beroo{\ms{i+1},\eps}(t-x)}}\nonumber\\
&\in \ex{x' \la X_i\mid x'\in \cX_i}{\exp\left(\frac{- 2\cdot t_0\cdot x + x^2 + 2\cdot t_0\cdot x' - x'^2}{2\cdot \ms{i+1}}\right)} \cdot \left(1 \pm \xi_1 \cdot \frac{\log^{1.5}\rnd}{\sqrt{\ms{i+1}}}\right)\nonumber\\
&\subseteq \left(1 \pm \xi_2\cdot \sqrt{\frac{\log{\rnd}}{\ml{i+1}}}\bigl(1 + \frac{\abs{x}}{\sqrt{\ml{i}}}\bigr)\right)\cdot \left(1 \pm \xi_1 \cdot \frac{\log^{1.5}\rnd}{\sqrt{\ms{i+1}}}\right)\nonumber\\
&\subseteq 1 \pm \xi_3\cdot \sqrt{\frac{\log{\rnd}}{\ml{i+1}}}\bigl(1 + \frac{\abs{x}}{\sqrt{\ml{i}}}\bigr),\nonumber
\end{align}
for some constants $\xi_1,\xi_2,\xi_3\in \R^+$ (independent of the game). The first transition holds by \cref{claim:ratioDifferentVal}, the third one by \cref{prop:binomProbRelation}, and the fourth one by \cref{prop:mainBound}.

Recalling that $i \leq \rnd-\log^{2.5}\rnd$, it follows that
\begin{align}
 \xi_3\cdot \sqrt{\frac{\log{\rnd}}{\ml{i+1}}}\cdot \bigl(1 + \frac{\abs{x}}{\sqrt{\ml{i}}}\bigl) \in O\left(\frac{\log{\rnd}}{\sqrt{\ml{i+1}}}\right) \in o(1)
\end{align}
Since $\frac{1}{1\pm z} \subseteq 1\pm 2z$ for every $z \in (-0.5,0.5)$, we deduce from \cref{eq:Simple:oneOverRatio} that
\begin{align}
\ratioo_{i,y,\paux}(x)
&\in 1 \pm 2\xi_3 \cdot \sqrt{\frac{\log{\rnd}}{\ml{i+1}}} \cdot  \bigl(1 + \frac{\abs{x}}{\sqrt{\ml{i}}}\bigl)
\end{align}
and thus
\begin{align}
\abs{1 - \ratioo_{i,y,\paux}(x)} \leq  2\xi_3 \cdot \sqrt{\frac{\log{\rnd}}{\ml{i+1}}} \cdot  \bigl(1 + \frac{\abs{x}}{\sqrt{\ml{i}}}\bigl)
\end{align}
Finally, since the above holds for every $i\leq \rnd-\log^{2.5}\rnd$, $y\in \cY_{i-1}$, $\paux\in \aset_{i,y}$ and $x\in \cX_i$, and recalling \cref{eq:Simple:ProbA}, we can apply \cref{lemma:Final} to get that $\bias(\game_{\tau,\rnd,\eps}) \leq \xi\cdot \frac{\log^{3}\rnd}{\rnd}$, for some universal constant  $\xi>0$.
\end{proof}

\subsection{Bounding the Hypergeometric Binomial Game}\label{sec:HG}
In this section we prove \cref{lemma:ValueHgGame} restated below.\HyperRecalls
\begin{definition}[hypergeometric game -- Restatement of \cref{def:ValueHgGame}]\label{bounds:def:ValueHgGame}
	\HyperGameDef
\end{definition}

\begin{lemma}[Restatement of \cref{lemma:ValueHgGame}]\label{bounds:lemma:ValueHgGame}
	\HyperGameLemma<\cref{bounds:def:ValueHgGame}>
\end{lemma}

\begin{proof}

We view the function $f$ as $g \circ \tau$, for $\tau(i,y)$ being the output of the following process. A random subset $\cI$ of size $2\cdot\ms{i+1}$ is uniformly chosen from $[2\cdot\ms{1}]$, where the output of $\tau$ is set to $(\w(\vct_{\cI}),y+t)$, for $\vct \in \oo^{2\cdot\ms{1}}$ with $\w(\vct) = p$ and for $t \la \Hyp{2\ms{i+1},\w(\vct,\cI),\ms{i+1}}$. The function $g$ on input $(p',y')$ outputs one if $y' \geq 0$, and zero otherwise. Since $\pr{g \circ \tau(i,y) = 1} = \pr{f(i,y) = 1}$, by \cref{lemma:InformationIncreaseValue} it suffices to bound the bias of the game $\game_{\tau,\rnd,\eps}$. In addition, as in the proof of  \cref{bounds:lemma:ValueBiomialGame}, we can assume \wlg that $\size{\eps} \leq  4\sqrt{\frac{\log{\rnd}}{\ms{1}}}$. Fix $i\in [\rnd-\floor{\log^{2.5}\rnd}]$ and $y \in \cY_{i-1}$, where $\cY_{i-1}$ is according to \cref{lemma:Final}. Let
\begin{align*}
\aset_{i,y} = \set{(p',y') \in \Z^2 \colon  \abs{p'},\abs{y'-y} \leq (\const+8) \sqrt{\log{\rnd} \cdot \ms{i+1}}},
\end{align*}

Since $\aux_i = (p',y+t)$ for $p' \la \Hyp{2\ms{1},p,2\ms{i+1}}$ and $t \la \Beroo{\ml{i},\eps} + \Hyp{2\ms{i+1},p',\ms{i+1}}$ (given that $Y_{i-1} = y$), it follows that
\begin{align*}
\lefteqn{\pr{\abs{\aux_i[0]} > (\const+8) \sqrt{\log{\rnd} \cdot \ms{i+1}}}}\\
&\leq \pr{\abs{\aux_i[0] - \frac{p\cdot \ms{i+1}}{\ms{1}}} > 8 \sqrt{\log{\rnd} \cdot \ms{i+1}}}\\
&\leq \exp\left(-\frac{64 \cdot \ms{i+1}\log{\rnd}}{2\cdot \ms{i+1}}\right)\\
&\leq \frac{1}{\rnd^4},
\end{align*}
where the second inequality holds by Hoeffding's inequality for hypergeometric distribution (\cref{fact:hyperHoeffding}). In addition, given that $\aux_i[0] = p'$ for $\abs{p'} \leq (\const+8) \sqrt{\log{\rnd} \cdot \ms{i+1}}$, it holds that $\bigl(\aux_i[1] - (y+X_i)\bigr)$ is distributed according to $\Hyp{2\ms{i+1},p',\ms{i+1}}$. This yields that
\begin{align*}
\lefteqn{\pr{\abs{\aux_i[1] - y} > (\const+8) \sqrt{\log{\rnd} \cdot \ms{i+1}} \mid Y_{i-1} = y}}\\
&\leq \pr{\abs{\aux_i[1] - (y+X_i)} > (\const+7) \sqrt{\log{\rnd} \cdot \ms{i+1}} \mid Y_{i-1} = y}\\
&\leq \pr{\abs{\bigl(\aux_i[1] - (y+X_i)\bigr) - \frac{p'}{2}} > 3 \sqrt{\log{\rnd} \cdot \ms{i+1}} \mid Y_{i-1} = y}\\
&\leq \exp\left(-\frac{9 \cdot \ms{i+1}\log{\rnd}}{2\cdot \ms{i+1}}\right)\\
&\leq \frac{1}{\rnd^4}.
\end{align*}
The first inequality holds since $\abs{X_i} \leq \ml{i} < \sqrt{\log{\rnd} \cdot \ms{i+1}}$, the second one holds since $\frac{\abs{p'}}{2} < \frac{\const+8}{2} \sqrt{\log{\rnd} \cdot \ms{i+1}}$, and the third one by Hoeffding's inequality for hypergeometric distribution (\cref{fact:hyperHoeffding}). It follows that
\begin{align}\label{eq:Hyp:ProbA}
\pr{\aux_i \notin \aset_{i,y} \mid Y_{i-1}=y} \leq \frac{1}{\rnd^4} + \frac{1}{\rnd^4} < \frac1{\rnd^2}
\end{align}

Fix $\paux = (p',y +t) \in \aset'_{i,y} \eqdef  \aset_{i,y} \bigcap \Supp(\aux_i \mid Y_{i-1} = y, X_i \in \cX_i)$, where $\cX_i$ is according to \cref{def:ratio}. Note that by the same arguments introduced in the analogous case in \cref{sec:Bernoulli}, it holds that $t-x \in \Supp(\Hyp{2\ms{i+1},p',\ms{i+1}})$ for every $x \in \cX_i$.

Fix $x \in \cX_i$ and compute
\begin{align}\label{eq:Hyper:oneOverRatio}
\frac{1}{\ratioo_{i,y,\paux}(x)}&= \frac{\Pr[\aux_i=(p',y + t) \mid Y_{i-1} = y, X_i \in \cX_i]}{\Pr[\aux_i=(p',y + t) \mid Y_{i-1} = y, X_i=x]}\\
&= \ex{x' \la X_i\mid x'\in \cX_i}{\frac{\Hyp{2\ms{i+1},p',\ms{i+1}}(t-x')}{\Hyp{2\ms{i+1},p',\ms{i+1}}(t-x)}}\nonumber\\
&\in \ex{x' \la X_i\mid x'\in \cX_i}{\exp\left(\frac{-2(t-\frac{p'}2)x + x^2 + 2(t-\frac{p'}2)x' - x'^2}{\ms{i+1}}\right)}\cdot
\left(1 \pm  \varphi_1(\const) \cdot \frac{\log^{1.5}\rnd}{\sqrt{\ms{i+1}}}\right)\nonumber\\
&\subseteq \left(1 \pm \varphi_2(\const) \sqrt{\frac{\log{\rnd}}{\ml{i+1}}}\bigl(1 + \frac{\abs{x}}{\sqrt{\ml{i}}}\bigr)\right)\cdot \left(1 \pm \varphi_1(\const) \cdot \frac{\log^{1.5}\rnd}{\sqrt{\ms{i+1}}}\right)\nonumber\\
&\subseteq 1 \pm \varphi_3(\const) \sqrt{\frac{\log{\rnd}}{\ml{i+1}}}\bigl(1 + \frac{\abs{x}}{\sqrt{\ml{i}}}\bigr),\nonumber
\end{align}
for some functions $\varphi_1,\varphi_2,\varphi_3\colon \R^+ \mapsto \R^+$ (independent of the game). The first transition holds by \cref{claim:ratioDifferentVal}, the third one by \cref{prop:hyperProbRelation} and the fourth one by \cref{prop:mainBound}.

Recalling that $i \leq \rnd-\log^{2.5}\rnd$, it follows that
\begin{align}
 \varphi_3(\const) \cdot \sqrt{\frac{\log{\rnd}}{\ml{i+1}}}\cdot \bigl(1 + \frac{\abs{x}}{\sqrt{\ml{i}}}\bigl) \in o(1)
\end{align}
Since $\frac{1}{1\pm z} \subseteq 1\pm 2z$ for every $z \in (-0.5,0.5)$, we deduce from \cref{eq:Hyper:oneOverRatio} that
\begin{align}
\ratioo_{i,y,\paux}(x)
&\in 1 \pm 2\varphi_3(\const) \sqrt{\frac{\log{\rnd}}{\ml{i+1}}} \cdot  \bigl(1 + \frac{\abs{x}}{\sqrt{\ml{i}}}\bigr)
\end{align}
and thus
\begin{align}
\abs{1 - \ratioo_{i,y,\paux}(x)}
&\leq 2\varphi_3(\const) \sqrt{\frac{\log{\rnd}}{\ml{i+1}}}\cdot \bigl(1 + \frac{\abs{x}}{\sqrt{\ml{i}}}\bigr)
\end{align}
Finally, since the above holds for every $i\leq \rnd-\log^{2.5}\rnd$, $y\in \cY_{i-1}$, $\paux\in \aset_{i,y}$ and $x\in \cX_i$, and recalling \cref{eq:Hyp:ProbA}, we can apply \cref{lemma:Final} to get that $\bias(\game_{\tau,\rnd,\eps}) \leq \varphi(\const) \cdot \frac{\log^{3}\rnd}{\rnd}$, for some universal function $\varphi\colon \R^+ \mapsto \R^+$.
\end{proof}

\subsection{Bounding the Vector Binomial Game}\label{sec:Vector}
In this section we prove \cref{lemma:ValueVectorGame} restated below.\VectorRecalls

\begin{definition}[vector game -- Restatement of \cref{def:ValueVectorGame}]\label{bounds:def:ValueVectorGame}
	\VectorGameDef
\end{definition}

\begin{lemma}[Restatement of \cref{lemma:ValueVectorGame}]\label{bounds:lemma:ValueVectorGame}
	\VectorGameLemma<\cref{bounds:def:ValueVectorGame}>
\end{lemma}

\begin{proof}
For $i \in [\rnd]$ and $y \in \Z$, let $\eps_{i}(y) \eqdef \sBias{\ms{1}}{\eo_{i+1}(y)}$  (recall that $\eo_{i+1}(y) = \vBeroo{\ms{i+1},0}(-y)$) and let $q = \const \cdot \ms{1}$.
Fix $i\in[\rnd- \floor{\log^{2.5}\rnd}]$ and $y\in \cY_{i-1}$, where $\cY_{i-1}$ is according to \cref{lemma:Final}. Note that
\begin{align}
\pr{f(i,y) = v} = 2^{-q}\cdot (1 + \eps_{i}(y))^{\frac q2 + \frac{\w(v)}2} \cdot (1 - \eps_{i}(y))^{\frac q2 -\frac{\w(v)}2}
\end{align}
for every $v \in \oo^{q}$. Let
\begin{align}
\aset_{i,y} = \set{v \in \oo^{q} \colon  \size{\w(v)} \leq \sqrt{d\cdot \log{\rnd} \cdot q}},
\end{align}
for $d = d(\const)$ to be determined by the analysis. In the following we let ${s_i} = \ms{i+1}\cdot \ms{1}$. Since $\rnd$ is large, \cref{prop:epsDiff} yields that $\eps_{i}(y+x) \in \frac{y+x}{\sqrt{{s_i}}} \pm \frac{\log^{2}\rnd}{\sqrt{{s_i}}}$ for every $x \in \Supp(X_i)$. Since $y\in \cY_{i-1}$, it follows that $\size{y} \leq 4\sqrt{\log{\rnd} \cdot \ms{i}}$. Therefore, $\size{\frac{y+x}{\sqrt{s_i}}} \leq \frac{4\sqrt{\log{\rnd} \cdot \ms{i}}+\ml{i}}{\sqrt{\ms{i+1}\cdot \ms{1}}} \leq \frac{5\sqrt{\log{\rnd} \cdot \ms{i}}}{\sqrt{\ms{i+1}\cdot \ms{1}}} \leq \frac{6\sqrt{\log{\rnd}}}{\sqrt{\ms{1}}} = \sqrt{\frac{36\const\cdot \log{\rnd}}{q}}$,
and thus,
$\size{\eps_{i}(y+x)} \leq (36\const + 1)\cdot \sqrt{\frac{\log{\rnd}}{q}}$ for every $x \in \Supp(X_i)$. By setting $d = (5 + 72\const)^2$, Hoeffding's bound yields that the following holds for every $x \in \Supp(X_i)$.
\begin{align}
\lefteqn{\pr{\aux_i \notin \aset_{i,y} \mid Y_i = y+x}}\nonumber\\
&= \ppr{z \la \Beroo{q,\eps_{i}(y+x)}}{\abs{z} > \sqrt{d\cdot \log{\rnd} \cdot q}}\nonumber\\
&\leq \ppr{z \la \Beroo{q,\eps_{i}(y+x)}}{\abs{z - 2q\cdot \eps_{i}(y+x)} > \sqrt{d\cdot \log{\rnd} \cdot q} - 2q\cdot \eps_{i}(y+x)}\nonumber\\
&\leq \ppr{z \la \Beroo{q,\eps_{i}(y+x)}}{\abs{z - 2q\cdot \eps_{i}(y+x)} > \sqrt{d\cdot \log{\rnd} \cdot q} - 2\cdot(36\const+1)\cdot\sqrt{\log{\rnd} \cdot q}}\nonumber\\
&= \ppr{z \la \Beroo{q,\eps_{i}(y+x)}}{\abs{z - 2q\cdot \eps_{i}(y+x)} > 3\cdot\sqrt{\log{\rnd} \cdot q}}\nonumber\\
&\leq \frac1{\rnd^2}.\nonumber
\end{align}
Thus,
\begin{align}\label{eq:Vector:ProbA}
\pr{\aux_i \notin \aset_{i,y} \mid Y_{i-1}=y}
= \ex{x \la X_i}{\pr{\aux_i \notin \aset_{i,y} \mid Y_i = y+x}}
\leq \frac1{\rnd^2}
\end{align}
Fix $x \in \cX_i$ and $v \in \aset'_{i,y} \eqdef  \aset_{i,y} \bigcap \Supp(\aux_i \mid Y_{i-1} = y, X_i \in \cX_i)$, where $\cX_i$ is as defined in \cref{lemma:Final}.
Compute
\begin{align}
\Pr[\aux_i=v \mid Y_i = y+x]
&= 2^{-q}\cdot (1 + \eps_{i}(y+x))^{\frac q2 + \frac{\w(v)}2}(1 - \eps_{i}(y+x))^{\frac q2 -\frac{\w(v)}2}\\
&= 2^{-q}\cdot (1 - \eps_{i}^2(y+x))^{\frac q2 - \frac{\w(v)}2}(1 + \eps_{i}(y+x))^{\w(v)}\nonumber
\end{align}
Since $1 + z \leq e^z$ for $z \in \R$, it holds that
\begin{align}
\Pr[\aux_i=v \mid Y_i = y+x]
&\leq 2^{-q}\cdot \exp\left(-\eps_{i}^2(y+x)\cdot(\frac q2 - \frac{\w(v)}2) + \eps_{i}(y+x) \cdot \w(v)\right)\label{eq:Vector:upperbound}
\end{align}
Since $\rnd$ is large, \cref{prop:epsDiff} yields that  $\eps_{i}(y+x) \in (-\frac12,\frac12)$. Using the inequality $1 + z \geq e^{z-z^2}$ for $z \in(-\frac12,\frac12)$, we deduce that
\begin{align}\label{eq:Vector:lowerbound}
\Pr[\aux_i=v \mid Y_i = y+x]
&\geq 2^{-q}\cdot e^{(-\eps_{i}^2(y+x)-\eps_{i}^4(y+x))(\frac q2 - \frac{\w(v)}2)}\cdot e^{(\eps_{i}(y+x)-\eps_{i}^2(y+x)) \cdot \w(v)}\\
&= 2^{-q}\cdot e^{-\eps_{i}^2(y+x)\cdot(\frac q2 - \frac{\w(v)}2)}\cdot e^{\eps_{i}(y+x) \cdot \w(v)}
\cdot e^{-\eps_{i}^4(y+x)\cdot(\frac q2 - \frac{\w(v)}2) - \eps_{i}^2(y+x)\cdot\w(v)}\nonumber\\
&\geq 2^{-q}\cdot \exp\left(-\eps_{i}^2(y+x)\cdot(\frac q2 - \frac{\w(v)}2) + \eps_{i}(y+x) \cdot \w(v) \right)\cdot (1 - \error(x)),\nonumber
\end{align}
for
\begin{align}
\error(x) \eqdef \size{1 - \exp\left(-\eps_{i}^4(y+x)\cdot(\frac q2 - \frac{\w(v)}2) - \eps_{i}^2(y+x) \cdot \w(v)\right)}
\end{align}
Using \cref{eq:Vector:upperbound,eq:Vector:lowerbound}, we can now write
\begin{align}
 \Pr[\aux_i=v \mid Y_i = y+x] \in  2^{-q}\cdot \exp\left(-\eps_{i}^2(y+x)\cdot(\frac q2 - \frac{\w(v)}2) + \eps_{i}(y+x) \cdot \w(v) \right) (1 \pm \error(x))
\end{align}

Let $x' \in \cX_i$, and assume \wlg that $\error(x) \geq \error(x')$. We show next that $\error(x) \in o(1)$. Hence, since $\frac{1\pm z}{1\pm z} \subseteq 1 \pm 4z$ for every $z \in [0,\frac12]$, it holds that
\begin{align}
\lefteqn{\frac{\Pr[\aux_i=v \mid Y_i = y+x']}{\Pr[\aux_i=v \mid Y_i = y+x]}}\\
&\in \frac{\exp\left(-\eps_{i}^2(y+x')\cdot (\frac q2 - \frac{\w(v)}2) + \eps_{i}(y+x') \cdot \w(v) \right)}
{\exp\left(-\eps_{i}^2(y+x)(\frac q2 - \frac{\w(v)}2) + \eps_{i}(y+x) \cdot \w(v) \right)}
\cdot (1 \pm 4\cdot \error(x))\nonumber\\
&= \exp\left((\eps_{i}(y+x)-\eps_{i}(y+x'))\left[
(\eps_{i}(y+x)+\eps_{i}(y+x'))(\frac q2 - \frac{\w(v)}2) - \w(v)
\right]\right)(1 \pm 4\cdot \error(x))\nonumber\\
&\subseteq \exp\left(\bigl(\frac{x-x'}{\sqrt{{s_i}}} \pm \frac{\log^{2}\rnd}{\sqrt{{s_i}}}\bigr)\left[
\bigl(\frac{2y+x+x'}{\sqrt{{s_i}}} \pm \frac{\log^{2}\rnd}{\sqrt{{s_i}}}\bigr)
\cdot (\frac q2 - \frac{\w(v)}2) - \w(v)\right]\right)\cdot (1 \pm 4\cdot \error(x))\nonumber\\
&\subseteq \exp\left(\frac{x-x'\pm \log^{2}\rnd}{\ms{i+1}}\cdot \left[
\frac{2y+x+x'\pm \log^{2}\rnd}{\ms{1}} \cdot
(\frac q2 - \frac{\w(v)}2) \pm \sqrt{\frac{\ms{i+1}}{\ms{1}}}\cdot\w(v)\right]\right)\cdot (1 \pm 4\cdot \error(x)),\nonumber
\end{align}
and therefore,
\begin{align}
\frac{\Pr[\aux_i=v \mid Y_i = y+x']}{\Pr[\aux_i=v \mid Y_i = y+x]}
&\in \exp\left(\bigl(\frac{x-x'}{\ms{i+1}} \pm \frac{\log^{2}\rnd}{\ms{i+1}}\bigr) \cdot \alpha\right)\cdot (1 \pm 4\cdot \error(x))\label{sec:vector:relation_with_alpha}
\end{align}
for some $\alpha \in (\frac{2y+x+x'}{\ms{1}} \pm \frac{\log^{2}\rnd}{\ms{1}}) \cdot (\frac q2 - \w(v)) \pm \sqrt{\frac{\ms{i+1}}{\ms{1}}}\cdot \w(v)$. The
third transition of the previous calculation holds by \cref{prop:epsDiff}.
By taking large enough $d' = d'(\const) >0$, we can bound $\size{\alpha}$ and $\error(x)$ by
\begin{align}
\abs{\alpha}  \leq d' \cdot \sqrt{\log{\rnd} \cdot \ms{i+1}}
\end{align}
and
\begin{align}
\error(x)
&= \abs{1 - \exp\left(-\eps_{i}^4(y+x)(\frac q2 - \frac{\w(v)}2) - \eps_{i}^2(y+x) \cdot \w(v)\right)}\nonumber\\
&\leq \max\left(\size{ 1 - \exp\left(-\left(\frac{y+x}{\sqrt{{s_i}}} \pm \frac{\log^{2}\rnd}{\sqrt{{s_i}}}\right)^4(\frac q2 - \frac{\w(v)}2) - \left(\frac{y+x}{\sqrt{{s_i}}} \pm \frac{\log^{2}\rnd}{\sqrt{{s_i}}}\right)^2\w(v)\right)}\right)\nonumber\\
&\leq 1 - \exp\left(-\left(3\cdot \sqrt{\frac{\log{\rnd}}{\ms{1}}}\right)^4\cdot (\const + 1)\cdot \ms{1} -
\left(3\cdot \sqrt{\frac{\log{\rnd}}{\ms{1}}}\right)^2\cdot \sqrt{d \cdot \ms{1} \cdot \log{\rnd} }\right)\nonumber\\
&\leq 1 - (1 - \frac{d' \cdot \log^{1.5}\rnd}{\sqrt{\ms{1}}})\nonumber\\
&= d' \cdot \frac{\log^{1.5}\rnd}{\sqrt{\ms{1}}},\nonumber
\end{align}
where the second transition holds by \cref{prop:epsDiff} and the first inequality holds by the bounds on $\abs{y}$, $\abs{x}$ and $\abs{w(v)}$.
Since $\frac{\log^{2}\rnd}{\ms{i+1}}\cdot \alpha \in o(1)$ and since $e^y \in 1 \pm 2y$ for $y \in (-0.5,0.5)$, it follows that
\begin{align}
\exp\left(\frac{\log^{2}\rnd}{\ms{i+1}}\cdot \alpha\right) \\
&\leq 1 + 2\cdot \frac{\log^{2}\rnd}{\ms{i+1}}\cdot \alpha \nonumber\\
&\leq 1 + 2\cdot \frac{\log^{2}\rnd}{\ms{i+1}}\cdot d'\cdot \sqrt{\log{\rnd} \cdot \ms{i+1}} \nonumber\\
&\leq 1 + 2d' \cdot \frac{\log^{2.5} \rnd}{\sqrt{\ms{i+1}}}. \nonumber
\end{align}
Therefore, \cref{sec:vector:relation_with_alpha} yields that
\begin{align}\label{eq:Vector:1}
 \frac{\Pr[\aux_i=v \mid Y_i = y+x']}{\Pr[\aux_i=v \mid Y_i = y+x]} \in \exp\left(\frac{x-x'}{\ms{i+1}} \cdot \alpha\right)\cdot \left(1 \pm \frac{\log^{3}\rnd}{\sqrt{\ms{i+1}}}\right)
\end{align}
and thus
\begin{align}\label{eq:Vector:oneOverRatio}
\frac{1}{\ratioo_{i,y,v}(x)}&= \frac{\Pr[\aux_i=v \mid Y_{i-1} = y, X_i \in \cX_i]}{\Pr[\aux_i=v \mid Y_{i-1} = y, X_i=x]}\\
&= \ex{x' \la X_i\mid x'\in \cX_i}{\frac{\Pr[\aux_i=v \mid Y_i = y + x']}{\Pr[\aux_i=v \mid Y_i = y+x]}}\nonumber\\
&\in \ex{x' \la X_i\mid x'\in \cX_i}
{\exp\left(\frac{\alpha (x - x')}{\ms{i+1}}\right) \cdot \left(1 \pm \frac{\log^{3}\rnd}{\ms{i+1}}\right)}\nonumber\\
&\subseteq \left(1 \pm \varphi_1(d) \cdot \sqrt{\frac{\log{\rnd}}{\ml{i+1}}}\cdot(1 + \frac{\abs{x}}{\sqrt{\ml{i}}})\right)\cdot \left(1 \pm \frac{\log^{3}\rnd}{\ms{i+1}}\right)\nonumber\\
&\subseteq 1 \pm \varphi_2(d) \sqrt{\frac{\log{\rnd}}{\ml{i+1}}}\cdot \bigl(1 + \frac{\abs{x}}{\sqrt{\ml{i}}}\bigl),\nonumber
\end{align}
for some universal functions $\varphi_1,\varphi_2\colon \R^+ \mapsto \R^+$. The first transition holds by \cref{claim:ratioDifferentVal}, the third one by \cref{eq:Vector:1} and the fourth one by \cref{prop:mainBound}. Recalling that $i \leq \rnd-\log^{2.5}\rnd$, it  follows that
\begin{align}
 \sqrt{\frac{\log{\rnd}}{\ml{i+1}}}\cdot \bigl(1 + \frac{\abs{x}}{\sqrt{\ml{i}}}\bigl) \in o(1)
\end{align}
Since $\frac{1}{1\pm z} \subseteq 1\pm 2z$ for every $z \in (-0.5,0.5)$, we deduce from \cref{eq:Vector:oneOverRatio} that
\begin{align}
\ratioo_{i,y,v}(x)
&\in 1 \pm 2\varphi_2(d) \cdot \sqrt{\frac{\log{\rnd}}{\ml{i+1}}} \cdot  \bigl(1 + \frac{\abs{x}}{\sqrt{\ml{i}}}\bigl)\\
&= 1 \pm \varphi_3(\const) \cdot \sqrt{\frac{\log{\rnd}}{\ml{i+1}}} \cdot  \bigl(1 + \frac{\abs{x}}{\sqrt{\ml{i}}}\bigl),\nonumber
\end{align}
where $\varphi_3(\const) = 2\cdot \varphi_2(d(\const))$. Thus
\begin{align}
\abs{1 - \ratioo_{i,y,v}(x)}\leq \varphi_3(\const) \cdot\sqrt{\frac{\log{\rnd}}{\ml{i+1}}}\cdot (1 + \frac{\abs{x}}{\sqrt{\ml{i}}})
\end{align}
Finally, since the above holds for every $i\leq \rnd-\log^{2.5}\rnd$, $y\in \cY_{i-1}$, $v\in \aset_{i,y}$ and $x\in \cX_i$, and recalling \cref{eq:Vector:ProbA}, we can apply \cref{lemma:Final}  to get that $\bias(\game_{f,\rnd,\eps}) \le \varphi(\const) \cdot \frac{\log^{3}\rnd}{\rnd}$, for some universal function $\varphi\colon \R^+ \mapsto \R^+$.
\end{proof}

\subsection*{Acknowledgment}
We are very grateful to Yuval Ishai, Yishay Mansour, Eran Omri and Alex Samorodnitsky for very useful discussions. We also thank Eran for encouraging us to tackle this beautiful problem. We also thans the anonymous referees for their very useful comments.

\bibliographystyle{abbrvnat}
\bibliography{crypto}

\appendix
\section{Missing Proofs}\label{sec:missinProofs}
This section contains missing proofs for statements given in \cref{sec:prelim:BasicInq,sec:prelim:Binomial,sec:prelim:HypGeo}.

\subsection{Basic Inequalities}\label{app:missinProofs:BasicInq}

\remove{

\begin{fact}[Restatement of \cref{fact:statisticaldist}]\label{app:missinProofs:statisticaldist}
	\factstatisticaldist
\end{fact}
\begin{proof}
Compute
\begin{align*}
\SD\bigl((P,Q),(P',Q')\bigr)
&= \frac12 \sum_{u \in \Uni, v \in \Uni'}\size{P(u)Q(v) - P'(u)Q'(v)}\\
&\leq \frac12 \sum_{u \in \Uni, v \in \Uni'}P(u)\size{Q(v) - Q'(v)} + \frac12 \sum_{u \in \Uni, v \in \Uni'}\size{P(u) - P'(u)}Q'(v)\\
&= \frac12 \sum_{v \in \Uni'}\size{Q(v) - Q'(v)}\cdot \sum_{u \in \Uni}P(u) + \frac12 \sum_{u \in \Uni}\size{P(u) - P'(u)} \cdot \sum_{v\in \Uni'}Q'(v)\\
&= \SD(P,P') + \SD(Q,Q')
\end{align*}
\end{proof}

}

\begin{proposition}[Restatement of \cref{claim:linearproblem}]\label{app:missinProofs:linearproblem}
\claimlinearproblem
\end{proposition}
\begin{proof}
We prove the proposition by showing that for every set $\cs = \set{p_j}_{j=k}^{n}$ satisfying the proposition's constrains, it holds that $\val(\cs) \eqdef \sum_{j=k}^{n} \frac{p_j}{(n+1-j)} \leq \sum_{j=k}^{n}\frac{\alpha}{(n+1-j)}$. Let $\cs = \set{p_j}_{j=k}^{n}$ be a set that satisfying the proposition's constrains with maximal $\val(\cs)$. Assume not all elements of $\cs$ equal $\alpha$, and let $\ist \in \set{k,k+1, \ldots, n}$ be the largest index such that $p_{\ist} \neq \alpha$. By the proposition's constrains, it follows that
$\sum_{j=\ist}^{n}p_j \leq \alpha \cdot (n +1-\ist)$. Since $\sum_{j=\ist+1}^{n}p_j = \alpha \cdot (n-\ist)$, it follows that $p_{\ist} + \alpha (n - \ist) \leq \alpha \cdot (n +1-\ist)$, and thus $
p_{\ist} \leq \alpha$. Since we assume $p_{\ist} \neq \alpha$, it follows that $p_{\ist} < \alpha$.

Assume $\ist = k$, then by changing $p_{\ist}$ to $\alpha$, we get a set $\cs'$ with $\val(\cs') > \val(\cs)$ that fulfills the proposition's constrains, in contradiction to the maximality of $\cs$.

Assume $\ist > k$ and let $\delta = \alpha - p_{\ist} > 0$.
Let $\cs' = \set{p_j'}_{j=k}^{n}$ defined by
$$p_j' = \left\{
           \begin{array}{ll}
             p_j + \delta, & j = \ist,\\
             p_j-\delta , & j = \ist-1,\\
             p_j,        &\hbox{otherwise.}
           \end{array}
         \right.
 $$
Note that $\cs'$ fulfills proposition's constrains, and
$$\val(\cs') = \sum_{j=k}^{n}\frac{p_j^{'}}{n - j + 1} = \sum_{j=k}^{n}\frac{p_j}{n - j + 1} + \frac{\delta}{n - \ist + 1} - \frac{\delta}{n - \ist + 2}> \sum_{j=k}^{n}\frac{p_j}{n - j + 1} = \val(\cs),$$
in contraction to the maximality of $\cs$.
\end{proof}

\subsection{Facts About the Binomial Distribution}\label{app:missinProofs:Binomial}

Recall that for $a\in \R$ and $b\geq 0$,  $a\pm b$ denotes for the interval $[a-b,a+b]$, and that given sets $\cs_1,\ldots,\cs_k$ and $k$-input function $f$,  $f(\cs_1,\ldots,\cs_k) = \set{f(x_1,\ldots,x_j) \colon x_i\in \cs_i}$, \eg $f(1\pm 0.1) = \set{f(x) \colon x\in [.9,1.1]}$.

We use the following estimation of the binomial coefficient.
\begin{proposition}\label{app:missinProofs:binomCoeffEstimation}
Let $n \in \N$ and $t \in \Z$ be such that $\abs{t} \leq n^{\frac{3}{5}}$ and $\frac{n+t}{2} \in (n)$. Then
\begin{align*}
\binom{n}{\frac{n+t}{2}} \cdot 2^{-n} \in (1 \pm \error) \cdot  \sqrt{\frac{2}{\pi}} \cdot \frac{1}{\sqrt{n}} \cdot e^{-\frac{t^2}{2n}},
\end{align*}
for  $\error = \xconst \cdot (\frac{\abs{t}^3}{n^2} + \frac{1}{n})$ and a universal constant $\xconst$.
\end{proposition}
\begin{proof}
In the following we focus on $n \geq 200$, smaller $n$'s are handled by setting the value of $\xconst$ to be large enough on these values.
We also assume that $n$ and $t$ are even, the proof of the odd case is analogous.
Let $m \eqdef \frac{n}{2} \geq 100$ and $k \eqdef \frac{t}{2}$.
Stirling's formula states that for every $\ell \in \N$ it holds that $1 \leq
\frac{\ell!}{\sqrt{2\pi \ell} \cdot (\frac{\ell}{e})^\ell} \leq e^{\frac{1}{12\ell}}$ which implies $\ell! \in (1 \pm \frac{1}{\ell})\sqrt{2\pi \ell} \cdot \ell^\ell \cdot e^{-\ell}$.
Compute
\begin{align*}
\binom{2m}{m + k}
&= \frac{(2m)!}{(m+k)!(m-k)!}\\
&\in \frac{(1 \pm \frac{1}{2m})\sqrt{2\pi \cdot 2m}(2m)^{2m} e^{-2m}}
{(1 \pm \frac{1}{m + k})\sqrt{2\pi (m+k)}(m+k)^{m+k} e^{-(m+k)} \cdot
(1 \pm \frac{1}{m-k})\sqrt{2\pi (m-k)}(m-k)^{m-k} e^{-(m-k)}}\nonumber\\
&\subseteq \frac{\sqrt{2\pi \cdot 2m}(2m)^{2m} e^{-2m}}{\sqrt{2\pi (m+k)}(m+k)^{m+k} e^{-(m+k)} \cdot \sqrt{2\pi (m-k)}(m-k)^{m-k} e^{-(m-k)}}\cdot (1 \pm \frac{20}{m})\nonumber\\
&= \frac{(2m)^{2m+\frac{1}{2}}}{\sqrt{2 \pi} \cdot (m+k)^{m+k+\frac{1}{2}} \cdot (m-k)^{m-k+\frac{1}{2}}}\cdot (1 \pm \frac{20}{m})\nonumber\\
&=2^{2m}\cdot \frac{1}{\sqrt{\pi m}\cdot (1 + \frac{k}{m})^{m+k+\frac{1}{2}} \cdot (1 - \frac{k}{m})^{m-k+\frac{1}{2}} }\cdot (1 \pm \frac{20}{m})\nonumber\\
&= 2^{2m}\cdot \frac{1}{\sqrt{\pi m}\cdot
(1 - \frac{k^2}{m^2})^{m-k+\frac{1}{2}} \cdot (1 + \frac{k}{m})^{2k}}\cdot (1 \pm \frac{20}{m}),\nonumber
\end{align*}
where the third transition holds by the bound on $m$ and $k$ which yields $\frac{(1 \pm \frac{1}{m})}{(1 \pm \frac{1}{m+k})(1 \pm \frac{1}{m-k})} \subseteq (1 \pm \frac{20}{m})$.
Since $1 + x \in e^{x\pm x^2}$ for $x \in (-0.5,0.5)$, it follows that
\begin{align}
\binom{n}{\frac{n+t}{2}} \cdot 2^{-n}
&= \binom{2m}{m + k} \cdot 2^{-2m}\\
&\in \frac{1}{\sqrt{\pi m} \cdot
e^{(-\frac{k^2}{m^2}\pm \frac{k^4}{m^4})(m-k+\frac{1}{2})} \cdot
e^{(\frac{k}{m}\pm \frac{k^2}{m^2})\cdot 2k}}\cdot (1 \pm \frac{20}{m})\nonumber\\
&= \frac{1}{\sqrt{\pi m}} \cdot e^{-\frac{k^2}{m}} \cdot e^{-\frac{k^3}{m^2} \pm \frac{3\abs{k}^3}{m^2} + \frac{k^2}{2m^2} \pm \frac{k^4}{m^4}(m-k+\frac{1}{2})}\cdot (1 \pm \frac{20}{m})\nonumber\\
&\subseteq \frac{1}{\sqrt{\pi m}} \cdot e^{-\frac{k^2}{m}} \cdot e^{\pm \frac{5\abs{k}^3}{m^2}}\cdot (1 \pm \frac{20}{m})\nonumber\\
&\subseteq \frac{1}{\sqrt{\pi m}} \cdot e^{-\frac{k^2}{m}} \cdot(1 \pm \frac{10\abs{k}^3}{m^2})\cdot (1 \pm \frac{20}{m})\nonumber\\
&\subseteq \frac{1}{\sqrt{\pi}} \cdot \bigl(1 \pm 20\cdot (\frac{\abs{k}^3}{m^2} + \frac{1}{m})\bigl)\cdot \frac{1}{\sqrt{m}} \cdot e^{-\frac{k^2}{m}},\nonumber\\
&\subseteq \sqrt{\frac{2}{\pi}} \cdot \bigl(1 \pm 40\cdot (\frac{\abs{t}^3}{n^2} + \frac{1}{n})\bigl) \cdot \frac{1}{\sqrt{n}} \cdot e^{-\frac{t^2}{2n}}\nonumber
\end{align}
where the third transition holds by the bounds on $m$ and $k$, and the fourth one holds since $\frac{4\abs{k}^3}{m^2} < 1$ and since $e^x \in 1 \pm 2x$ for every $\abs{x} < 1$.
\end{proof}

Recall that for $n\in \N$ and $\eps \in [-1,1]$, we let $\Beroo{n,\eps}$ be the binomial distribution induced by the sum of $n$ independent random variables over $\oo$, each takes the value $1$ with probability $\frac{1}{2}(1+\eps)$ and $-1$ otherwise. The following proposition uses the previous estimation for the binomial coefficient for achieving an estimation for the binomial probability $\Beroo{n,\eps}(t) \eqdef \ppr{x \la \Beroo{n,\eps}}{x = t}$.

\begin{proposition}[Restatement of \cref{prop:binomProbEstimation}]\label{app:missinProofs:binomProbEstimation}
\propBinomProbEstimation
\end{proposition}
\begin{proof}
In the following we focus on $n \geq 200$, smaller $n$'s are handled by setting the value of $\xconst$ to be large enough on these values. Let $\xconst_1$ be the universal constant from \cref{app:missinProofs:binomCoeffEstimation}. Compute
\begin{align}
\Beroo{n, \eps}(t)
&= \binom{n}{\frac{n+t}{2}}2^{-n}(1+\eps)^{\frac{n+t}{2}}(1-\eps)^{\frac{n-t}{2}}\\
&\in \sqrt{\frac{2}{\pi}}(1 \pm \xconst_1\cdot (\frac{\abs{t}^3}{n^2} + \frac{1}{n})) \cdot \frac{1}{\sqrt{n}} \cdot e^{-\frac{t^2}{2n}} \cdot (1-\eps^2)^{\frac{n-t}{2}}(1+\eps)^{t},\nonumber
\end{align}
where the second transition holds by \cref{app:missinProofs:binomCoeffEstimation}.
Since $1 + x \in e^{x\pm x^2}$ for $x \in (-0.5,0.5)$, it follows that:
\begin{align*}
\Beroo{n, \eps}(t)
&\in \sqrt{\frac{2}{\pi}}(1 \pm \xconst_1\cdot (\frac{\abs{t}^3}{n^2} + \frac{1}{n})) \cdot \frac{1}{\sqrt{n}} \cdot e^{-\frac{t^2}{2n}} \cdot e^{(-\eps^2\pm \eps^4)\cdot \frac{n-t}{2}}e^{(\eps \pm \eps^2) t}\\
&\subseteq \sqrt{\frac{2}{\pi}}(1 \pm \xconst_1\cdot (\frac{\abs{t}^3}{n^2} + \frac{1}{n})) \cdot \frac{1}{\sqrt{n}} \cdot e^{-\frac{t^2}{2n} - \frac{\eps^2 n}{2} + \eps t} \cdot e^{\pm(2\eps^2 \abs{t} + \frac{\eps^4 n}{2})}\nonumber\\
&\subseteq  \sqrt{\frac{2}{\pi}}(1 \pm \xconst_1\cdot (\frac{\abs{t}^3}{n^2} + \frac{1}{n})) \cdot \frac{1}{\sqrt{n}} \cdot e^{-\frac{(t-\eps n)^2}{2n}}(1 \pm 4\cdot (\eps^2 \abs{t} + \eps^4 n))\nonumber\\
&\subseteq \sqrt{\frac{2}{\pi}}(1 \pm \xconst\cdot (\eps^2 \abs{t} + \frac{\abs{t}^3}{n^2} + \frac{1}{n} + \eps^4 n)) \cdot \frac{1}{\sqrt{n}} \cdot e^{-\frac{(t-\eps n)^2}{2n}},\nonumber
\end{align*}
where $\xconst = 4\xconst_1 + 4$. Note that since $e^x \in 1 \pm 2x$ for every $\abs{x} < 1$, and since $2\eps^2\abs{t} + \frac{\eps^4 n}{2} < 2n^{-\frac15} + \frac12 n^{-\frac35} < 1$, it follows that $e^{\pm(2\eps^2 \abs{t} + \frac{\eps^4 n}{2})} \subseteq 1 \pm 4\cdot (\eps^2 \abs{t} + \eps^4 n)$ which yields the third transition. In addition, note that $\frac{\abs{t}^3}{n^2} + \frac{1}{n} < n^{-\frac15} + \frac{1}{n} < 1$, which implies the last transition.
\end{proof}

Using the above estimation for the binomial probability, the following proposition estimate the relation between two binomial probabilities.
\begin{proposition}[Restatement of \cref{prop:binomProbRelation}]\label{app:missinProofs:binomProbRelation}
\propBinomProbRelation
\end{proposition}
\begin{proof}
Let $\xconst$ be the constant from \cref{app:missinProofs:binomProbEstimation}. There exists a function $\vartheta \colon \R^+ \mapsto \N$ such that $n^{\frac35} > 2\const \cdot \sqrt{n\log{n}}$ and $\xconst\cdot(\const^4 + 10\const^3 + 1)\cdot \frac{\log^{1.5}n}{\sqrt{n}} < \frac12$  for every $n \geq \vartheta(\const)$. In the following we focus on $n \geq \vartheta(\const)$, where smaller $n$'s are handled by setting the value of $\varphi(\const)$ to be large enough on these values. Let $\varphi(\const) \eqdef 4\cdot \xconst \cdot (\const^4 + 10\const^3 + 1)$. It follows that

\begin{align*}
\frac{\Beroo{n,\eps}(t-x')}{\Beroo{n,\eps}(t-x)}
&\in \frac{\left(1 \pm   \xconst \cdot  (\eps^2 \abs{t-x'} + \frac{1}{n} + \frac{\abs{t-x'}^3}{n^2} + \eps^4 n)\right)\cdot \sqrt{\frac{2}{\pi}} \cdot \frac{1}{\sqrt{n}} \cdot e^{-\frac{(t - \eps n - x')^2}{2n}}}{\left(1 \pm   \xconst \cdot  (\eps^2 \abs{t-x} + \frac{1}{n} + \frac{\abs{t-x}^3}{n^2} + \eps^4 n)\right)\cdot \sqrt{\frac{2}{\pi}}\cdot \frac{1}{\sqrt{n}}\cdot e^{-\frac{(t - \eps n - x)^2}{2n}}}\\
&\subseteq \frac{\left(1 \pm   \xconst \cdot  (\const^4 + 10\const^3 + 1)\cdot \frac{\log^{1.5}n}{\sqrt{n}}\right)\cdot e^{-\frac{(t - \eps n - x')^2}{2n}}}{\left(1 \pm   \xconst \cdot  (\const^4 + 10\const^3 + 1)\cdot \frac{\log^{1.5}n}{\sqrt{n}}\right)\cdot e^{-\frac{(t - \eps n - x)^2}{2n}}}\nonumber\\
&\subseteq (1 \pm \varphi(\const) \cdot  \frac{\log^{1.5}n}{\sqrt{n}}) \cdot \exp\left(\frac{(t - \eps n - x)^2}{2n} - \frac{(t - \eps n - x')^2}{2n}\right)\nonumber\\
&= (1 \pm \varphi(\const) \cdot  \frac{\log^{1.5}n}{\sqrt{n}}) \cdot \exp\left(\frac{- 2\cdot (t - \eps n)\cdot x + x^2 + 2\cdot (t - \eps n)\cdot x' - x'^2}{2n}\right),\nonumber
\end{align*}
where the first transition holds by \cref{app:missinProofs:binomProbEstimation}, the second one holds by the bounds on $\size{t}$, $\size{x}$, $\size{x'}$ and $\size{\eps}$, and the third one holds since $\frac{1\pm y}{1\pm y} \subseteq 1 \pm 4y$ for every $y \in [0,\frac12]$.
\end{proof}

Recall that for $n \in \N$ and $k \in \Z$ we let $\vBeroo{n, \eps}(k) \eqdef \ppr{x \la \Beroo{n,\eps}}{x \geq  k} = \sum_{t\geq k}\Beroo{n,\eps}(t)$. 
Assuming that $n$ is larger than some universal constant, the following proposition gives a useful bound on the probability of the event that a binomial distribution is in a certain range of value.

\begin{proposition}[Restatement of \cref{prop:gameValuesDifferenceBound}]\label{app:missinProofs:gameValuesDifferenceBound}
\propGameValuesDifferenceBound
\end{proposition}
\begin{proof}
By \cref{app:missinProofs:binomProbEstimation}, for every $t \in \Z$ with $\size{t} \leq n^{\frac35}$, it holds that
\begin{align*}
\Beroo{n, \eps}(t) &\in (1 \pm 0.1) \cdot \sqrt{\frac{2}{\pi}} \cdot \frac{1}{\sqrt{n}} \cdot e^{-\frac{(t-\eps n)^2}{2n}},
\end{align*}
and therefore
\begin{align*}
\Beroo{n, \eps}(t)
&\leq \frac{1}{\sqrt{n}} \cdot e^{-\frac{(t-\eps n)^2}{2n}}
\leq \frac{1}{\sqrt{n}}.
\end{align*}
Assume \wlg that $k' \ge k$, it holds that $\vBeroo{n, \eps}(k) - \vBeroo{n, \eps}(k') = \sum _{t = k}^{k'} \Beroo{n, \eps}(t)$, which by the bound above, is at most  $\frac{(k'-k)}{\sqrt{n}}$.
\end{proof}

Recall that the function $\Phi\colon \R \mapsto (0,1)$ defined as $\Phi(x) \eqdef \frac{1}{\sqrt{2\pi}}\int_{x}^{\infty}e^{-\frac{t^2}{2}}dt$ is the cumulative distribution function of the standard normal distribution. The following fact and proposition are the first steps towards estimating the value of $\vBeroo{n,\eps}(k)$ in \cref{app:missinProofs:gameValueEstimation}.


\begin{fact}[\cite{AbramowitzS64}]\label{app:missinProofs:normalBound}
For $x \geq 0$ it holds that
\begin{align*}
\sqrt{\frac{2}{\pi}} \cdot \frac{e^{-\frac{x^2}{2}}}{x + \sqrt{x^2 + 4}} \leq
\Phi(x) \leq \sqrt{\frac{2}{\pi}}\cdot \frac{e^{-\frac{x^2}{2}}}{x + \sqrt{x^2 + \frac{8}{\pi}}}.
\end{align*}
\end{fact}

\begin{proposition}\label{app:missinProofs:estimateSumWithIntegral}
Let $n \in \N$, $\eps \in (-1,1)$ and $k,k' \in \Z$ be such that $k' \geq k \geq \frac{\eps n}{2}$. Then
\begin{align*}
\abs{\sum_{t=k}^{k'}e^{-\frac{(2t-\eps n)^2}{2n}} - \int_{k}^{k'} e^{-\frac{(2t-\eps n)^2}{2n}}dt}
\leq e^{-\frac{(2k - \eps n)^2}{2n}}.
\end{align*}
\end{proposition}
\begin{proof}
Consider the function $f(t) = e^{-\frac{(2t-\eps n)^2}{2n}}$. The function $f$ obtains its maximum at $t = \frac{\eps n}{2}$ and is monotonic decreasing in $[\frac{\eps n}{2},\infty)$. In particular, it is decreasing
in $[k,\infty)$. Since $\sum_{t=k}^{k'}f(t)$ is an upper Darboux sum of $f$ \wrt $\{k,k+1, \ldots, k'+1\}$, it holds that
$\int_{k}^{k'}f(t)dt \leq \int_{k}^{k'+1}f(t)dt \leq \sum_{t=k}^{k'}f(t)$. In addition, since $\sum_{t=k+1}^{k'}f(t)$ is a lower Darboux sum of $f$ \wrt $\{k,k+1, \ldots, k'\}$, it holds that  $\sum_{t=k}^{k'}f(t) \leq \int_{k}^{k'}f(t)dt + f(k)$. The proof follows, since the difference between the above sums is at most  $f(k) = e^{-\frac{(2k - \eps n)^2}{2n}}$.
\end{proof}

We are now ready for estimating $\vBeroo{n,\eps}(k)$ using the function $\Phi$.

\begin{proposition}\label{app:missinProofs:gameValueEstimation}
Let $n \in \N$, $k\in \Z$, $\eps \in [-1,1]$ and $\const > 0$ be such that $\size{\eps} \leq \const\cdot \sqrt{\frac{\log{n}}{n}}$ and $\abs{k} < \const\cdot \sqrt{n\log{n}}$. Then
\begin{align*}
\vBeroo{n, \eps}(k) \in \Phi(\frac{k - \eps n}{\sqrt{n}}) \pm \error,
\end{align*}
for $\error = \varphi(\const) \cdot \frac{\log^{1.5}n}{\sqrt{n}}\cdot e^{-\frac{(k-\eps n)^2}{2n}}$ and a universal function $\varphi$.
\end{proposition}
\begin{proof}
Without loss of generality, assume that $\const \geq 4$. Note that there exists a function $\vartheta \colon \R^+ \mapsto \N$ such that $n^\frac35 > 5\const\cdot \sqrt{n\log{n}}$  for every $n \geq \vartheta(\const)$. In the following we focus on $n \geq \vartheta(\const)$, where smaller $n$'s are handled by setting the value of $\varphi(\const)$ to be large enough on these values. We also assume for simplicity that $n$ and $k$ are both even, where the proofs of the other cases are analogous. Let $\xconst_1$ be the constant defined in \cref{app:missinProofs:binomProbEstimation}, and let $\ell \eqdef 4\cdot \ceil{\const\sqrt{n\log n}}< 5\const\cdot \sqrt{n\log n}$. We start by handling the case $k \geq \eps n$. It holds that
\begin{align}\label{app:missinProofs:gameValueEstimation:1}
\sum_{t=k}^{\ell} \Beroo{n, \eps}(t)
&= \sum_{t=\frac{k}{2}}^{\frac{\ell}{2}} \Beroo{n, \eps}(2t)\\
&\in \sum_{t=\frac{k}{2}}^{\frac{\ell}{2}} \sqrt{\frac{2}{\pi}}(1 \pm \xconst_1 \cdot (\eps^2 \abs{2t} + \frac{\abs{2t}^3}{n^2} + \frac{1}{n} + \eps^4 n)) \cdot \frac{1}{\sqrt{n}} e^{-\frac{(2t-\eps n)^2}{2n}}\nonumber\\
&\subseteq \sum_{t=\frac{k}{2}}^{\frac{\ell}{2}} \sqrt{\frac{2}{\pi}}(1 \pm \varphi'(\const) \cdot \frac{\log^{1.5}n}{\sqrt{n}}) \cdot \frac{1}{\sqrt{n}} e^{-\frac{(2t-\eps n)^2}{2n}}\nonumber\\
&\subseteq (1 \pm \varphi'(\const) \cdot \frac{\log^{1.5}n}{\sqrt{n}}) \cdot A(n,k,\eps,\const), \nonumber
\end{align}
letting $\varphi'(\const) \eqdef \xconst_1\cdot (\const^4 + 1034\const^3 + 1)$ and $A(n,k,\eps,\const) \eqdef \sum_{t=\frac{k}{2}}^{\frac{\ell}{2}}  \sqrt{\frac{2}{\pi}} \cdot\frac{1}{\sqrt{n}} \cdot e^{-\frac{(2t-\eps n)^2}{2n}}$. The first transition holds since even $n$ yields that $\Beroo{n, \eps}(j) = 0$ for every odd $j$, the second one holds by \cref{app:missinProofs:binomProbEstimation} and the third one holds by the bounds on $\ell$, $\eps$ and $k$.

Compute
\begin{align}\label{app:missinProofs:gameValueEstimation:2}
A(n,k,\eps,\const)
&= \sum_{t=\frac{k}{2}}^{\frac{\ell}{2}}  \sqrt{\frac{2}{\pi}} \cdot\frac{1}{\sqrt{n}} \cdot e^{-\frac{(2t-\eps n)^2}{2n}}\\
&\in \int_{\frac{k}{2}}^{\frac{\ell}{2}} \sqrt{\frac{2}{\pi}} \cdot\frac{1}{\sqrt{n}} \cdot e^{-\frac{(2t-\eps n)^2}{2n}}dt \pm \frac{1}{\sqrt{n}}\cdot e^{-\frac{(k-\eps n)^2}{2n}}\nonumber\\
&= \int_{\frac{k-\eps n}{\sqrt{n}}}^{\frac{\ell - \eps n}{\sqrt{n}}} \frac{1}{\sqrt{2\pi}} \cdot e^{-\frac{x^2}{2}}dx \pm \frac{1}{\sqrt{n}}\cdot e^{-\frac{(k-\eps n)^2}{2n}}\nonumber\\
&= \Phi(\frac{k-\eps n}{\sqrt{n}}) - \Phi(\frac{\ell - \eps n}{\sqrt{n}}) \pm \frac{1}{\sqrt{n}}\cdot e^{-\frac{(k-\eps n)^2}{2n}}\nonumber\\
&\subseteq \Phi(\frac{k-\eps n}{\sqrt{n}}) \pm \frac{1}{n^{4c^2}} \pm \frac{1}{\sqrt{n}}\cdot e^{-\frac{(k-\eps n)^2}{2n}}\nonumber\\
&\subseteq \Phi(\frac{k-\eps n}{\sqrt{n}}) \pm \frac{2}{\sqrt{n}}\cdot e^{-\frac{(k-\eps n)^2}{2n}},\nonumber
\end{align}
where the second transition holds by \cref{app:missinProofs:estimateSumWithIntegral} (and since $k \geq \eps n$), the third one holds by letting $x = \frac{2t-\eps n}{\sqrt{n}}$, the fifth one holds by \cref{app:missinProofs:normalBound} which yields that $\Phi(\frac{\ell -\eps n}{\sqrt{n}}) \leq \Phi(3\const\sqrt{\log n}) \leq \frac{1}{n^{4c^2}}$, and the last one holds since $\frac{1}{\sqrt{n}}\cdot e^{-\frac{(k-\eps n)^2}{2n}} \geq \frac{1}{n^{2\const^2 + \frac{1}{2}}} \geq  \frac{1}{n^{4\const^2}}$.
Applying \cref{app:missinProofs:gameValueEstimation:2} on \cref{app:missinProofs:gameValueEstimation:1} yields that
\begin{align}
\sum_{t=k}^{\ell} \Beroo{n, \eps}(t)\label{app:missinProofs:gameValueEst:1}
&\in (1 \pm \varphi'(\const) \cdot \frac{\log^{1.5}n}{\sqrt{n}}) \cdot (\Phi(\frac{k-\eps n}{\sqrt{n}}) \pm \frac{2}{\sqrt{n}}\cdot e^{-\frac{(k-\eps n)^2}{2n}}) \\
&= \Phi(\frac{k-\eps n}{\sqrt{n}}) \pm \varphi'(\const) \cdot \frac{\log^{1.5}n}{\sqrt{n}} \cdot \Phi(\frac{k-\eps n}{\sqrt{n}}) \pm
2\cdot \varphi'(\const) \cdot \frac{\log^{1.5}n}{n}\cdot e^{-\frac{(k-\eps n)^2}{2n}} \pm \frac{2}{\sqrt{n}}\cdot e^{-\frac{(k-\eps n)^2}{2n}}\nonumber\\
&\subseteq \Phi(\frac{k-\eps n}{\sqrt{n}}) \pm \varphi''(\const)\cdot \frac{\log^{1.5}n}{\sqrt{n}}\cdot e^{-\frac{(k-\eps n)^2}{2n}},\nonumber
\end{align}
letting $\varphi''(\const) \eqdef 3\cdot \varphi'(\const) + 2$. We conclude that
\begin{align}
\vBeroo{n, \eps}(k) &= \sum_{t=k}^{n} \Beroo{n, \eps}(t)\label{app:missinProofs:gameValueEst:2}\\
&= \sum_{t=k}^{\ell} \Beroo{n, \eps}(t) + \ppr{x\la \Beroo{n,\eps}}{x > \ell}\nonumber\\
&\in \sum_{t=k}^{\ell} \Beroo{n, \eps}(t)\pm \frac{1}{n^{4c^2}}\nonumber\\
&\subseteq \left(\Phi(\frac{k-\eps n}{\sqrt{n}}) \pm \varphi''(\const)\cdot \frac{\log^{1.5}n}{\sqrt{n}}\cdot e^{-\frac{(k-\eps n)^2}{2n}}\right)
\pm \frac{1}{n^{4\const^2}}\nonumber\\
&\subseteq \Phi(\frac{k-\eps n}{\sqrt{n}}) \pm (\varphi''(\const)+1)\cdot \frac{\log^{1.5}n}{\sqrt{n}}\cdot e^{-\frac{(k-\eps n)^2}{2n}},\nonumber
\end{align}
where the third transition holds by Hoeffding's inequality (\cref{claim:Hoeffding}) and the fourth one holds by \cref{app:missinProofs:gameValueEst:1}.

It is left to handle the case $k < \eps n$. For such $k$, it holds that
\begin{align}
\vBeroo{n, \eps}(k) &= 1 - \vBeroo{n, -\eps}(-k) + \Beroo{n, \eps}(k)\\
&\in 1 - \vBeroo{n, -\eps}(-k) \pm \frac{1}{\sqrt{n}} \cdot e^{-\frac{(k-\eps n)^2}{2n}}\nonumber\\
&\subseteq \left(1 - \Phi(\frac{-k + \eps n}{\sqrt{n}}) \pm (\varphi''(\const)+1)\cdot \frac{\log^{1.5}n}{\sqrt{n}}\cdot e^{-\frac{(k-\eps n)^2}{2n}}\right)
\pm \frac{1}{\sqrt{n}} \cdot e^{-\frac{(k-\eps n)^2}{2n}}\nonumber\\
&\subseteq \Phi(\frac{k - \eps n}{\sqrt{n}}) \pm (\varphi''(\const)+2)\cdot \frac{\log^{1.5}n}{\sqrt{n}}\cdot e^{-\frac{(k-\eps n)^2}{2n}},\nonumber
\end{align}
where the second transition holds by evaluating the value of $\Beroo{n, \eps}(k)$ using \cref{app:missinProofs:binomProbEstimation}
and the third one holds by \cref{app:missinProofs:gameValueEst:2} applied to $-k$ and $-\eps$.
\end{proof}

Recall that for $n \in \N$ and $\delta \in [0,1]$ we let $\sBias{n}{\delta}$ be the value $\eps \in [-1,1]$ with $\vBeroo{n, \eps}(0) = \delta$. The following proposition gives an estimation for $\sBias{n}{\delta}$ using \cref{app:missinProofs:gameValueEstimation}.

\begin{proposition}\label{app:missinProofs:epsDiffImproved}
Let $n \in \N$, $\delta \in [0,1]$ and $\const > 0$ be such that $\delta \in (\frac{1}{n^c}, 1-\frac{1}{n^c})$. Then,
$$\sBias{n}{\delta} \in -\frac{\Phi^{-1}(\delta)}{\sqrt{n}} \pm \error$$
for $\error = \varphi(\const) \cdot \frac{\log^{1.5} n}{n}$ and a universal function $\varphi$.
\end{proposition}
\begin{proof}
Let $\varphi'\colon \R^+ \mapsto \R^+$ be the function from \cref{app:missinProofs:gameValueEstimation}, and let $\varphi(\const) \eqdef 6\cdot\varphi'(\sqrt{2\const}+1) + 1$.  There exists a function $\vartheta \colon \R^+ \mapsto \N$ such that the two conditions
\begin{enumerate}
\item $\min(2\const,\const^2) \cdot \log n > 1$ \label{app:missinProofs:epsDiffImproved:bound1}

\item $\max(\sqrt{2\const}, \frac{1}{\sqrt{2\const}}) \cdot \max(\varphi^2(\const), 1) \cdot \frac{\log^2 n}{\sqrt{n}} < \frac18$ \label{app:missinProofs:epsDiffImproved:bound2}

\end{enumerate}
holds for every $n \geq \vartheta(\const)$. In the following we focus on $n \geq \vartheta(\const)$, where smaller $n$'s are handled by setting the value of $\varphi(c)$ to be large enough on these values.
Let $x \eqdef \Phi^{-1}(\delta)$ and $\Delta \eqdef \varphi(\const) \cdot \frac{\log^{1.5} n}{\sqrt{n}}$, and let $\eps^+ \eqdef -\frac{x}{\sqrt{n}} + \frac{\Delta}{\sqrt{n}}$ and $\eps^- \eqdef -\frac{x}{\sqrt{n}} -\frac{\Delta}{\sqrt{n}}$.
We prove that $\eps^- < \sBias{n}{\delta} < \eps^+$, yielding the required bound.
For simplicity, we focus on the upper bound, whereas the lower bound can be proven analogously.

Since $\delta \in (\frac{1}{n^\const}, 1-\frac{1}{n^\const})$, it follows by \cref{app:missinProofs:normalBound} and condition \ref{app:missinProofs:epsDiffImproved:bound1} that $\abs{x} \leq \sqrt{2\const \cdot \log n}$ and hence, using condition \ref{app:missinProofs:epsDiffImproved:bound2} it follows that $\abs{\eps^+} < (\sqrt{2\const} + 1) \cdot \sqrt{\frac{\log n}{n}}$.
Therefore, \cref{app:missinProofs:gameValueEstimation} yields that
\begin{align}
\vBeroo{n,\eps^+}(0)&\in \Phi(-\eps^+ \cdot \sqrt{n}) \pm \varphi'(\sqrt{2\const} + 1) \cdot \frac{\log^{1.5}{n}}{\sqrt{n}}\cdot e^{-\frac{{\eps^+}^2 \cdot n}{2}} \label{app:missinProofs:epsDiffImproved:basicBound}\\
&= \Phi(x - \Delta) \pm \varphi'(\sqrt{2\const} + 1) \cdot \frac{\log^{1.5}{n}}{\sqrt{n}}\cdot e^{-\frac{(x - \Delta)^2}{2}}\nonumber\\
&= \delta + \frac1{\sqrt{2\pi}}\cdot\int_{x-\Delta}^{x} e^{-\frac{t^2}{2}}dt \pm \varphi'(\sqrt{2\const} + 1) \cdot \frac{\log^{1.5}{n}}{\sqrt{n}}\cdot e^{-\frac{(x - \Delta)^2}{2}}.\nonumber
\end{align}
Note that
\begin{align}
e^{-\frac{(x-\Delta)^2}{2}} - e^{-\frac{x^2}{2}}
&= (1 - e^{\frac{-2\Delta x + \Delta^2}{2}}) \cdot e^{-\frac{(x-\Delta)^2}{2}}\\
&= (1 - e^{-2\Delta(\frac{x}{2} - \frac{\Delta}{4})}) \cdot e^{-\frac{(x-\Delta)^2}{2}}\nonumber\\
&\in (1 - e^{\pm 2\Delta\cdot \sqrt{2\const \cdot \log n}} ) \cdot e^{-\frac{(x-\Delta)^2}{2}}\nonumber\\
&\subseteq \pm 4\Delta\cdot \sqrt{2\const \cdot \log n} \cdot e^{-\frac{(x-\Delta)^2}{2}},\nonumber
\end{align}
where the third transition holds since $\abs{x}, \Delta < \sqrt{2\const\cdot \log n}$,
and the fourth one holds by the bound on $\Delta$ and using condition \ref{app:missinProofs:epsDiffImproved:bound2} since $e^y \in 1\pm 2\abs{y}$ for $y \in (-1,1)$. Therefore,
\begin{align}
\int_{x-\Delta}^{x} e^{-\frac{t^2}{2}}dt
&\in \Delta \cdot [\min(e^{-\frac{(x-\Delta)^2}{2}}, e^{-\frac{x^2}{2}}) , \max(e^{-\frac{(x-\Delta)^2}{2}}, e^{-\frac{x^2}{2}})] \label{app:missinProofs:epsDiffImproved:integralBound}\\
&\in \Delta \cdot e^{-\frac{(x-\Delta)^2}{2}} \cdot (1 \pm 4\Delta\cdot \sqrt{2\const\cdot \log n}).\nonumber
\end{align}
Applying \cref{app:missinProofs:epsDiffImproved:integralBound} on \cref{app:missinProofs:epsDiffImproved:basicBound} yields that
\begin{align}
\vBeroo{n,\eps^+}(0) - \delta
&\in \frac1{\sqrt{2\pi}}\cdot \Delta \cdot e^{-\frac{(x-\Delta)^2}{2}} \cdot \bigl(1 \pm 4\Delta\cdot \sqrt{2\const\cdot \log n}\bigr) \pm \varphi'(\sqrt{2\const} + 1) \cdot \frac{\log^{1.5}{n}}{\sqrt{n}}\cdot e^{-\frac{(x - \Delta)^2}{2}}\nonumber\\
&= \frac1{\sqrt{2\pi}}\cdot \left(\Delta \pm \bigl(4\Delta^2\cdot \sqrt{2\const\cdot \log n} + \sqrt{2\pi}\cdot \varphi'(\sqrt{2\const} + 1) \cdot \frac{\log^{1.5}{n}}{\sqrt{n}}\bigr)\right)\cdot e^{-\frac{(x-\Delta)^2}{2}}.\nonumber
\end{align}
By the definition of $\varphi$ and $\Delta$, and using condition \ref{app:missinProofs:epsDiffImproved:bound2}, it follows that
\begin{align}
\lefteqn{4\Delta^2\cdot \sqrt{2\const\cdot \log n} + \sqrt{2\pi}\cdot\varphi'(\sqrt{2\const} + 1) \cdot \frac{\log^{1.5}{n}}{\sqrt{n}}}\nonumber\\
&= 4\Delta \cdot \sqrt{2\const}\cdot \varphi(\const)\cdot \frac{\log^2 n}{\sqrt{n}} + \frac{\sqrt{2\pi}}{6}(\varphi(\const)-1)\cdot \frac{\log^{1.5}{n}}{\sqrt{n}} \nonumber\\
&< 4\Delta \cdot \frac18 + \frac12\Delta\nonumber\\
&= \Delta,\nonumber
\end{align}
and thus, $\vBeroo{n,\eps^+}(0) > \delta$, as required.
\end{proof}

In order to use \cref{app:missinProofs:epsDiffImproved} with $\delta = \vBeroo{n,\eps}(k)$, the following proposition first estimate the value of $\Phi^{-1}(\delta)$.

\begin{proposition}\label{app:missinProofs:PhiMinusOne}
Let $n \in \N$, $k\in \Z$, $\eps \in [-1,1]$ and $\const > 0$ be such that $\abs{k} \leq \const \cdot \sqrt{n \log n}$, $\size{\eps}  \leq \const \cdot \sqrt{\frac{\log n}{n}}$, and let $\delta = \vBeroo{n,\eps}(k)$. Then,
\begin{align*}
\Phi^{-1}(\delta) \in \frac{k - \eps n}{\sqrt{n}} \pm \error,
\end{align*}
for $\error = \varphi(\const) \cdot \frac{\log^{1.5} n}{\sqrt{n}}$ and a universal function $\varphi$.
\end{proposition}
\begin{proof}
Let $\varphi'\colon \R^+ \mapsto \R^+$ be the function from \cref{app:missinProofs:gameValueEstimation}, let $\Delta \eqdef 2\varphi'(\const)\cdot \log^{1.5}n$ and let $k_0 \eqdef k-\eps n$. Note that there exists a function $\vartheta \colon \R^+ \mapsto \N$ such that $e^{-4\varphi'(\const)(\varphi'(\const)+\const)\cdot\frac{\log^3 n}{\sqrt{n}}} \geq \frac12$  for every $n \geq \vartheta(\const)$. In the following we focus on $n \geq \vartheta(\const)$, where smaller $n$'s are handled by setting the value of $\varphi(\const)$ to be large enough on these values.

We prove that $\Phi(\frac{k_0 + \Delta}{\sqrt{n}}) \leq \delta \leq \Phi(\frac{k_0 - \Delta}{\sqrt{n}})$, which yields the required bound since $\Phi$ is monotonic decreasing. We focus on the upper bound, whereas the lower bound can be proven analogously.
Since
\begin{align}\label{app:missinProofs:PhiMinusOne:1}
\frac{\Delta}{\sqrt{n}} \cdot e^{-\frac{k_0^2}{2n}} \geq \varphi'(\const) \cdot \frac{\log^{1.5}n}{\sqrt{n}}\cdot e^{-\frac{k_0^2}{2n}}
\end{align}
and
\begin{align}\label{app:missinProofs:PhiMinusOne:2}
\frac{\Delta}{\sqrt{n}} \cdot e^{-\frac{(k_0-\Delta)^2}{2n}}
&= \frac{\Delta}{\sqrt{n}} \cdot e^{-\frac{k_0^2}{2n}} \cdot e^{\frac{2k_0\Delta - \Delta^2}{2n}}\\
&\geq \frac{\Delta}{\sqrt{n}} \cdot e^{-\frac{k_0^2}{2n}} \cdot e^{-4\varphi'(\const)(\varphi'(\const)+\const)\cdot\frac{\log^3 n}{\sqrt{n}}}\nonumber\\
&\geq \frac{\Delta}{\sqrt{n}} \cdot e^{-\frac{k_0^2}{2n}} \cdot \frac12\nonumber\\
&= \varphi'(\const) \cdot \frac{\log^{1.5}n}{\sqrt{n}}\cdot e^{-\frac{k_0^2}{2n}},\nonumber
\end{align}
it follows that
\begin{align}
\delta
&\leq \Phi(\frac{k_0}{\sqrt{n}}) + \varphi'(\const) \cdot \frac{\log^{1.5}n}{\sqrt{n}}\cdot e^{-\frac{k_0^2}{2n}}\\
&\leq \Phi(\frac{k_0}{\sqrt{n}}) + \frac{\Delta}{\sqrt{n}}\cdot \min(e^{-\frac{k_0^2}{2n}}, e^{-\frac{(k_0-\Delta)^2}{2n}})\nonumber\\
&\leq \Phi(\frac{k_0}{\sqrt{n}}) + \int_{\frac{k_0 - \Delta}{\sqrt{n}}}^{\frac{k_0}{\sqrt{n}}}e^{-\frac{t^2}{2}}dt\nonumber\\
&= \Phi(\frac{k_0}{\sqrt{n}} - \frac{\Delta}{\sqrt{n}}),\nonumber
\end{align}
where the first inequality holds by \cref{app:missinProofs:gameValueEstimation} and the second one by \cref{app:missinProofs:PhiMinusOne:1} and \cref{app:missinProofs:PhiMinusOne:2}.
\end{proof}

We are now ready for estimating the value of $\sBias{n'}{\delta}$ for $\delta = \vBeroo{n,\eps}(k)$ and for some $n' \geq n$.

\begin{proposition}[Restatement of \cref{prop:epsDiff}]\label{app:missinProofs:epsDiff}
\propEpsDiff
\end{proposition}
\begin{proof}
Let $\varphi_1$ be the function from \cref{app:missinProofs:gameValueEstimation}, $\varphi_2$ be the function from \cref{app:missinProofs:epsDiffImproved}, $\varphi_3$ be the function from \cref{app:missinProofs:PhiMinusOne} and let $\varphi(\const) \eqdef \varphi_2(2\const^2+1)+\varphi_3(\const)$. There exists a function $\vartheta \colon \R^+ \mapsto \N$ such that the two conditions
\begin{enumerate}
\item $\min(\const,1) \cdot \log n > 4$ \label{app:missinProofs:epsDiff:bound1}

\item $\max(\const, \varphi_1(\const)) \cdot \frac{\log^2 n}{\sqrt{n}} < \frac18$ \label{app:missinProofs:epsDiff:bound2}

\end{enumerate}
holds for every $n \geq \vartheta(\const)$. In the following we focus on $n \geq \vartheta(\const)$, where smaller $n$'s are handled by setting the value of $\varphi(c)$ to be large enough on these values.
In order to use \cref{app:missinProofs:epsDiffImproved}, we first prove that $\delta \in (\frac{1}{n^{2\const^2+1}}, 1-\frac{1}{n^{2\const^2+1}})$.
Let $k_0 \eqdef k-\eps n$. For simplicity, we assume $k_0 \geq 0$, whereas the case $k_0 < 0$ holds by symmetry.
Compute
\begin{align}
\delta
&\in \Phi(\frac{k_0}{\sqrt{n}}) \pm \varphi_1(\const)\cdot \frac{\log^{1.5}n}{\sqrt{n}}\cdot e^{-\frac{k_0^2}{2n}}\\
&\subseteq \left(\frac1{\frac{k_0}{\sqrt{n}} + \sqrt{\frac{k_0^2}{n} + 4 \pm 2}} \pm \varphi_1(\const)\cdot \frac{\log^{1.5}n}{\sqrt{n}}\right)\cdot e^{-\frac{k_0^2}{2n}},\nonumber\\
&\subseteq \frac{1\pm\frac12}{\frac{k_0}{\sqrt{n}} + \sqrt{\frac{k_0^2}{n} + 4 \pm 2}} \cdot e^{-\frac{k_0^2}{2n}}\nonumber\\
&\subseteq (\frac{1}{8\const \cdot \sqrt{\log n}\cdot n^{2\const^2}}, \frac34)\nonumber\\
&\subseteq (\frac{1}{n^{2\const^2+1}}, 1-\frac{1}{n^{2\const^2+1}})\nonumber
\end{align}
where the first transition holds by \cref{app:missinProofs:gameValueEstimation}, the second one holds by \cref{app:missinProofs:normalBound}, the third one holds by condition \ref{app:missinProofs:epsDiff:bound2} and since $k_0 \leq 2\const \cdot \sqrt{n\log n}$, the fourth one also holds since $k_0 \leq 2\const \cdot \sqrt{n\log n}$ and the last one holds by conditions \ref{app:missinProofs:epsDiff:bound1} and \ref{app:missinProofs:epsDiff:bound2}.

Finally, it holds that
\begin{align}
\sBias{n'}{\delta}
&\in -\frac{\Phi^{-1}(\delta)}{\sqrt{n'}} \pm \varphi_2(2\const^2+1) \cdot \frac{\log^{1.5} (n')}{n'}\\
&\subseteq -\frac{\left(\frac{k - \eps n}{\sqrt{n}} \pm \varphi_3(\const) \cdot \frac{\log^{1.5} n}{\sqrt{n}}\right)}{\sqrt{n'}} \pm \varphi_2(2\const^2+1) \cdot \frac{\log^{1.5} (n')}{n'}\nonumber\\
&\subseteq \frac{\eps n - k}{\sqrt{n\cdot n'}} \pm \varphi(\const) \cdot \frac{\log^{1.5} n}{\sqrt{n\cdot n'}},\nonumber
\end{align}
where the first transition holds by \cref{app:missinProofs:epsDiffImproved}, the second one by \cref{app:missinProofs:PhiMinusOne}
and the last one holds since $n \leq n'$.
\end{proof}

For the following two propositions, recall that for $n\in \N$ and $i \in [n]$ we let $\NL{i} = n-i+1$ and $\NS{i} = \sum_{j=i}^{n}\NL{i} = \frac12\cdot \NL{i}(\NL{i}+1)$. The following proposition is the main step towards proving \cref{prop:mainBound}.

\begin{proposition}\label{app:missinProofs:estimatingPolyX}
Let $n\in \N$, integer $i \in [n-\floor{\log^{2.5}n}]$, $x,\beta,\alpha \in \Z$, $\eps \in [-1,1]$ and $\const > 0$ be such that $\size{\alpha} \leq \sqrt{\const \cdot  \NS{i} \cdot \log n}$,  $\size{x}  \leq \sqrt{\const \cdot \NL{i} \cdot \log{n}}$, $\size{\beta} \le 1$ and $\size{\eps} \leq \sqrt{\const \cdot \frac{\log n}{\NS{i}}}$. Then
\begin{align*}
\exp\left(\frac{\alpha \cdot x + \beta \cdot x^2}{\NS{i+1}}\right) \in 1 \pm \varphi(\const)\cdot
\frac{\sqrt{\log n} \cdot \abs{x}}{\NL{i}},
\end{align*}
for a universal function $\varphi$.
\end{proposition}

\begin{proof}
Assume that $n\geq 4$. By taking the maximum possible values of $\abs{\alpha}$, $\abs{\beta}$ and $\abs{x}$ it follows that
\begin{align}
\abs{\frac{\alpha \cdot x + \beta \cdot x^2}{\NS{i+1}}}&\leq \frac{\sqrt{\const \cdot  \NS{i} \cdot \log n} \cdot \sqrt{\const \cdot \NL{i} \cdot \log n} + \const \cdot \NL{i} \cdot \log n}{\NS{i+1}}\\
&= \const \cdot \frac{\sqrt{2(n-i+2)}}{n-i}\cdot \log n + 2\const\cdot \frac{\log n}{n-i}\nonumber\\
&\leq 2\const \cdot \frac{\log n}{\sqrt{n-i}} + 2\const \cdot \frac{\log n}{n-i}\nonumber\\
&\leq 2\const \cdot \frac{1}{\log^{0.25} n}+ 2\const \cdot \frac{1}{\log^{1.5} n},\nonumber
\end{align}
where the second inequality holds since $\frac{n-i+2}{n-i} < 2$.
Therefore, there exists a function $\vartheta \colon \R^+ \mapsto \N$ such that $\abs{\frac{\alpha \cdot x + \beta \cdot x^2}{\NS{i+1}}} < 1$  for every $n \geq \vartheta(\const)$. In the following we focus on $n \geq \vartheta(\const)$, where smaller $n$'s are handled by setting the value of $\varphi(\const)$ to be large enough on these values.
Since $e^y \in 1 \pm 2\abs{y}$ for $y \in [-1,1]$, it follows that
\begin{align*}
\exp\left(\frac{\alpha \cdot x + \beta \cdot x^2}{\NS{i+1}}\right)
&\in 1 \pm 2\cdot \abs{\frac{\alpha \cdot x + \beta \cdot x^2}{\NS{i+1}}}\\
&\subseteq 1 \pm 4\cdot \abs{\frac{\alpha \cdot x + \beta \cdot x^2}{\NS{i}}}\\
&\subseteq 1 \pm 4\cdot \left(\frac{\abs{\alpha} \cdot \abs{x} + \abs{\beta} \cdot x^2}{\NS{i}}\right)\\
&\subseteq 1 \pm 4\cdot \left(\frac{\sqrt{\const \cdot \NS{i} \cdot \log{n}} \cdot \abs{x} + x^2}{\NS{i}}\right)\\
&\subseteq 1 \pm 4\cdot \left(\sqrt{\frac{\const \cdot \log{n}}{\NS{i}}}\cdot \abs{x} + \frac{\sqrt{\const \cdot \NL{i} \cdot \log n}}{\NS{i}}\cdot \abs{x}\right)\\
&= 1 \pm 4\cdot \left(\sqrt{\frac{\const \cdot \log{n}}{\frac12\NL{i}(\NL{i}+1)}}\cdot \abs{x} + \frac{\sqrt{\const \cdot \NL{i} \cdot \log n}}{\frac12\NL{i}(\NL{i}+1)}\cdot \abs{x}\right)\\
&\subseteq 1 \pm 8\cdot \sqrt{\frac{\const \cdot \log{n}}{\NL{i}}}\left(\frac{\abs{x}}{\sqrt{\NL{i}}} + \frac{\abs{x}}{\NL{i}}\right)\\
&\subseteq 1 \pm 16\sqrt{\const}\cdot \frac{\sqrt{\log n} \cdot \abs{x}}{\NL{i}},
\end{align*}
where the second transitions holds since $\frac{\NS{i}}{\NS{i+1}} < 2$, the fourth one holds by taking the maximum possible values of $\abs{\alpha}$ and $\abs{\beta}$ and fifth one by taking the  maximum possible value of $\abs{x}$.
\end{proof}

Using the above fact, we can prove \cref{prop:mainBound}.
\begin{proposition}[Restatement of \cref{prop:mainBound}]\label{app:missinProofs:mainBound}
\propMainBound
\end{proposition}
\begin{proof}
Let $\varphi'$ be the function from \cref{app:missinProofs:estimatingPolyX}. Compute
\begin{align*}
\lefteqn{\ex{x'\la \Beroo{\NL{i},\eps} \mid x'\in \cs}
{\exp\left(\frac{\alpha \cdot x + \beta \cdot x^2 + \alpha' \cdot x' + \beta' \cdot x'^2}{\NS{i+1}}\right)}}\\
&= \exp\left(\frac{\alpha \cdot x + \beta \cdot x^2}{\NS{i+1}}\right)\cdot \ex{x'\la \Beroo{\NL{i},\eps} \mid x'\in \cs}{\left(\frac{\alpha' \cdot x' + \beta' \cdot x'^2}{\NS{i+1}}\right)}\\
&\in \left(1 \pm \varphi'(\const)\cdot \frac{\sqrt{\log{n}}\cdot \abs{x}}{\NL{i}}\right)\cdot
\ex{x'\la \Beroo{\NL{i},\eps} \mid x'\in \cs}{1 \pm \varphi'(\const)\cdot \frac{\sqrt{\log{n}}\cdot \abs{x'}}{\NL{i}}}\\
&= \left(1 \pm \varphi'(\const)\cdot \frac{\sqrt{\log{n}}\cdot \abs{x}}{\NL{i}}\right)\cdot
\left(1 \pm \varphi'(\const)\cdot \frac{\sqrt{\log{n}}\cdot \ex{x'\la \Beroo{\NL{i},\eps} \mid x'\in \cs}{\abs{x'}}}{\NL{i}}\right)\\
&\subseteq \left(1 \pm \varphi'(\const)\cdot \frac{\sqrt{\log{n}}\cdot \abs{x}}{\NL{i}}\right)\cdot
\left(1 \pm 2\varphi'(\const)\cdot \sqrt{\frac{\log{n}}{\NL{i}}}\right)\\
&\subseteq 1 \pm 4(\varphi'(\const)+ \varphi'(\const)^2)\cdot \sqrt{\frac{\log{n}}{\NL{i}}}\cdot \left(1 + \frac{\abs{x}}{\sqrt{\NL{i}}}\right),
\end{align*}
where the second transition holds by \cref{app:missinProofs:estimatingPolyX} and the fourth one holds by \cref{fact:E_x_m}.
\end{proof}

\subsection{Facts About the Hypergeometric Distribution}\label{app:missinProofs:HypGeo}

Recall that for a vector $\vct \in \oo^\ast$ we let $\w(\vct) \eqdef \sum_{i \in [\size{\cI}]}\vct_i$, and given a set of indexes $\cI \subseteq [\size{\vct}]$, we let $\vct_{\cI} = (\vct_{i_1},\ldots,\vct_{i_{\size{\cI}}})$ where $i_1,\ldots,i_{\size{\cI}}$ are the ordered elements of $\cI$. In addition, recall that for $n\in \N$, $\ell \in [n]$, and an integer $p\in [-n,n]$, we define the hypergeometric probability distribution $\Hyp{n,p,\ell}$ by $\Hyp{n,p,\ell}(k) \eqdef \ppr{\cI}{\w(\vct_\cI) = k}$, where $\cI$ is an $\ell$-size set uniformly chosen from $[n]$ and $\vct \in \oo^n$ with $w(\vct)= p$. The following proposition gives an estimation for the hypergeometric probability $\Hyp{2n, p, n}(t)$ using the binomial coefficient's estimation done in \cref{app:missinProofs:binomCoeffEstimation}.

\remove{shows that the value of $\Hyp{2n, p, n}(t)$ is very close to the value of $\Beroo{\frac{n}{2},\frac{p}{n}}(t)$.
This fact helps us to prove \cref{prop:hyperProbRelation} by using \cref{prop:binomProbRelation}.}

\begin{proposition}[Restatement of \cref{prop:hyperProbEstimation}]\label{app:missinProofs:hyperProbEstimation}
\propHyperProbEstimation
\end{proposition}
\begin{proof}
Let $\xconst_1$ be the constant from \cref{app:missinProofs:binomCoeffEstimation} and let $\omega \eqdef \frac{p}{2}$.
In the following we focus on $n \geq 1000(1 + \xconst_1^2)$, smaller $n$'s are handled by setting the value of $\xconst$ to be large enough on these values.
Note that for any vector $\vct \in\oo^{2n}$ with $\w(\vct) = p$, the number of ones in $\vct$ is $n+\omega$. It follows that
\begin{align}\label{app:missinProofs:hyperProbEstimation:1}
\Hyp{2n, p, n}(t)
&= \frac{\binom{n+\omega}{\frac{n+t}{2}} \cdot
\binom{n-\omega}{\frac{n-t}{2}}}{\binom{2n}{n}}\\
&= \frac{\binom{n+\omega}{\frac{n+\omega}{2} + \frac{t-\omega}{2}} \cdot
\binom{n-\omega}{\frac{n-\omega}{2} - \frac{t-\omega}{2}}}
{\binom{2n}{n}}\nonumber\\
&\in \frac{\sqrt{\frac{2}{\pi}}(1 \pm \xconst_1 \cdot (\frac{1}{n} + \frac{\abs{t-\omega}^3}{n^2}))\frac{1}{\sqrt{n+\omega}}e^{-\frac{(t-\omega)^2}{2(n+\omega)}}
\cdot \sqrt{\frac{2}{\pi}} \cdot (1 \pm \xconst_1 \cdot (\frac{1}{n} + \frac{\abs{t-\omega}^3}{n^2}))\frac{1}{\sqrt{n-\omega}}e^{-\frac{(t-\omega)^2}{2(n-\omega)}}}
{\sqrt{\frac{2}{\pi}} \cdot (1 \pm \xconst_1 \cdot \frac{1}{n})\frac{1}{\sqrt{2n}}}\nonumber\\
&\subseteq \sqrt{\frac{2}{\pi}} \cdot (1 \pm \xconst_2 \cdot (\frac{1}{n} + \frac{\abs{t-\omega}^3}{n^2}))\cdot
A(n,t,\omega),\nonumber
\end{align}
where the third transition holds by \cref{app:missinProofs:binomCoeffEstimation}, $\xconst_2 \eqdef 8\cdot (\xconst_1 + \xconst_1^2)$ and $A(n,t,\omega) \eqdef \sqrt{2n}\cdot \frac{1}{\sqrt{n+\omega}}e^{-\frac{(t-\omega)^2}{2(n+\omega)}}
\cdot \frac{1}{\sqrt{n-\omega}}e^{-\frac{(t-\omega)^2}{2(n-\omega)}}$. Compute
\begin{align}\label{app:missinProofs:hyperProbEstimation:2}
A(n,t,\omega)
&= \sqrt{\frac{2}{n}} \cdot
\frac{n}{\sqrt{n+\omega}\cdot \sqrt{n-\omega}} \cdot e^{-\frac{(t-\omega)^2}{2(n+\omega)}} \cdot e^{-\frac{(t-\omega)^2}{2(n-\omega)}}\\
&= \sqrt{\frac{2}{n}} \cdot \frac{1}{\sqrt{1-\frac{\omega^2}{n^2}}} \cdot e^{-\frac{(t-\omega)^2}{n}} \cdot e^{-(t-\omega)^2(\frac{1}{2(n+\omega)} + \frac{1}{2(n-\omega)} - \frac{1}{n})}\nonumber\\
&\in \sqrt{\frac{2}{n}} \cdot \left(1 \pm 2\cdot \frac{\omega^2}{n^2}\right) \cdot e^{-\frac{(t-\omega)^2}{n}} \cdot e^{-\frac{(t-\omega)^2\omega^2}{n(n^2-\omega^2)}}\nonumber\\
&\subseteq \sqrt{\frac{2}{n}}\cdot e^{-\frac{(t-\omega)^2}{n}} \cdot \left(1 \pm 2\cdot \frac{\omega^2}{n^2}\right) \cdot \left(1 \pm 2\cdot\frac{(t-\omega)^2\omega^2}{n(n^2-\omega^2)}\right)\nonumber\\
&\subseteq \sqrt{\frac{2}{n}} \cdot e^{-\frac{(t-\omega)^2}{n}} \cdot
\left(1 \pm  4\cdot (\frac{(t-\omega)^2\omega^2}{n^3}+\frac{\omega^2}{n^2})\right), \nonumber
\end{align}
where the third transition holds since $\frac{1}{\sqrt{1-x}} \in 1 \pm 2x$ for $x \in [0,\frac{1}{4}]$, and the fourth one holds since
$e^x \in 1\pm 2x$ for $\abs{x} < 1$.

We conclude from \cref{app:missinProofs:hyperProbEstimation:1,app:missinProofs:hyperProbEstimation:2} that
\begin{align*}
\Hyp{2n, p, n}(t)
&\in \sqrt{\frac{2}{\pi}}\cdot (1 \pm \xconst_3 \cdot(\frac{n + \abs{\omega}^3 + \abs{t}^3}{n^2})) \cdot \sqrt{\frac{2}{n}} \cdot e^{-\frac{(t-\omega)^2}{n}}\\
&\in (1 \pm \xconst \cdot(\frac{n + \abs{p}^3 + \abs{t}^3}{n^2})) \cdot \frac{2}{\sqrt{\pi\cdot n}} \cdot e^{-\frac{(t-\frac{p}2)^2}{n}},
\end{align*}
where $\xconst_3 \eqdef 16\cdot(1+\xconst_2)$ and $\xconst \eqdef 8\cdot \xconst_3$.
\end{proof}

Using the above estimation for the hypergeometric probability, the following proposition estimates the relation between two hypergeometric probabilities.
\begin{proposition}[Restatement of \cref{prop:hyperProbRelation}]\label{app:missinProofs:hyperProbRelation}
\propHyperProbRelation
\end{proposition}
\begin{proof}
Let $\xconst$ be the constant from  \cref{app:missinProofs:hyperProbEstimation}. There exists a function $\vartheta \colon \R^+ \mapsto \N$ such that $n^{\frac35} > 2\const\cdot \sqrt{n\log{n}}$ and $\xconst\cdot(10\const^3 + 1)\cdot \frac{\log^{1.5}n}{\sqrt{n}} < \frac12$  for every $n \geq \vartheta(\const)$. In the following we focus on $n \geq \vartheta(\const)$, where smaller $n$'s are handled by setting the value of $\varphi(\const)$ to be large enough on these values. Let $\varphi(\const) \eqdef 4\cdot \xconst \cdot (10\const^3 + 1)$. It follows that

\begin{align*}
\frac{\Hyp{2n, p, n}(t-x')}{\Hyp{2n, p, n}(t-x)}
&\in \frac{\left(1 \pm   \xconst \cdot  \frac{n+\abs{p}^3+\abs{t-x'}^3}{n^2}\right)\cdot \frac{2}{\sqrt{\pi\cdot n}} \cdot e^{-\frac{(t - \frac{p}{2} - x')^2}{n}}}{\left(1 \pm   \xconst \cdot  \frac{n+\abs{p}^3+\abs{t-x'}^3}{n^2}\right)\cdot \frac{2}{\sqrt{\pi\cdot n}}\cdot e^{-\frac{(t - \frac{p}{2} - x)^2}{n}}}\\
&\subseteq \frac{\left(1 \pm   \xconst \cdot  (10\const^3 + 1)\cdot \frac{\log^{1.5}n}{\sqrt{n}}\right)\cdot e^{-\frac{(t - \frac{p}{2} - x')^2}{n}}}{\left(1 \pm   \xconst \cdot  (10\const^3 + 1)\cdot \frac{\log^{1.5}n}{\sqrt{n}}\right)\cdot e^{-\frac{(t - \frac{p}{2} - x)^2}{n}}}\nonumber\\
&\subseteq (1 \pm \varphi(\const) \cdot  \frac{\log^{1.5}n}{\sqrt{n}}) \cdot \exp\left(\frac{(t - \frac{p}{2} - x)^2}{n} - \frac{(t - \frac{p}{2} - x')^2}{n}\right)\nonumber\\
&= (1 \pm \varphi(\const) \cdot  \frac{\log^{1.5}n}{\sqrt{n}}) \cdot \exp\left(\frac{- 2\cdot (t - \frac{p}{2})\cdot x + x^2 + 2\cdot (t - \frac{p}{2})\cdot x' - x'^2}{n}\right),\nonumber
\end{align*}
where the first transition holds by \cref{app:missinProofs:hyperProbEstimation}, the second one holds by the bounds on $\size{t}$, $\size{x}$, $\size{x'}$ and $\size{p}$, and the third one holds since $\frac{1\pm y}{1\pm y} \subseteq 1 \pm 4y$ for every $y \in [0,\frac12]$.
\end{proof}

\end{document}